\documentclass[11pt]{article}
\usepackage{amssymb, amsthm, amsfonts, amsmath, amscd}
\usepackage{enumitem}
\usepackage[utf8]{inputenc}
\usepackage[margin=1in]{geometry}
\usepackage{graphicx}
\usepackage{tikz}
\usepackage{tikz-cd}
\usepackage{mathrsfs}
\usepackage{hyperref}
\usepackage{comment}
\usepackage{accents}
\usepackage{cleveref}
\usepackage{subcaption}
\usepackage{mathtools}
\usepackage{floatrow}
\usepackage{wrapfig}
\usepackage{MnSymbol}
\usepackage{ytableau}
\usetikzlibrary{decorations.markings}

\theoremstyle{definition}

\newtheorem{thm}{Theorem}[section]
\newtheorem{lem}[thm]{Lemma}
\newtheorem{df}[thm]{Definition}
\newtheorem{prop}[thm]{Proposition}
\newtheorem{cor}[thm]{Corollary}
\newtheorem{rem}[thm]{Remark}
\newtheorem{exam}[thm]{Example}

\newcommand{\C}{\mathbb{C}}

\newcommand{\R}{\mathbb{R}}
\newcommand{\Z}{\mathbb{Z}}
\newcommand{\N}{\mathbb{N}}

\newcommand{\Y}{\mathbb{Y}}
\newcommand{\la}{\lambda}
\newcommand{\SP}{\text{SP}}
\renewcommand{\sp}{\text{sp}}

\newcommand{\lavec}{\vec{\la}}
\newcommand{\muvec}{\vec{\mu}}
\renewcommand{\Pr}{\mathbb{P}}
\newcommand{\x}{\mathcal{X}}
\newcommand{\y}{\mathcal{Y}}
\renewcommand{\a}{\mathcal{A}}
\renewcommand{\b}{\mathcal{B}}

\newcommand{\veca}{\vec{\mathcal{A}}}
\newcommand{\vecb}{\vec{\mathcal{B}}}

\newcommand{\M}{\mathbb{M}}
\newcommand{\SSP}{\text{SSP}}
\renewcommand{\SS}{\text{SS}}
\newcommand{\dd}{\mathrm{d}}

\numberwithin{equation}{section}

\title{The Symplectic Schur Process}

\author{Cesar Cuenca and Matteo Mucciconi}

\date{}

\begin{document}

\maketitle

\begin{abstract}
    We define a measure on tuples of partitions, called the \emph{symplectic Schur process}, that should be regarded as the right analogue of the Schur process of Okounkov-Reshetikhin for the Cartan type~C. The weights of our measure include factors that are universal symplectic characters, as well as a novel family of ``Down-Up Schur functions" that we define and for which we prove new identities of Cauchy-Littlewood-type. Our main structural result is that the point process corresponding to the symplectic Schur process is determinantal and we find an explicit correlation kernel.
    We also present dynamics that preserve the family of symplectic Schur processes and explore an alternative sampling scheme, based on the Berele insertion algorithm, in a special case. Finally, we study the asymptotics of the Berele insertion process and find explicit formulas for the limit shape and fluctuations near the bulk and the edge. One of the limit regimes leads to a new kernel that resembles the symmetric Pearcey kernel.
\end{abstract}

\begin{center}
    \emph{To the memory of Anatoly Moiseevich Vershik, with admiration.}
\end{center}

\tableofcontents

\section{Introduction}

\subsection{Background}\label{subs:background}

In their seminal paper~\cite{OR-2003}, Okounkov and Reshetikhin introduced the Schur process, a probability measure on the set of tuples of integer partitions
$$
    \varnothing\subseteq\la^{(1)}\supseteq\mu^{(1)}\subseteq\la^{(2)}\supseteq\,\cdots\,\subseteq\la^{(k-1)}\supseteq\mu^{(k-1)}\subseteq\la^{(k)}\supseteq\varnothing,
$$
determined by
\begin{multline} \label{eq:Schur process}
    \Pr^\SP\!\left(\lavec,\muvec \,\big|\, \vec{\mathbf{x}},\vec{\mathbf{y}} \right) \propto s_{\la^{(1)}}(\mathbf{x}^1) \,s_{\la^{(1)} / \mu^{(1)}}(\mathbf{y}^1) \,s_{\la^{(2)} / \mu^{(1)}}(\mathbf{x}^2)
    \,\cdots\\
    \cdots\, s_{\la^{(k-1)} / \mu^{(k-1)}}(\mathbf{y}^{k-1}) \,s_{\la^{(k)} / \mu^{(k-1)}}(\mathbf{x}^k) \,s_{\la^{(k)}}(\mathbf{y}^k),
\end{multline}
where $s_\la$ (resp.~$s_{\la/\mu}$) are the Schur functions  (resp.~skew Schur functions) and $\mathbf{x}^1,\dots,\mathbf{x}^k$, $\mathbf{y}^1,\dots,\mathbf{y}^k$ are sets of variables. Since its introduction, this measure has been the object of constant attention due to its ubiquitous nature and applications in the asymptotic analysis of random growth processes~\cite{BF-2014}, line ensembles~\cite{CH-2014}, random tilings~\cite{Gor-2021}, asymptotic representation theory~\cite{B-2011}, Gromov-Witten invariants~\cite{OP-2006}, topological string theory~\cite{ORV-2006}, among others.
The Schur process combines the desirable properties of being amenable to very precise combinatorial and asymptotic analysis and of being a good mathematical model to describe a large variety of physical phenomena involving randomly interacting particles. In fact, the existence itself of the Schur process as an example of a completely solvable model of randomly interacting bodies has been one of the catalysts of the success of Integrable Probability; e.g.~see the survey~\cite{BG-2016}.

The Schur polynomials, which are the fundamental building block of the Schur process, have notable representation theoretic interpretations. For example, they are (i)~characters of the irreducible representations of the symmetric group (under the characteristic isometry) and (ii)~characters of the polynomial irreducible representations of general linear groups.
Another connection with representation theory comes by looking at the marginals of the Schur process, which are the Schur measures of Okounkov~\cite{O-2001}, and are based on the Cauchy identity:
\begin{equation}\label{eqn:cauchy_identity}
    \sum_\la s_\la(x_1,\dots,x_n) s_\la(y_1,\dots,y_m) = \prod_{i=1}^n \prod_{j=1}^m \frac{1}{1-x_i y_j}.
\end{equation}
The identity~\eqref{eqn:cauchy_identity} is a manifestation of the Howe duality for the action of the pair of Lie groups $(GL(n), GL(m))$ on $\mathrm{Sym}(\C^n\otimes\C^m)$.
Additionally, the Schur measures generalize the $z$-measures that arise as the solution to the problem of noncommutative harmonic analysis for the infinite symmetric group $S(\infty)$ and, in fact, this was the original motivation for their definition, see e.g.~\cite{BO-2017,Ok-2001}.
As a result of the previous observations, the Schur measure and Schur process should be regarded as objects associated to the \emph{Cartan type $A$}, since they are related to the symmetric and general linear groups.
Naturally, there exist canonical analogs of some of the above algebraic structures for other symmetry types, and some of the corresponding probability measures have already been considered in the literature, e.g.~\cite{NOS-2024} studies measures associated to skew-Howe dualities, \cite{CG-2020} considers measures arising from the branching rule of symplectic Schur polynomials with $q$-specializations, \cite{Bisi-2018,BZ-22} study relations between models of last passage percolation and characters of classical Lie groups, etc. 

 We will work here with the \emph{symplectic Schur polynomials} $\sp_\la(x_1^{\pm} ,\dots, x_n^{\pm})$, which are characters of the irreducible representations of the symplectic Lie group $Sp(2n)$ \cite{King75}. Alternatively, they can be defined by the summation identity \cite{Lit-1950,Weyl-1946}
    \begin{equation} \label{eq:intro CL identity}
        \sum_{\la :\, \ell(\la) \le n} \sp_\la(x_1^{\pm},\dots, x_n^{\pm}) s_\la(y_1, \dots, y_n) = \frac{\prod_{1\le i < j \le n} (1-y_i y_j)}{\prod_{i,j=1}^n (1-x_iy_j) (1-x_i^{-1}y_j)},
    \end{equation}
usually referred to as Cauchy-Littlewood identity, which is a consequence of the duality between irreducible finite-dimensional representations of the compact symplectic group $Sp(2n)$ and irreducible infinite-dimensional unitary highest weight representations of the non-compact Lie group $SO^*(2n)$~\cite{Howe-1989,KV-1978}. The symplectic Schur polynomials $\sp_\la(x_1^{\pm} ,\dots, x_n^{\pm})$ also possess interesting combinatorial interpretations; for instance, they are generating functions of a class of semistandard tableaux called \emph{symplectic tableaux}~\cite{King75} (see \Cref{sec:Berele} below) or of Kashiwara-Nakashima tableaux~\cite{KN-1994}.

Since all terms of the left hand side summation are nonnegative (if $x_i,y_i>0$), it is tempting to use \eqref{eq:intro CL identity} to define a probability measure on integer partitions as a symplectic variant of the Schur measure. This was done in~\cite{B-2018}, where it was found that this measure defines a determinantal point process with an explicit correlation kernel. In the same paper, Betea carried out the asymptotic analysis when the specializations in~\eqref{eq:intro CL identity} are Plancherel and showed that the result is a signed measure with the kernel $\text{Airy}_{2\to 1}$ arising in a certain limit regime.

In this paper, we define the \emph{symplectic Schur process} as a natural analog of \eqref{eq:Schur process}, having marginal distributions given by the symplectic Schur measure from~\cite{B-2018}. This construction is not canonical and is far from straightforward because of two main factors:

(i) The symplectic Schur polynomials have non-trivial branching structures, and the skew polynomials $\sp_{\la/\mu}$ depend on the number of zeros that $\la$ and $\mu$ have appended at the end. We avoid this complication by considering symmetric functions; it turns out that $\sp_{\la/\mu}$ turn into the skew Schur functions $s_{\la/\mu}$ when infinitely many variables are considered (see Equation~\eqref{eq:branching rule intro}).
We also need to define new symmetric functions $T_{\la,\mu}$ that are skew-dual to $\sp_{\la/\mu}$ (see \Cref{sec:intro_skew_dual}).

(ii) Summation identities such as \eqref{eq:intro CL identity} for skew variants of $\sp$ are not well-established in the literature with the exception of a particular variant recently found in \cite{JLW-2024}.
In fact, the understanding of that paper from the point of view of symmetric functions was the origin of the present article.

\subsection{Probabilistic interpretation: a push-pull-block dynamics}\label{subs:probabilistic interpretation}

Our symplectic Schur process describes a simple probabilistic model of interacting particles that evolves according to push, pull, and block mechanisms that we now describe. 
The space where this model lives will be the set of \emph{half-triangular interlacing arrays}, namely the set of integral arrays
\[
\big( \mathsf{x}_{k,i} \ :\ 1 \le i \le \lceil k/2 \rceil,\ k\ge 1 \big),
\]
where the coordinates $x_{k,i}\in\Z$ are subject to the following interlacing conditions:
\[
\mathsf{x}_{k+1,i+1} \le \mathsf{x}_{k,i} \le \mathsf{x}_{k+1,i}, \quad \text{for all admissible } k \text{ and } i.
\]
In the case where \( i+1 > \lceil (k+1)/2 \rceil \), we adopt the convention that \( \mathsf{x}_{k+1,i+1} = 0 \).
For any $k\ge 1$, We will refer to the sub-array
\[
\mathsf{x}_{k,\lceil k/2 \rceil},\, \mathsf{x}_{k,\lceil k/2 \rceil-1},\, \dots,\, \mathsf{x}_{k,2},\, \mathsf{x}_{k,1}
\]
as the \emph{\(k\)-th level} of the array \( (\mathsf{x}_{k,i})_{1 \le i \le \lceil k/2 \rceil} \).
The arrays described here, or variations thereof, are sometimes referred to as \emph{symplectic Gelfand-Tsetlin patterns} in the literature, e.g.~see \cite{Def-2012} and references therein.

On this space of half-triangular interlacing arrays, we define a stochastic dynamics that comprises a random number of deterministic update rules.
We first describe the deterministic updates.

\subsubsection*{Deterministic Updates}

Let \(\mathsf{x}\) denote the array before the update and call \(\mathsf{x}'\) the updated array. The propagation of moves to go from \(\mathsf{x}\) to \(\mathsf{x}'\) will be sequential from level $k$ to $k+1$ and the update changes the position of at most one particle per level. The dynamics are described via the following \emph{jump instructions}.

\medskip

\noindent

\textbf{Propagation of a Jump Instruction.}  
    Suppose that for some indices \(k\) and \(i\) we have just performed the update
    \[
    \mathsf{x}_{k,i}' = \mathsf{x}_{k,i} \pm 1.
    \]
    Then:
    \begin{itemize}
        \item \emph{Right Jumps:}  
        If \(\mathsf{x}_{k,i}' = \mathsf{x}_{k,i} + 1\), then
        \begin{itemize}
            \item if \(\mathsf{x}_{k,i} = \mathsf{x}_{k+1,i}\), particle \(\mathsf{x}_{k+1,i}\) receives a \emph{right jump instruction}: see \Cref{fig:GT dynamics}~(a);  
            \item otherwise, particle \(\mathsf{x}_{k+1,i+1}\) receives a \emph{right jump instruction}: see \Cref{fig:GT dynamics}~(b).
        \end{itemize}
        \item \emph{Left Jumps:}  
        If \(\mathsf{x}_{k,i}' = \mathsf{x}_{k,i} -1\), then
        \begin{itemize}
            \item if \(\mathsf{x}_{k,i} = \mathsf{x}_{k+1,i+1}\), particle \(\mathsf{x}_{k+1,i+1}\) receives a \emph{left jump instruction}: see \Cref{fig:GT dynamics}~(c);
            \item otherwise, particle \(\mathsf{x}_{k+1,i}\) receives a \emph{left jump instruction}: see \Cref{fig:GT dynamics}~(d).
        \end{itemize}
    \end{itemize}

\medskip

\noindent

\textbf{Execution of a Jump Instruction.}  
A jump instruction typically causes the designated particle to move by one unit in the indicated direction. The only exception occurs when a particle of the form
\[
    \mathsf{x}_{2m-1,m},
\]
i.e., the leftmost particle in an odd level, receives a right jump instruction. In this case:
\begin{itemize}
    \item If $\mathsf{x}_{2m-1,m}< \mathsf{x}_{2m,m}$, then no jump occurs for \(\mathsf{x}_{2m-1,m}\) and we set $\mathsf{x}_{2m-1,m}' = \mathsf{x}_{2m-1,m}$; moreover, the particle \(\mathsf{x}_{2m,m}\) receives a left jump instruction: see \Cref{fig:dynamics excpetional particle} (a).
    \item If $\mathsf{x}_{2m-1,m} = \mathsf{x}_{2m,m}$, then \(\mathsf{x}_{2m-1,m}\) moves one unit to the right and a right jump instruction is passed on to \(\mathsf{x}_{2m,m}\): see \Cref{fig:dynamics excpetional particle} (b).
\end{itemize}
In all other cases, if $\mathsf{x}_{k,i}$ receives a left jump instruction, we set $\mathsf{x}_{k,i}' = \mathsf{x}_{k,i} - 1$, and if $\mathsf{x}_{k,i}$ receives a right jump instruction, we set $\mathsf{x}_{k,i}' = \mathsf{x}_{k,i} + 1$.

\subsubsection*{Random Dynamics}

The discrete-time dynamics \(\mathsf{x}(1)\to\mathsf{x}(2)\to\mathsf{x}(3)\to\cdots\) combines the deterministic rules above with sequences of random numbers of jump instructions.
To describe the full random dynamics, we sample a two-index sequence \(\{ \mathsf{G}_{k,t} : k,t\in\Z_{\ge 1} \}\) of independent random variables with geometric distributions
\begin{equation}\label{eq:geo intro}
    \mathsf{G}_{k,t} \sim 
    \begin{cases}
        \operatorname{Geo}(x_{\lceil k/2 \rceil} y_t), & \text{if \( k \) is odd},\\[1mm]
        \operatorname{Geo}(x_{\lceil k/2 \rceil}^{-1} y_t), & \text{if \( k \) is even},
    \end{cases}    
\end{equation}
for suitable choices of real numbers $x_1,x_2,\dots$ and $y_1,y_2,\dots$. Our convention is that $G\sim\operatorname{Geo}(p)$ means $\text{Prob}(G=k)=(1-p)p^k$, for all $k=0,1,2,\dots$.
At each time \(t\), the update proceeds level by level.
Starting at level \(k=1\), the particle \(\mathsf{x}_{1,1}\) receives \(\mathsf{G}_{1,t}\) consecutive right jump instructions. Each instruction is executed according to the deterministic update rules described above, and the array is updated immediately after each jump. Once all \(\mathsf{G}_{1,t}\) instructions have been processed at level 1, the particle \(\mathsf{x}_{2,1}\) receives \(\mathsf{G}_{2,t}\) consecutive right jump instructions that are to be processed to completion by following the deterministic update rules, then the procedure is repeated at level \(k=3\), and so on. After all the instructions have been performed, we have completed the step $\mathsf{x}(t-1)\to\mathsf{x}(t)$.
An example of the stochastic evolution defined here is illustrated in \Cref{fig:dynamics example}.

\begin{figure}
\centering

\begin{subfigure}[b]{0.4\linewidth}
\includegraphics[width=\linewidth]{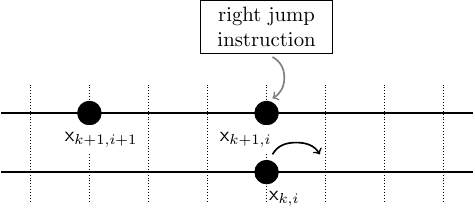}
\caption{}
\end{subfigure}
\hfill
\begin{subfigure}[b]{0.4\linewidth}
\includegraphics[width=\linewidth]{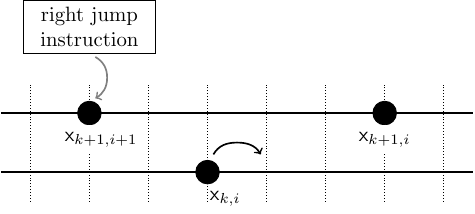}
\caption{}
\end{subfigure}

\vspace{.5cm}

\begin{subfigure}[b]{0.4\linewidth}
\includegraphics[width=\linewidth]{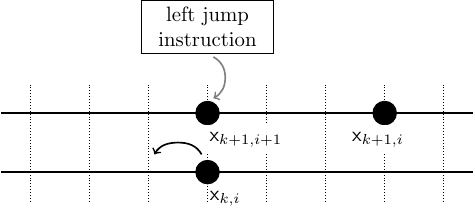}
\caption{}
\end{subfigure}
\hfill
\begin{subfigure}[b]{0.4\linewidth}
\includegraphics[width=\linewidth]{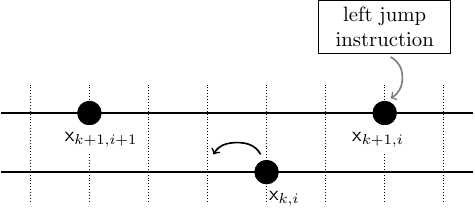}
\caption{}
\end{subfigure}

\caption{Propagation of jump instructions to particles at level $k+1$ from the update of a particle at level $k$.}
\label{fig:GT dynamics}
\end{figure}

\begin{figure}
\centering

\begin{subfigure}[b]{0.4\linewidth}
\includegraphics[width=\linewidth]{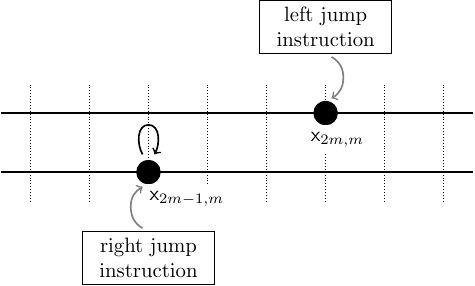}
\caption{}
\end{subfigure}
\hfill
\begin{subfigure}[b]{0.4\linewidth}
\includegraphics[width=\linewidth]{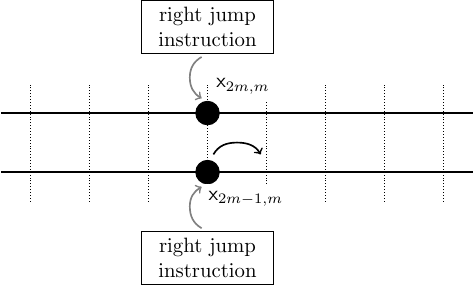}
\caption{}
\end{subfigure}

\caption{Propagation of jump instructions to particles at level $2m$ from the update of leftmost particle at level $2m-1$.}
\label{fig:dynamics excpetional particle}
\end{figure}

\begin{figure}
\centering

\begin{subfigure}[b]{0.3\linewidth}
\includegraphics[width=\linewidth]{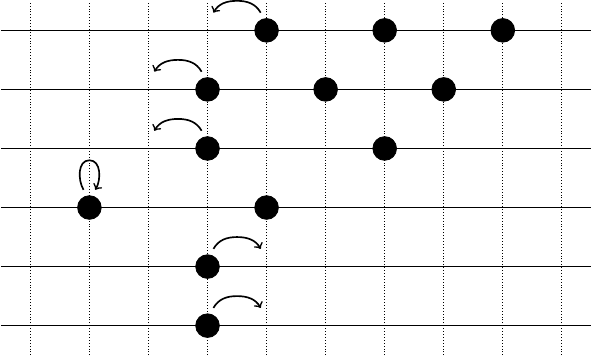}
\end{subfigure}
\hfill
\begin{subfigure}[b]{0.3\linewidth}
\includegraphics[width=\linewidth]{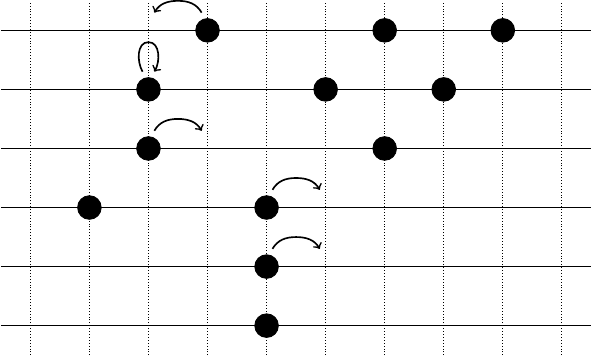}
\end{subfigure}
\hfill
\begin{subfigure}[b]{0.3\linewidth}
\includegraphics[width=\linewidth]{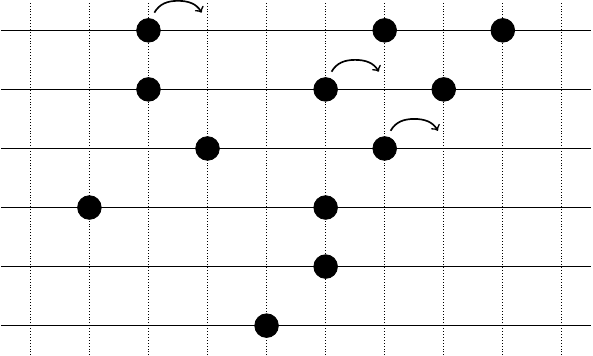}
\end{subfigure}

\caption{The random update of a half-triangular array $\mathsf{x}(t)$. Here $\mathsf{G}_{1,t}=1$, $\mathsf{G}_{2,t}=1$, $\mathsf{G}_{3,t}=0$, $\mathsf{G}_{4,t}=1$. On the left panel the jump instruction given to particle $\mathsf{x}_{1,1}$ triggers the instantaneous updates at level above. Similarly, in the central panel, the initial jump instruction if given to particle $\mathsf{x}_{2,1}$. In the right panel, corresponding to the draw $\mathsf{G}_{4,t}=1$, the particle $\mathsf{x}_{4,1}$ receives a right jump instruction.}
\label{fig:dynamics example}
\end{figure}

\medskip

It turns out that the stochastic dynamics just defined can be described through the formalism of the symplectic Schur process which we introduce and study in this paper. More precisely, for any $n\in\Z_{\ge 1}$, the multi-time process 
\begin{equation} \label{eq:push block}
\big\{ (t,\mathsf{x}_{2n,i}(t)-i) \ :\ i,t=1,\dots,n \big\},
\end{equation}
under the initial conditions
\begin{equation}
    \mathsf{x}_{k,i}(0)=0, \qquad \text{for all } k,i,
\end{equation}
is described by a particular case of the symplectic Schur process \eqref{eq:ssp intro}. A large sample of the point process \eqref{eq:push block} is depicted in \Cref{fig: intro}: there, the parameters that enter the geometric random variables \eqref{eq:geo intro} have been set to $x_j=1$, $y_j=1/2$, for all $j\ge 1$.

\medskip

Similar particle dynamics on interlacing arrays have been constructed and studied in the literature (see, e.g., \cite{BK-2010,WW-2009,Def-2012,Nte-2016,BF-2008}), motivated by dynamics with local hard-edged interactions (exclusion processes, random tilings, etc.) that appear as marginals.
In \Cref{subs:Symplectic Schur process through Berele sampling} and, more extensively, in \Cref{sec:asymptotics}, we show how the connection with the symplectic Schur process enables a precise asymptotic analysis of the multi-time process~\eqref{eq:push block}.
From the models on half-triangular interlacing arrays, with the exception of~\cite{BK-2010}, no other work has carried out such asymptotic analysis.
We note that all these models are said to be \emph{constrained by a wall}, including ours, because always $x_{k,i}(t)\ge 0$, due to the definition of the dynamics.
The paper~\cite{BK-2010} is concerned with a process where particles are reflected by a wall, while in ours, particles are blocked by an impenetrable wall; moreover, their methods are completely different from ours.

\subsection{Main results}

\subsubsection{Down-up Schur functions and skew Cauchy-Littlewood identities}\label{sec:intro_skew_dual}

The first part of our work is on the theory of symmetric functions and lays the combinatorial foundation that will be employed to study the symplectic Schur process later.
Recall that the \emph{universal symplectic characters} $\SP_\la(x)$, where $\la$ ranges over all partitions, were defined in~\cite{KT-1987} as symplectic analogues of the Schur functions.
When $x$ is the finite collection of variables $(x_1,x_1^{-1}, \dots ,x_n,x_n^{-1})$ and $\ell(\la)\le n$, this symmetric function reduces to the symplectic Schur polynomial
\begin{equation*}
    \SP_\la(x_1,x_1^{-1}, \dots ,x_n,x_n^{-1}) = \sp_\la(x_1^\pm, \dots ,x_n^\pm), \quad\text{for } \ell(\la) \le n,
\end{equation*}
and thus carries the information about the characters of irreducible representations of $Sp(2n)$.

    As in the case of Schur processes~\eqref{eq:Schur process}, one would expect some kind of ``skew universal symplectic characters''  to exist and to play a role in any sensible definition of symplectic Schur process. However, a special case of \cite[Thm.~7.22]{Rai-2005} (see~\Cref{branching_SP_thm} below) gives
    \begin{equation} \label{eq:branching rule intro}
            \SP_\la(x,y)= \sum_{\mu} \SP_\mu(x) s_{\la/\mu}(y).
    \end{equation}
    Naively, this seems to indicate (by comparing with the skew Schur functions defined in \cite[Ch.~I, Equation~(5.9)]{M-1995}) that the reasonable definition of ``skew universal symplectic characters'' should be exactly the same as the skew Schur functions.
    However, by using only skew Schur functions, we obtain nothing other than the classical Schur process, as in~\eqref{eq:Schur process}.
    Our new key insight is that, in order to define the symplectic Schur process, it is necessary to introduce and develop the combinatorial properties of a novel family of symmetric functions $T_{\lambda,\mu}$ that are generalizations of Schur polynomials $s_\lambda$.
    
    Define the symmetric functions $T_{\la,\mu}$, parametrized by two partitions $\la,\mu$, which will be called the \emph{down-up Schur functions}, by the formula
    \begin{equation} \label{eq:T intro}
        T_{\la,\mu}(y) := \sum_\nu d_{\mu,\nu}^\la s_\nu(y),
    \end{equation}
    where $d_{\mu,\nu}^\la$ are the Newell-Littlewood (NL) coefficients~\cite{Lit-1958, New-1951}. These coefficients are symplectic versions of the Littlewood-Richardson coefficients, as they compute the multiplicity of an irreducible representation in the decomposition of tensor products of irreducible representations of the symplectic group; see~\cite{KT-1987,GOY-2021}.
    In \Cref{sec:du_schur}, we prove the following new combinatorial formulas for the symmetric functions $T_{\la,\mu}(y)$.

    \begin{thm}\label{thm:main_1}
        The down-up Schur functions $T_{\la,\mu}(y)$, defined by~\eqref{eq:T intro}, satisfy identities similar to those satisfied by the skew Schur functions, including
        \begin{itemize}
            \item the branching rule: see \Cref{branching_star};
            \item the skew down-up Cauchy identity: see \Cref{new_skew_cauchy_thm};
            \item the Cauchy-Littlewood identity: see \Cref{gen_cauchy_littlewood_prop_2};
        \end{itemize}
        as well as the following formula, stated in \Cref{combinatorial_S_star}:
        \begin{equation} \label{eq:down up intro}
        T_{\la,\mu}(y) = \sum_{\kappa} s_{\la/\kappa}(y)  s_{\mu/\kappa}(y).
    \end{equation}
    \end{thm}

    The formula~\eqref{eq:down up intro} is the reason for the qualifier \emph{down-up}, since $s_{\la/\kappa}(y)s_{\mu/\kappa}(y)$ vanishes, unless $\la\supseteq\kappa\subseteq\mu$.
    The formulas above will be instrumental for our probabilistic applications.
    
    The main motivation to introduce the down-up Schur functions is that they are \emph{skew-duals} to the universal symplectic characters; this means that they satisfy the variant of the Cauchy-Littlewood identity in \Cref{gen_cauchy_littlewood_prop_2}.
    Recently, Jing-Li-Wang~\cite{JLW-2024} constructed skew-dual variants of symplectic Schur polynomials by using the formalism of vertex operators. This construction expresses $T_{\la,\mu}(y_1,\dots,y_n)$, denoted there by $S^{*}_{\la/\mu}(y_1,\dots,y_n)$, in terms of a Weyl-type determinant with the result being a priori a Laurent polynomial with no precise combinatorial description. We believe that our symmetric function definition~\eqref{eq:T intro} provides more insights into the combinatorial and algebraic nature of the symmetric functions $T_{\la,\mu}$.

\subsubsection{The symplectic Schur process}

Having discussed the necessary notions from symmetric functions, we can introduce the main object of this paper. We define the \emph{symplectic Schur process} as the measure on (pairs of) sequences of partitions $\lavec = (\la^{(1)},\dots,\la^{(k)})$, $\muvec= (\mu^{(1)},\dots,\mu^{(k-1)})$, given by
    \begin{equation} \label{eq:ssp intro}
        \Pr^\SSP\!\left(\lavec,\muvec \,\big|\, \veca,\vecb \right) := \frac{1}{Z(\vec{\a},\vec{\b})}\cdot s_{\la^{(1)}}(\b^1)\cdot \prod_{j=1}^{k-1}{\left[ s_{\la^{(j)}/\mu^{(j)}}(\a^j)\, T_{\mu^{(j)},\la^{(j+1)}}(\b^{j+1}) \right]}\cdot\SP_{\la^{(k)}}(\a^k),
    \end{equation}
where $\a^1, \dots \a^k, \b^1,\dots, \b^k$ are specializations of the algebra of symmetric functions (see \Cref{specs_sym}). Moreover, $Z(\veca,\vecb)$ is the normalization constant that makes the total mass of the measure equal to $1$.
A closed formula for $Z( \vec{\a},\vec{\b})$ can be computed using our first main \Cref{thm:main_1}; see \Cref{partition_function_1} for details.
The reader should compare our definintion with the equation~\eqref{eq:Schur process} that defines the Schur process.
Note that the support of the symplectic Schur process consists of tuples $(\lavec,\muvec)$ such that $\mu^{(j)}\subseteq\la^{(j)}$, for all $1\le j\le k-1$, but unlike the Schur process, the containments $\mu^{(j)}\subseteq\la^{(j+1)}$ are not necessarily true; see \Cref{fig:support} for a representation of the support.

\begin{figure}[h]
    \centering
    \includegraphics{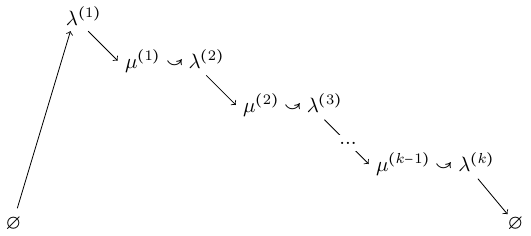}
    \caption{Representation of the support of the symplectic Schur process.}
    \label{fig:support}
\end{figure}

As a (possibly signed) measure, we prove that the symplectic Schur process is a determinantal point process. To state this result precisely, consider for each pair $(\lavec,\muvec)$ the point configuration $\mathcal{L}(\lavec,\muvec)$ on the space $\{1,\frac{3}{2},2,\dots,k-\frac{1}{2},k\}\times\Z$,\footnote{In the body of the paper, $\mathcal{L}(\lavec,\muvec)$ lives in the different space $\{1,1',2,2',\dots,(k-1)',k\}$. This is not a crucial difference, as the important thing is that our space consists of $(2k-1)$ copies of $\Z$, to describe the $(2k-1)$ partitions in $\lavec,\muvec$.} given by
    \begin{multline*}
        \mathcal{L}(\lavec,\muvec) := 
        \left\{ \left( 1,\, \la^{(1)}_i - i \right) \right\}_{i\ge 1}\cup
        \left\{ \left( \frac{3}{2},\, \mu^{(1)}_i - i \right) \right\}_{i\ge 1}\cup
        \left\{ \left( 2,\, \la^{(2)}_i - i \right) \right\}_{i\ge 1}\cup\ \cdots\\
        \cdots\ \cup \left\{ \left( k-1,\, \la^{(k-1)}_i - i \right) \right\}_{i\ge 1}
        \cup\left\{ \left( k-\frac{1}{2},\, \mu^{(k-1)}_i - i \right) \right\}_{i\ge 1}\cup
        \left\{ \left( k,\, \la^{(k)}_i - i \right) \right\}_{i\ge 1}.
    \end{multline*}

\begin{thm}[See \Cref{thm:det} in the text]\label{eq:main theorem intro}
    For any $S\ge 1$ distinct points $(i_1,u_1),\dots,(i_S,u_S)$ of $\{1,\frac{3}{2},2,\dots,k-1,k-\frac{1}{2},k\}\times\Z$, we have the following identity
    \begin{equation*}
        \sum_{\{(i_1,u_1),\dots,(i_S,u_S)\}\subset\mathcal{L}(\lavec,\muvec)}{ \Pr^\SSP\!\left(\lavec,\muvec \,\big|\, \veca,\vecb \right) }
        = \det_{1\le s,t\le S} \left[ K^{\SSP}(i_s,u_s;i_t,u_t) \right],
    \end{equation*}
    where the explicit expression for the kernel $K^\SSP$ is given in~\eqref{final_kernel} as a double contour integral.
\end{thm}

\begin{rem}
    In analogy with the Schur process, the marginals of the symplectic Schur process are given by the symplectic Schur measure. Moreover, a determinantal correlation kernel for the symplectic Schur measure was found by Betea in \cite{B-2018} and corresponds exactly to $K^{\textrm{Betea}}(u,v)=K^{\SSP}(i,u;i,v)$, for $i$ fixed.
\end{rem}

\begin{wrapfigure}[43]{R}{4cm}
\vspace{-.5cm}
  \begin{center}
    \begin{tikzpicture}
        \node[] at (0,0) {\includegraphics[width=3cm]{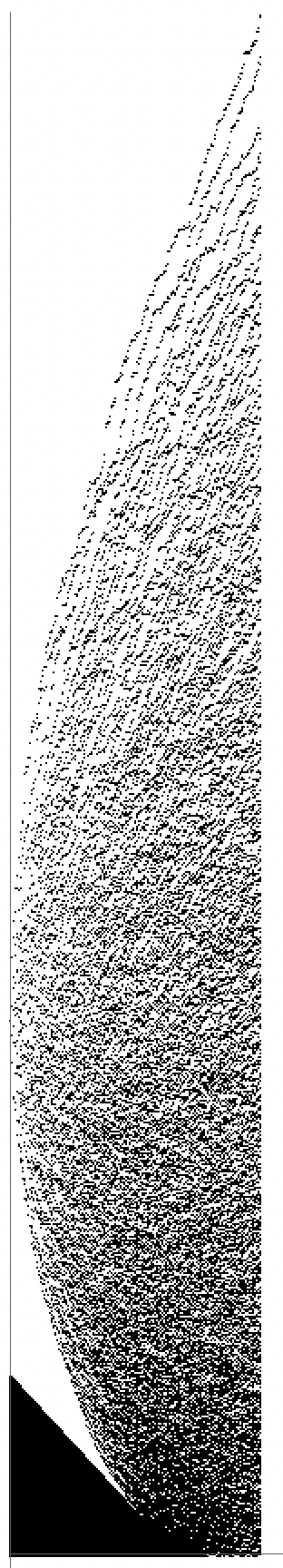}};
        \draw[xshift=-1.5cm,yshift=-6.25cm,->] (0,0) node[left] {\footnotesize $(0,0)$} -- (3,0) node[below] {\footnotesize $j$};
        \draw[xshift=-1.4cm,yshift=-6.25cm,->] (0,-2) -- (0,14.5) node[left] {\footnotesize $\lambda^{(j)}_i-i$};
    \end{tikzpicture}
  \end{center}
  \caption{\footnotesize A depiction of the point process $\mathcal{L}(\lavec)$ constructed through Berele insertion.}
  \label{fig: intro}
\end{wrapfigure}

For general specializations $\a^i, \b^j$, the symplectic Schur process is a signed measure.
It is an important question to classify all specializations for which $\Pr^\SSP$ is a probability measure, i.e.~for which the product in the right hand side of \eqref{eq:ssp intro} is nonnegative for every choice of sequences of partitions $\lavec,\muvec$.
In \Cref{sec:special_cases}, we find large families of examples of specializations for which this is the case.
We note that this contribution is more difficult than in the case of Schur processes.
This boils down to the fact that the classical problem of classification of \emph{Schur-positive specializations}, i.e. specializations $\a$ such that
    \begin{equation*}
        s_{\la/\mu}(\a) \ge 0,
        \qquad
        \text{for all }\la,\mu\in\Y,
    \end{equation*}
has been solved in 1952 by~\cite{Whi-1952,ASW-1952,Edr-1952}. On the other hand, for the universal symplectic character $\SP_\la$, the existence of positive specializations is not established and, in fact, it is not even clear they exist at all. For instance, in the case when $\a=(x_1, x_1^{-1},\dots, x_n, x_n^{-1})$ with $x_1,\dots,x_n >0$, we have
    \begin{equation*}
        \SP_\la(x_1, x_1^{-1},\dots, x_n, x_n^{-1}) \ge 0,
        \quad
        \text{whenever } \ell(\la)\le n.
    \end{equation*}
However, if $\ell(\la)>n$, then it is possible that $\SP_\la(x_1, x_1^{-1},\dots, x_n, x_n^{-1}) < 0$, even if $x_1,x_2>0$. For instance, if $n=1$ and $\la=(1,1,1)$, we have $\SP_{(1,1,1)}(x_1, x_1^{-1}) =-2(x_1+x_1^{-1})$.

\subsubsection{Symplectic Schur process through Berele sampling and asymptotics} \label{subs:Symplectic Schur process through Berele sampling}

In \Cref{sec:Berele}, we study a dynamics that preserves a special version of symplectic Schur processes~\eqref{eq:ssp intro} and analyze the asymptotics of its fixed-time distributions.
This dynamics can be alternatively sampled by a process on integer partitions $\la^{(1)},\la^{(2)},\la^{(3)},\dots$ that adds or removes boxes from the partitions at each step, and that is  equivalent to the particle dynamics described in \Cref{subs:probabilistic interpretation} under the identification $\lambda^{(t)}_i = \mathsf{x}_{2n,i}(t)$.

Consider the symplectic Schur process with the following choice of specializations
    \begin{center}
        $\a^1 = \dots = \a^{k-1} = \varnothing,$
            \quad
            $\a^k= (x_1, x_1^{-1},\dots, x_n, x_n^{-1}),$
            \quad
            $\b^1=(y_1),$
            \dots,
            $\b^k=(y_k),$
    \end{center}
where $k\le n$ and $x_1,\dots,x_n,y_1,\dots,y_k$ are positive real numbers satisfying $ 0< x_iy_j, x_i^{-1}y_j <1$, for all $1\le i \le n$, $1\le j \le k$. In this case, we see the symplectic Schur process as a probability measure on $k$-tuples of partitions $\lavec=(\la^{(1)},\dots,\la^{(k)})$ -- as we trace over partitions $\mu^{(1)},\dots,\mu^{(k-1)}$ -- that admits an explicit sampling procedure through the Berele insertion algorithm~\cite{B-1986}, a symplectic version of the Robinson-Schensted-Knuth algorithm; see \Cref{sec:Berele}.

To construct a random sequence $\lavec = (\la^{(1)}, \dots, \la^{(k)})$, one first needs to prepare $k$ random words $w_1,\dots,w_k$ in the alphabet $\mathsf{A}_n=\{1,\overline{1}, \dots, n, \overline{n}\}$ by setting
    \begin{equation*}
        w_j = \underbrace{1\,\cdots\, 1}_{m_{1,j}} \underbrace{\overline{1}\,\cdots\, \overline{1}}_{\overline{m}_{1,j}} \,\cdots\,\underbrace{n\,\cdots\, n}_{m_{n,j}} \underbrace{\overline{n}\,\cdots\, \overline{n}}_{\overline{m}_{n,j}},
    \end{equation*}
where $m_{i,j},\overline{m}_{i,j}$, counting the number of occurrences of $i,\overline{i}$ in $w_j$, are sampled independently with the geometric laws $\mathrm{Geo}(x_iy_j), \mathrm{Geo}(x_i^{-1}y_j)$, respectively. Then the Berele insertion scheme offers a procedure to insert the words $w_j$'s sequentially into an initially empty symplectic Young tableau and obtain the sequence
    \begin{equation*}
        \varnothing \xrightarrow[]{w_1} P_1 \xrightarrow[]{w_2} P_2 \xrightarrow[]{w_3} \,\cdots\, \xrightarrow[]{w_k} P_k,
    \end{equation*}
whose shapes 
    \begin{equation*}
        \la^{(1)}, \dots, \la^{(k)},
    \end{equation*}
namely $\la^{(j)} = \mathrm{shape}(P_j)$, follow the law of a symplectic Schur process.
Note that the dynamics described in \Cref{subs:probabilistic interpretation} is recovered simply by interpreting the symplectic tableaux $P_1,\dots,P_k$ as half-triangular interlacing arrays.
This result is contained in \Cref{prop:symplectic Schur berele}. In \Cref{fig: intro}, we report a sample of the point process $\mathcal{L}(\lavec) := \{ (j,\,\la^{(j)}_i-i) \mid i\ge 1,\, j=1,\dots,k\}$, constructed through our Berele insertion process with parameters
\begin{equation}\label{eq:choice parameters intro}
    x_1=\cdots=x_n=1, 
    \qquad
    y_1= \cdots = y_k =1/2.
\end{equation}

As an application of our second main \Cref{eq:main theorem intro}, we also perform the asymptotic analysis for the special choice of parameters~\eqref{eq:choice parameters intro}, and obtain the following results, which are elaborated on in \Cref{sec:asymptotics}.

\begin{thm}[Limit shape and limiting density] \label{thm:into limit shape}
    We have
    \begin{equation}
        \lim_{n\to\infty} \mathbb{P}\Big( \big(\lfloor nx \rfloor,\, \lfloor ny \rfloor\big) \in \mathcal{L}(\lavec) \Big) = 
        \begin{cases}
            1, \qquad & \text{if } 0\le x< 1, \ -1\le y \le -x,
            \\
            \rho(x,y), \qquad &\text{if } 0< x< 1, \ y_-(x) < y < y_+(x),
            \\
            0, \qquad & \text{otherwise},
        \end{cases} 
    \end{equation}
where the functions $\rho(x,y)$ and $y_\pm(x)$ are defined in \eqref{eq:limiting density} and \eqref{eq:limit shape}, respectively.
\end{thm}

By studying the scaling limits of the correlation kernel $K^{\mathrm{SSP}}$, we can describe fluctuations of the point process $\mathcal{L}(\lavec)$ around the limit shape described in \Cref{thm:into limit shape}. An example of such statements is the following.

\begin{thm}[Limiting fluctuations in the bulk]
Fix arbitrary $x\in (0,1)$ and $y\in (y_-(x), y_+(x))$, the process $\mathcal{L}(\lavec)-\big(\lfloor n x \rfloor,\, \lfloor n y \rfloor\big)$ converges in distribution to the extended Sine process with the determinantal correlation kernel given in \eqref{eq:extended_Sine_kernel}.
\end{thm}

Similarly, we can characterize the limiting fluctuations of the process $\mathcal{L}(\lavec)$, when $y=y_-(x)$ or $y=y_+(x)$. After proper rescaling, we find that for $x=y=0$ and $x=-y=4/5$, the process $\mathcal{L}(\lavec)$ converges to a GUE corner process; see \Cref{subs:GUE corner processes}. In addition, for $x\in (0,1)$ and $y=y_\pm(x)$ (with $(x,y)\neq(4/5,-4/5)$), we prove convergence to the Airy point process; see \Cref{prop:limit Airy} for the precise statement.

Finally, we uncover a new limiting object that, to our knowledge, has not appeared in previous analyses of Schur processes. In particular, we show that near the microscopic region where the limit shape meets the horizontal line \(y=-1\), the kernel \(K^{\SSP}\) from \Cref{eq:main theorem intro} converges to the novel variant of the Pearcey kernel in the following result (\Cref{prop:pearcey} in the text).

\begin{thm}[Variant of the Pearcey Kernel]\label{thm:main_3}
Consider the rescaling
\[
i(\tau)=\left\lfloor\frac{9n}{8} +\sqrt{\frac{27n}{64}}\,\tau\right\rfloor,\qquad
u(\alpha)= -n + \left\lfloor\left( \frac{n}{12} \right)^{\frac{1}{4}} \alpha\right\rfloor.
\]
For any fixed \(k\in\Z_{\ge 1}\) and \(\tau_1,\dots,\tau_k\in\R\) and \(\alpha_1,\dots,\alpha_k\in\R_{+}\), let \(i_\ell=i(\tau_\ell)\), \(u_\ell=u(\alpha_\ell)\), then
\[
\lim_{n\to\infty} \,
\det_{1\le\ell,\ell'\le k} \left[ \left( \frac{n}{12} \right)^{\frac{1}{4}} \Big(  \delta_{i_\ell,i_{\ell'}} \delta_{u_\ell,u_{\ell'}} - K^{\SSP}(i_\ell,u_\ell;i_{\ell'},u_{\ell'}) \Big) \right] = \det_{1\le\ell,\ell'\le k} \left[ \mathcal{K}(\tau_\ell, \alpha_\ell; \tau_{\ell'},  \alpha_{\ell'}) \right],
\]
where \(\mathcal{K}\) is the Pearcey-like kernel defined in \eqref{eq:symplectic Percey}.
\end{thm}
In \Cref{sec:asymptotics at 1.125}, we carry out the asymptotic analysis of the correlation kernel \(K^{\SSP}\) in the regime where \(k>n\). In this range, the Berele insertion process deviates from the symplectic Schur process and the measure \(\mathbb{P}^\SSP\) is, in general, no longer a probability measure. As a result \Cref{thm:main_3} does not allow us to conclude that the kernel \(\mathcal{K}\) captures the behavior at the intersection of the liquid region with the horizontal line \([0,+\infty)\times\{-n\}\). Nevertheless, the Pearcey-like kernel \(\mathcal{K}\) appears to govern a determinantal point process, as suggested by numerical tests and its structure closely resembles that of the symmetric Pearcey kernel discussed in~\cite{BK-2010}. Similar kernels have recently emerged in the study of probability measures linked to the skew \((GL(n), GL(k))\)-duality with piecewise-constant and alternating specializations~\cite{BNNS-2024+}. A detailed comparison between these kernels and the one obtained here presents an interesting direction for future work.

\subsection{Outlook and future research directions}

Since its introduction, numerous generalizations of the Schur process have been proposed. Some notable ones are: the \emph{periodic Schur process}~\cite{B-2007}, which can be regarded as a Schur process of Cartan type $\hat{A}$~\cite{Tin-2007}; the \emph{Macdonald process}~\cite{BC-2014}, which employs Macdonald functions in place of Schur functions; the \emph{Spin Hall-Littlewood processes}~\cite{BP-2019,BMP-2021,AB-2024} where Schur functions are replaced by symmetric functions originating from partition functions of integrable vertex models; among others.
 
The list presented above is far from complete.
We note that these generalizations also have applications that are unreachable by the Schur process, e.g.~to the analysis of $\beta$-Jacobi corners process~\cite{BG-2015} and directed polymers~\cite{BC-2014}.
In general, the presence of branching rules and Cauchy-type identities allows to construct \emph{solvable} Schur-type processes; other recent examples are~\cite{OCSZ-14,MP-22,Kos-2021,CGK-2022,A-2023,GP-2024}.
Just like the Schur process, we envision the symplectic variant defined here to be amenable to various generalizations and further applications. A natural one is to lift universal symplectic characters to Macdonald-Koornwinder functions~\cite{Rai-2005}.
Some initial steps in this direction were taken by~\cite{Nte-2016} for $q$-Whittaker polynomials.
On the other hand, \cite{WZJ-2016} studied integrable vertex models to prove identities for BC-type Hall-Littlewood polynomials. More recently, stochastic vertex models of symplectic type were defined in \cite{Zho-2022}, but their connection to random partitions is unclear.
 
Another interesting direction includes testing alternative sampling schemes for the symplectic Schur process. 
Contrary to the case of combinatorics in type $A$, which is largely encompassed by the celebrated RSK correspondence, combinatorics in type $C$ is richer and less straightforward. For instance, the Berele insertion recalled below in \Cref{sec:Berele} is not the canonical extension of the RSK correspondence, and other combinatorial correspondences can be put in place to prove the Cauchy-Littlewood identity (see~\cite{HK-2022}); they have other desirable properties, such as compatibility with crystal symmetries~\cite{Bak-2000,Lec-2002}. It would be interesting to investigate sampling schemes based on these other combinatorial correspondences, which would possibly prove fruitful to understand if extensions in affine setting, in the style of~\cite{B-2007} or connections with solitonic systems~\cite{IMS-2023} are possible.

Finally, we mention that the dynamics preserving the symplectic Schur process from \Cref{sec:dynamics ssp} can lead to various surface growth models; for the construction of similar dynamics that preserve the Schur and Macdonald processes, see~\cite{BP-2016,BF-2014} and the survey~\cite{BP-2014}.
It would be interesting to explore various general schemes for random growth of sequences of partitions and stepped surfaces.
Related Markov processes that preserve the BC-type $z$-measures (see \cite{C-2018a} and \cite[Sec.~8]{BO-2005}) from asymptotic representation theory were constructed in~\cite{C-2018b}; see also \cite{A-2018}.

\subsection*{Acknowledgments}
This research was partially funded by the European Union’s Horizon 2020 research and innovation programme under the Marie Skłodowska-Curie grant agreement No.~101030938.
C.C. is grateful to Alexei Borodin, Grigori Olshanski and Anton Nazarov for helpful conversations, and to the the University of Warwick for hosting him during the beginning stages of this project.

\section{Background on symmetric polynomials and Schur functions}

We will assume that the reader is familiar with the language of symmetric functions, e.g. from~\cite[Ch.~I]{M-1995}. In this section, we simply recall some definitions and set our notations.

Partitions will be denoted by lowercase Greek letters $\la,\mu,\rho$, etc., and they will be identified with their corresponding Young diagrams.
The set of all partitions is denoted by $\Y$. For any $\la\in\Y$, the size of $\la$ is denoted by $|\la|$ and set $\Y_n:=\{\la\in\Y:|\la|=n\}$, so that the set of all partitions is $\Y=\bigsqcup_{n\ge 0}{\Y_n}$ and includes the empty partition $\varnothing\in\Y_0$.
Finally, $\ell(\la)$ denotes the length of $\la$.

\subsection{Schur polynomials and their symplectic analogues}

\begin{df} \label{def:schur poly}
Let $n\in\Z_{\ge 0}$ and let $\la$ be any partition with $\ell(\la)\le n$.
	\begin{itemize}
		\item For variables $y_1,\dots,y_n$, let $y:=(y_1,\dots,y_n)$ and define
\begin{equation*}
s_\la(y) := \frac{\det_{1\le i,j\le n}\left(y_i^{\la_j+n-j}\right)}{\prod_{1\le i<j\le n}(y_i-y_j)}.
\end{equation*}
We call $s_\la(y)$ a \emph{Schur polynomial}.

		\item For variables $x_1,\dots,x_n$, let $x^\pm:=(x_1^\pm,\dots,x_n^\pm)=(x_1,x_1^{-1},\dots,x_n,x_n^{-1})$ and define
\begin{equation*}
\sp_\la(x^\pm) := \frac{\det_{1\le i,j\le n}\left(x_i^{\la_j+n-j+1} - x_i^{-(\la_j+n-j+1)}\right)}
{\prod_{1\le i<j\le n}(x_i-x_j)(1-x_i^{-1}x_j^{-1})\prod_{i=1}^n(x_i - x_i^{-1})}.
\end{equation*}
We call $\sp_\la(x^\pm)$ a \emph{symplectic Schur polynomial}.
	\end{itemize}
\end{df}

Clearly, $s_\la(y)$ is a symmetric polynomial in the variables $y_1,\dots,y_n$, homogeneous of degree $|\la|$.
The Schur polynomials are character values of the irreducible polynomial representations of $GL(n,\C)$ evaluated at the maximal torus $\mathbb{T}^n := \{ \text{diag}(y_1,\dots,y_n): y_i\in\C^\times\}$, see e.g.~\cite{Weyl-1946}.

By contrast, $\sp_\la(x^\pm)$ are symmetric Laurent polynomials, i.e.~polynomials in $x_1^\pm,\dots,x_n^\pm$, symmetric with respect to the natural action of the hyperoctahedral group $S_n\ltimes(\Z/2\Z)^n$.
The symplectic Schur polynomial $\sp_\la(x^\pm)$ can also be expressed as a symmetric polynomial in the variables $x_1+x_1^{-1},\dots,x_n+x_n^{-1}$, of degree $|\la|$, but inhomogeneous. They are also character values of the irreducible representations of $Sp(2n)$ evaluated at a maximal torus, see e.g.~\cite{King75}.

The following identity will be useful; for proofs, see e.g.~\cite{K-1989,S-1990}.

\begin{prop}[Cauchy-Littlewood identity for symplectic Schur polynomials]\label{cauchy_littlewood_prop}
Let $n\ge k$ be two nonnegative integers, and let $x^\pm=(x_1^\pm,\dots,x_n^\pm),\, y=(y_1,\dots, y_k)$. Then
\begin{equation}\label{cauchy_littlewood_eqn}
\sum_{\mu\,:\,\ell(\mu)\le k}{\sp_{\mu}(x^\pm) s_\mu(y)} =
\frac{\prod_{1\le i<j\le k}(1 - y_iy_j)}{\prod_{i=1}^n \prod_{j=1}^k{(1-x_iy_j)(1-x_i^{-1}y_j)}}.
\end{equation}
\end{prop}

The identity~\eqref{cauchy_littlewood_eqn} can be interpreted either as an equality of power series in $\C[[x_1+x_1^{-1},\dots, x_n+x_n^{-1},y_1,\dots,y_n]]$, or as a numerical equality when $|y_j|<\min\{|x_i|, |x_i|^{-1}\}$, for all $1\le i\le n$, $1\le j\le k$ (so that the left hand side is absolutely convergent).

\subsection{The algebra of symmetric functions}

We denote the graded real algebra of symmetric functions by $\Lambda=\bigoplus_{n\ge 0}{\Lambda_n}$, where $\Lambda_n$ consists of the homogeneous symmetric functions of degree $n$, and $\Lambda_0=\R$.
The power sums $p_1,p_2,\dots$ are algebraically independent generators of $\Lambda$ with $\deg(p_n)=n$, for all $n\ge 1$, so that $\Lambda=\R[p_1,p_2,\dots]$.

Symmetric functions can be regarded as functions on an infinite set of variables $y_1,y_2,\dots$ that are invariant with respect to any finitary permutation of variables, e.g.~setting $p_k=y_1^k+y_2^k+\cdots$.
Indeed, let $\C[y_1,\dots,y_n]^{S_n}$ be the space of symmetric polynomials in the $n$ variables $y_1,\dots,y_n$ and let $\pi^{n+1}_n:\C[y_1,\dots,y_{n+1}]^{S_{n+1}}\to\C[y_1,\dots,y_n]^{S_n}$ be the map $y_{n+1}\mapsto 0$.
Then $\Lambda$ is identified as the projective limit of the chain $\left( \C[y_1,\dots,y_n]^{S_n},\pi^{n+1}_n \right)$; in this case we also write $\Lambda$ as $\Lambda_y$.

The \emph{complete homogeneous symmetric functions} $h_1,h_2,\dots$ are defined by the generating series
\begin{equation*}
1+\sum_{k=1}^\infty{h_kz^k} = \exp\left( \sum_{n=1}^\infty{\frac{p_nz^n}{n}} \right) = \frac{1}{\prod_{i\ge 1}(1 - y_iz)},
\end{equation*}
and furnish a set of algebraically independent generators of $\Lambda=\Lambda_y$.
As it is customary, we also define $h_0:=1$ and $h_m:=0$, for all $m<0$.

\subsection{Schur symmetric functions}

The Schur polynomials obey the stability relations $s_\la(y_1,\dots,y_n,0)=s_\la(y_1,\dots,y_n)$, so there exists a unique symmetric function that specializes to Schur polynomials when all but finitely many variables vanish.
In other words, if $\pi^\infty_n:\Lambda\to\C[y_1,\dots,y_n]^{S_n}$ are the natural projections, the following definition makes sense.

\begin{df}
For any $\la\in\Y$, there exists a unique $s_\la\in\Lambda$, homogeneous of degree $|\la|$, such that
\begin{equation*}
\pi^\infty_n(s_\la) = s_\la(y_1,\dots,y_n),\text{ for all }n\ge 1.
\end{equation*}
We call $s_\la$ a \emph{Schur symmetric function} or just a \emph{Schur function}.
\end{df}

Note that $\pi^\infty_n(s_\la)=0$, whenever $n<\ell(\la)$.

Let us recall some facts about Schur functions from our main reference on symmetric functions~\cite[Ch.~I]{M-1995}.
The set $\{s_\la\}_{\la\in\Y}$ is a homogeneous basis of $\Lambda$, and each Schur function can be expressed in terms of the $h_\ell$'s by means of the following \emph{Jacobi-Trudi formula}
\begin{equation}\label{eqn:jacobi_trudi}
s_\la = \det_{1\le i,j\le\ell(\la)}\left[ h_{\la_i+j-i} \right].
\end{equation}

The \emph{Littlewood-Richardson coefficients} $c^\la_{\mu,\nu}$ are defined as the structure constants of $\Lambda$ with respect to the basis of Schur functions, i.e.~for all $\mu,\nu\in\Y$:\footnote{In~\eqref{eq:s s expansion}, $\la$ ranges over the set $\Y$ of all partitions; the same is true for all sums in this paper, where it is not specified where the partitions range over.}
\begin{equation} \label{eq:s s expansion}
    s_\mu\cdot s_\nu = \sum_\la{c^{\la}_{\mu,\nu}\,s_\la}.
\end{equation}

It is known that $c^{\la}_{\mu,\nu}=0$, unless $|\la|=|\mu|+|\nu|$ and $\mu,\nu\subseteq\la$. The Littlewood-Richardson coefficients can also be used to define the \emph{skew Schur functions} $s_{\la/\mu}$:
\begin{equation} \label{eq:skew schur LR expansion}
    s_{\la/\mu} := \sum_\nu{c^{\la}_{\mu,\nu}\,s_\nu},
\end{equation}
for any partitions $\mu,\la$, although $s_{\la/\mu}$ vanishes unless $\mu\subseteq\la$. The skew Schur functions satisfy the \emph{branching rule}
    \begin{equation} \label{eq:branching rules schur}
        s_{\la/\mu}(x,y) = \sum_{\nu} s_{\la/\nu}(x) s_{\nu/\mu}(y),
    \end{equation}
and the \emph{skew Cauchy identity}
    \begin{equation} \label{eq:skew cauchy id}
        \sum_{\la} s_{\la/\rho}(x) s_{\la/\eta}(y) = H(x;y) \sum_{\kappa} s_{\rho/\kappa}(y) s_{\eta/\kappa}(x),
    \end{equation}
where
    \begin{equation} \label{eq: H}
        H(x;y) := \exp \left\{ \sum_{k\ge 1} \frac{p_k(x) p_k(y)}{k} \right\}.
    \end{equation}
Equalities \eqref{eq:branching rules schur}, \eqref{eq:skew cauchy id} should be interpreted as identities in $\Lambda_x \otimes \Lambda_y$. In the case when $x,y$ are lists of finitely many variables and $\rho=\eta=\varnothing$, then~\eqref{eq:skew cauchy id} reduces to the classical Cauchy identity~\eqref{eqn:cauchy_identity}.

\section{Universal symplectic characters}\label{sec:symplectic_functions}

In this section, we recall the symplectic analogues of the Schur functions, which are denoted by $\SP_\la\in\Lambda$ and are called the universal symplectic characters.
Our main references are \cite{KT-1987,K-1989}.

\subsection{Definition and basic properties} \label{subs:SP}

\begin{df}
For any $\la\in\Y$, the \emph{universal symplectic character} $\SP_\la\in\Lambda$ is defined by:
\begin{equation*}
\SP_\la := \frac{1}{2} \det \left[ h_{\lambda_i-i+j} + h_{\lambda_i-i-j + 2} \right]_{1 \le i,j\le \ell(\lambda)}.
\end{equation*}
If $\la=\varnothing$, then $\SP_\varnothing:=1$, by convention.
\end{df}

The previous definition is an analogue of the Jacobi-Trudi formula~\eqref{eqn:jacobi_trudi} for Schur functions.
Some examples of universal symplectic characters are: $\SP_{(n)}=h_n$, for all $n\ge 1$, and $\SP_{(1,1)}=h_1^2-h_2-1$.

It is known~\cite{Lit-1950} that the universal symplectic characters admit the following expansion in terms of skew Schur functions
\begin{equation} \label{eq:symplectic Schur frobenius}
    \SP_\la = \sum_\alpha {(-1)^{|\alpha|/2} s_{\la/\alpha}},
\end{equation}
where $\alpha$ ranges over partitions with Frobenius coordinates of the form $(a_1, a_2, \cdots \,|\, a_1+1, a_2+1, \cdots )$. We recall that a partition $\la$ is said to have Frobenius coordinates $(a_1.a_2,\cdots\,|\, b_1,b_2,\cdots)$ if $a_i=\la_i-i$ is the number of boxes at row $i$ strictly to the right of the main diagonal, while $b_i=\la_i'-i$ is the number of boxes at column $i$ strictly to the left of the main diagonal.

The set $\{\SP_\la\}_{\la\in\Y}$ is an inhomogeneous basis of $\Lambda$.
One should regard $\SP_\la$ as the \emph{symplectic analogue} of the Schur function, as they are related to the symplectic Schur polynomials in the following fashion.
Consider the composition
\begin{equation}\label{eqn:big_pi}
\Pi^\infty_n : \Lambda_x\xrightarrow{\pi^\infty_{2n}} \C[z_1,\dots,z_{2n}]^{S_{2n}}
\xrightarrow{\tau_n}\C[z_1,z_1^{-1},\dots,z_n,z_n^{-1}]^{S_n\ltimes(\Z/2\Z)^n}
\end{equation}
of the natural projection $\pi^{\infty}_{2n}:\Lambda\to\C[z_1,\dots,z_{2n}]^{S_{2n}}$ and the map $\tau_n$ that sends $z_{2n+1-i}\mapsto z_i^{-1}$, for all $i=1,2,\dots,n$.
Then
\begin{equation}\label{SP_equations}
\Pi^\infty_n(\SP_\la) = \sp_\la(z_1^\pm,\dots,z_n^\pm),\text{ whenever }\ell(\la)\le n.
\end{equation}
In other words, the equations \eqref{SP_equations} say that if we set all but $(2n)$ variables to zero, set $n$ of them to $z_1,\dots,z_n$ and set the last $n$ to $z_1^{-1},\dots,z_n^{-1}$, then $\SP_\la$ turns into the symplectic Schur polynomial $\sp_\la(z_1^\pm,\dots,z_n^\pm)$, as long as $n\ge\ell(\la)$.
Given this observation, it is possible to define $\SP_\la$ as the unique element of $\Lambda$ such that the equations~\eqref{SP_equations} are satisfied for all $n$ and all $\la$ such that $\ell(\la)\le n$.

\begin{rem} \label{rem:SP for long partitions}
An important difference with the case of Schur functions is that it happens sometimes that $\Pi^\infty_n(\SP_\la)\ne 0$ whenever $\ell(\la)>n$.
In fact, when $\ell(\la)>n$, then $\Pi^\infty_n(\SP_\la)=0$ or $\pm\sp_\mu(z_1^\pm,\dots,z_n^\pm)$, for some $\mu\in\Y$ with $\ell(\mu)\le n$ and sign determined by $\la$; for a description of the precise value of $\Pi^\infty_n(\SP_\la)$, see~\cite{KT-1987}.
\end{rem}

From \Cref{cauchy_littlewood_prop} for symmetric polynomials, we deduce the following for symmetric functions.

\begin{prop}[Cauchy-Littlewood identity for universal symplectic characters]\label{cauchy_littlewood_prop_2}
Let $x=(x_1,x_2,\dots),\, y=(y_1,y_2,\dots)$ be two sequences of variables. Then
\begin{equation}\label{cauchy_littlewood_eqn_2}
\sum_\mu{\SP_\mu(x) s_\mu(y)} = G(y) H(x;y),
\end{equation}
where $H(x;y)$ was defined in \eqref{eq: H} and
\begin{equation}\label{eq: G}
    G(y) := \exp \left\{ - \sum_{k\ge 1} \frac{p_k(y)^2 - p_{2k}(y)}{2k} \right\}.
\end{equation}
\end{prop}

The identity~\eqref{cauchy_littlewood_eqn_2} is interpreted as an equality in $\Lambda_x\otimes\Lambda_y$.

\begin{rem}\label{rem:evaluation H and G}
    Functions $H(x;y)$ and $G(y)$ defined in \eqref{eq: H} and \eqref{eq: G}, when specializations $x,y$ consists of finitely many parameters, can be evaluated through the following elementary computations; see \cite[Chapter I, section 2]{M-1995}. When $x=(x_1,\dots,x_n)$ and $y=(y_1,\dots,y_m)$, for some $n,m\in\N$, and we have $|x_i y_j|<1$, for all $i,j$, we have
    \begin{equation}
        H(x,y) = \exp\left\{ \sum_{i=1}^n \sum_{j=1}^m \sum_{k \ge 1} \frac{(x_i y_j)^k}{k} \right\} = \exp\left\{ - \sum_{i=1}^n \sum_{j=1}^m \log(1-x_i y_j)  \right\} = \prod_{i=1}^n \prod_{j=1}^m \frac{1}{1-x_i y_j}.
    \end{equation}
    Similarly, if $|y_iy_j|<1$ for all $i,j$, we have
    \begin{equation}
        G(y) = \exp\left\{ - \sum_{1 \le i < j \le m} \sum_{k \ge 1} \frac{(y_i y_j)^k}{k} \right\} = \exp\left\{ \sum_{1 \le i < j \le m} \log(1-y_i y_j) \right\} = \prod_{1 \le i < j \le m} (1-y_i y_j).
    \end{equation}
    These expressions for $H(x,y)$ and $G(y)$ clarify that \Cref{cauchy_littlewood_prop_2} is indeed the symmetric function version of the identity in \Cref{cauchy_littlewood_prop}.
\end{rem}

\subsection{Newell-Littlewood coefficients}

The algebra structure constants with respect to the basis $\{\SP_\la\}_{\la\in\Y}$ are denoted by $d^{\la}_{\mu,\nu}$ and called the \emph{Newell-Littlewood coefficients}. Precisely, they are defined by the equations
\begin{equation} \label{eq:SP SP expansion}
    \SP_\mu \cdot \SP_\nu = \sum_{\la}{d^\la_{\mu,\nu}\, \SP_\la}
\end{equation}
for all $\mu,\nu\in\Y$.
By comparing this equation with~\eqref{eq:s s expansion}, the Newell-Littlewood coefficients should be regarded as ``symplectic analogues'' of the Littlewood-Richardson coefficients.

\begin{rem}
    The Newell-Littlewood coefficients $d^\la_{\mu,\nu}$ can also be regarded as the ``orthogonal analogues'' of the Littlewood-Richardson coefficients, because they satisfy a version of equation~\eqref{eq:SP SP expansion}, but with \emph{universal orthogonal characters} instead of the symplectic ones, as shown in~\cite{KT-1987}.
    Consequently, we believe that most results in this article have ``orthogonal versions'', but we stick to the symplectic case to avoid lengthening the paper.
\end{rem}

\begin{thm}[Thm.~3.1~from~\cite{K-1989}]\label{thm_koike}
For any $\la,\mu,\nu\in\Y$, we have
\begin{equation}\label{thm_koike_eqn}
d^{\la}_{\mu,\nu} =  \sum_{\alpha,\beta,\gamma}{c^{\la}_{\alpha,\beta}\,c^{\mu}_{\gamma,\beta}\,c^{\nu}_{\alpha,\gamma}},
\end{equation}
where all indices in the sum range over the set of all partitions.
\end{thm}

Observe that only finitely many terms in the sum of \Cref{thm_koike} are actually nonzero because $c^{\la}_{\alpha,\beta}\ne 0$ implies that $\alpha,\beta\subseteq\la$, and similarly, $\gamma\subseteq\mu$.

\begin{lem}\label{lemma_NL_coeffs}
\begin{enumerate}[label=(\alph*)]
	\item $d^\la_{\mu,\nu}$ is symmetric with respect to any permutation of the partitions $\la,\mu,\nu$.

	\item $d^\la_{\mu,\varnothing}=\delta_{\la,\mu}$, for any $\la,\mu\in\Y$.

	\item $d^\la_{\mu,\nu}\ne 0$ implies that $|\ell(\la)-\ell(\mu)|\le\ell(\nu)\le\ell(\la)+\ell(\mu)$, i.e.~the lengths $\ell(\la),\ell(\mu),\ell(\nu)$ form the sides of a possibly degenerate triangle.
\end{enumerate}
\end{lem}
\begin{proof}
For (a) and (c), see \cite[Lem.~2.2]{GOY-2021}. For (b), we use \Cref{thm_koike}. Note that, as $c^{\rho}_{\rho_1,\rho_2}\ne 0$ implies $\rho_1,\rho_2\subseteq\rho$, then the only summands in~\eqref{thm_koike_eqn} for $d^{\la}_{\mu,\varnothing}$ that do not vanish correspond to $\alpha=\gamma=\varnothing$.
Therefore $d^{\la}_{\mu,\varnothing}=\sum_\beta c^{\la}_{\varnothing,\beta}c^\mu_{\varnothing,\beta}c^\varnothing_{\varnothing,\varnothing}= \sum_\beta c^{\la}_{\varnothing,\beta}c^\mu_{\varnothing,\beta}$.
And since $c^{\rho_1}_{\varnothing,\rho_2}=\delta_{\rho_1,\rho_2}$, then we obtain $d^{\la}_{\mu,\varnothing}=\sum_\beta{c^{\la}_{\varnothing,\beta}c^\mu_{\varnothing,\beta}} = \sum_\beta{\delta_{\la,\beta}\delta_{\mu,\beta}} = \delta_{\la,\mu}$.
\end{proof}

\begin{rem}
Unlike for the Littlewood-Richardson coefficients, $d^{\la}_{\mu,\nu}\ne 0$ implies neither $|\la|=|\mu|+|\nu|$ nor $\mu,\nu\subseteq\la$, for example:
\begin{equation*}
d^{(2,1)}_{(3),(1,1)} = d^{(4)}_{(3),(3)} = 1.
\end{equation*}
\end{rem}

\subsection{Branching rule}

The following result is the special case of \cite[Theorem 7.22]{Rai-2005}, when the Koornwinder parameters are specialized to $(t_0,t_1,t_2,t_3)=(q,-q,\sqrt{q},-\sqrt{q})$ (see Remark~2 following~\cite[Theorem~7.19]{Rai-2005}).
We provide a self-contained proof, for completeness.

\begin{thm}[Branching rule for universal symplectic characters]\label{branching_SP_thm}
Let $\la$ be any partition, and let $y=(y_1,y_2,\dots),\, z=(z_1,z_2,\dots)$ be two sequences of variables. Then
\begin{equation}\label{branching_SP}
\SP_\la(y,z) = \sum_{\mu}{s_{\la/\mu}(y)\SP_\mu(z)}.
\end{equation}
The equation above should be interpreted as an identity in $\Lambda_y\otimes\Lambda_z$.
\end{thm}

\begin{rem}
In the formula above, $\SP_\la(y,z)$ is a symmetric function in the infinite set of variables $y_1,y_2,\dots,z_1,z_2,\dots$. In other words, this means that $\SP_\la(y,z)$ is the image of $\SP_\la\in\Lambda$ under the comultiplication map $\Delta:\Lambda\to\Lambda\otimes\Lambda\cong\Lambda_y\otimes\Lambda_z$ defined as the algebra homomorphism with $\Delta(p_n)=p_n\otimes 1+1\otimes p_n$, for all $n\ge 1$, see e.g.~\cite[Ch.~I, Sec.~5, Example~25]{M-1995}.
\end{rem}

\begin{proof}
Let $u=(u_1,u_2,\dots)$ be yet another set of variables.
We claim that the following identity holds
\begin{equation}\label{cauchy_littlewood}
\sum_{\la}{\SP_\la(y,z)s_\la(u)} = \sum_\la{\left(\sum_\mu{s_{\la/\mu}(y)\SP_\mu(z)}\right)s_\la(u)},
\end{equation}
where all sums are over the set of all partitions. By the Cauchy-Littlewood identity, the left hand side of~\eqref{cauchy_littlewood} is equal to $G(u)H(y,z;u)$.
On the other hand, by the skew Cauchy identity \eqref{eq:skew cauchy id}, and the Cauchy-Littlewood identity, the right hand side of~\eqref{cauchy_littlewood} is equal to
\begin{align*}
\sum_\mu{\left(\sum_\la{s_{\la/\mu}(y) s_\la(u)}\right) \SP_\mu(z)}
&= H(y;u) \sum_\mu s_\mu(u) \SP_\mu(z) \\
&= G(u) H(y,z;u).
\end{align*}
Then \eqref{cauchy_littlewood} is proved; the desired identity then follows by equating the coefficients of $s_\la(u)$ on both sides.
\end{proof}

\begin{rem}
    In light of the combinatorial interpretation of the symplectic Schur polynomials, the branching rule looks counterintuitive. In fact, the definition of $\sp_\la(x_1^\pm,\dots,x_n^\pm)$ as the generating function of symplectic tableaux of shape $\la$ gives rise to an alternative branching formula, which for partitions $\la$, such that $\ell(\la)\le n$, reads
    \begin{equation}\label{eq:wrong branching rule intro}
        \sp_\la(x_1^\pm, \dots, x_{n-k}^\pm,y_1^\pm,\dots,y_k^\pm) 
        = \sum_{\mu=(\mu_1,\dots,\mu_{n-k})} \sp_\mu(x_1^\pm,\dots,x_{n-k}^\pm) \,\sp_{\la/\mu}(y_1^\pm,\dots,y_k^\pm),
    \end{equation}
    where the branching factors $\sp_{\la/\mu}(y_1^\pm,\dots,y_k^\pm)$ are generating functions of certain symplectic skew-tableaux. Recently, explicit determinantal formulas of Jacobi-Trudi type were found for $\sp_{\la/\mu}(y_1^\pm,\dots,y_k^\pm)$ in~\cite{JLW-2024,AFHSA-2024}. The caveat to reconcile the two (apparently contradictory) expressions \eqref{branching_SP} and \eqref{eq:wrong branching rule intro} is that, while the first summation ranges over all partitions, the latter ranges over partitions $\mu$ of length bounded by $n-k$.
    Another difference is that~\eqref{branching_SP} is an identity of symmetric functions, while~\eqref{eq:wrong branching rule intro} is the result of specializing to finitely many variables.
    In principle, it should be possible to derive \eqref{eq:wrong branching rule intro} from \eqref{branching_SP}, but we will not attempt this.
\end{rem}

\section{The down-up Schur functions}\label{sec:du_schur}

In this section, we propose and develop the properties of a ``symplectic analogue" of the skew Schur functions $s_{\la/\mu}$ that we call the \emph{down-up Schur functions} $T_{\la,\mu}\in\Lambda$.
They depend on two partitions $\la,\mu$, and unlike the skew Schur functions, one of the partitions does not need to contain the other. In the special case of symmetric functions in finitely many variables, these functions were defined recently in \cite{JLW-2024} (there, denoted as $S^*_{\la/\mu}$) by the formalism of vertex operators. 

\subsection{The down-up Schur functions}\label{sec:updown_functions}

\begin{df}\label{main_df_fns}
Given any $\la,\mu\in\Y$, define the symmetric function $T_{\la,\mu}\in\Lambda$ by:
\begin{equation} \label{eq:def T}
    T_{\la,\mu} := \sum_{\nu\in\Y}{d^\la_{\mu,\nu}\,s_\nu}.
\end{equation}
We will call $T_{\la,\mu}\in\Lambda$ the \emph{down-up Schur functions}.
The projection $\pi^\infty_k(T_{\la,\mu})$, $k\in\Z_{\ge 1}$, will be a symmetric $k$-variate polynomial denoted by $T_{\la,\mu}(y_1,\dots,y_k)$.
\end{df}

The name ``down-up Schur function'' is explained by the following proposition.

\begin{prop}[Down-up formula]\label{combinatorial_S_star}
For any $\la,\mu\in\Y$, we have
\begin{align}
T_{\la,\mu} = \sum_{\alpha}{s_{\la/\alpha} s_{\mu/\alpha}}.\label{ortho_formula_3}
\end{align}
\end{prop}

\begin{proof}
    Using the definition of functions $T_{\la,\mu}$, \Cref{thm_koike} and properties \eqref{eq:s s expansion}, \eqref{eq:skew schur LR expansion} of the LR coefficients, we have the chain of equalities
    \begin{equation*}
        \begin{split}
            T_{\la,\mu} = \sum_\nu d^\la_{\mu,\nu} s_\nu=
            \sum_{\alpha,\beta,\gamma,\nu} c_{\alpha, \beta}^\la c_{\alpha, \gamma}^\mu c_{\beta, \gamma}^\nu s_\nu
            =
            \sum_{\alpha,\beta,\gamma} c_{\alpha, \beta}^\la c_{\alpha, \gamma}^\mu s_\beta s_\gamma
            =
            \sum_{\alpha} s_{\la/\alpha} s_{\mu/\alpha},
        \end{split}
    \end{equation*}
    which completes the proof.
\end{proof}

Some immediate properties of the down-up Schur functions are given next.

\begin{prop}\label{prop_orth_functions}
\begin{enumerate}[itemsep=0.05cm,label=(\alph*)]
	\item $T_{\la,\mu}=T_{\mu,\la}$, for all $\la,\mu\in\Y$.

	\item $T_{\la,\varnothing}=s_\la$, for all $\la\in\Y$.

	\item $T_{\la,\mu}(y_1,\dots,y_k)=0$, if $k<|\ell(\la)-\ell(\mu)|$.
\end{enumerate}
\end{prop}

\begin{proof}
Parts (a) and (b) immediately follow from parts (a) and (b) from \Cref{lemma_NL_coeffs}. Finally note that if $\nu\in\Y$ is such that $\ell(\nu)<|\ell(\la)-\ell(\mu)|$, then part (c) from \Cref{lemma_NL_coeffs} implies that $d^{\la}_{\mu,\nu}=0$.
On the other hand if $\ell(\nu)\ge|\ell(\la)-\ell(\mu)|$ and $|\ell(\la)-\ell(\mu)|>k$, then $\ell(\nu)>k$ implies $s_\nu(y_1,\dots,y_k)=0$. Combining these two statements, part (c) follows.
\end{proof}

The down-up Schur functions satisfy the following variant of the branching rule.

\begin{prop}[Branching rule for down-up Schur functions]\label{branching_star}
Let $\rho,\mu\in\Y$ and let $y = (y_1,y_2,\dots)$ , $z=(z_1,z_2,\dots)$ be two sequences of variables. Then
\begin{equation} \label{branching_star_eqn}
\sum_\la{T_{\rho,\la}(y) T_{\la,\mu}(z)}
= H(y;z) \, T_{\rho,\mu}(y,z),
\end{equation}
where $H(y;z)$ was defined in \eqref{eq: H}.
\end{prop}

\begin{proof}
    Using the down-up formula (\Cref{combinatorial_S_star}) and the summation rules \eqref{eq:branching rules schur}, \eqref{eq:skew cauchy id}, we have
    \begin{equation*}
    \begin{split}
        \sum_\la {T_{\rho,\la}(y) T_{\la,\mu}(z)} 
        & = \sum_{\la,\kappa,\theta} s_{\rho/\kappa}(y) s_{\la/\kappa}(y) s_{\la/\theta}(z) s_{\mu/\theta}(z)
        \\
        & = H(y;z) \sum_{\eta,\kappa, \theta} s_{\rho/\kappa}(y) s_{\kappa/\eta}(z) s_{\theta/\eta}(y) s_{\mu/\theta}(z)
        \\
        & = H(y;z) \sum_{\eta} s_{\rho/\eta}(y,z)  s_{\mu/\eta}(y,z),
    \end{split}
    \end{equation*}
    which reduces to the right hand side of \eqref{branching_star_eqn}.
\end{proof}

\subsection{Cauchy and Littlewood identities}

\begin{thm}[Skew down-up Cauchy identity]\label{new_skew_cauchy_thm}
Let $\la,\mu\in\Y$ and let $y,z$ be two sequences of variables. Then
\begin{equation} \label{eq:skew cauchy up down}
\sum_\rho{s_{\rho/\mu}(y) T_{\rho,\la}(z)}
= H(y;z) \sum_\eta{s_{\la/\eta}(y)T_{\mu,\eta}(z)},
\end{equation}
where $H(y;z)$ was defined in \eqref{eq: H}.
\end{thm}

\begin{proof}
Using the down-up formula, the skew Cauchy identity, the branching rule of skew Schur functions and the symmetry of the skew Schur functions, we have
        \begin{equation*}
        \begin{split}
            \sum_\rho{s_{\rho/\mu}(y) T_{\rho,\la}(z)} &= \sum_{\rho, \kappa}s_{\rho/\mu}(y)  s_{\rho/\kappa}(z) s_{\la/\kappa}(z)
            \\
            &= H(y;z) \sum_{\alpha, \kappa}s_{\mu/\alpha}(z)  s_{\la/\kappa}(z) s_{\kappa/\alpha}(y)
            \\
            & = H(y;z) \sum_\alpha s_{\mu/\alpha}(z)  s_{\la/\alpha}(y,z)
            \\
            & = H(y;z) \sum_{\alpha, \eta}s_{\mu/\alpha}(z) s_{\eta/\alpha}(z) s_{\la/\eta}(y),
        \end{split}
        \end{equation*}
which reduces to the right hand side of \eqref{eq:skew cauchy up down}.
\end{proof}

\begin{cor}\label{cor:new_skew_cauchy_thm}
    Let $\la \in \Y$ and let $y,z$ be two sequences of variables. Then
    \begin{equation} \label{eq: skew cauchy T simplified}
        \sum_\rho{s_{\rho}(y) T_{\rho,\la}(z)}
        = H(y;z) s_{\la}(y,z),
\end{equation}
\end{cor}
\begin{proof}
    This is a consequence of \eqref{eq:skew cauchy up down} in the case $\mu=\varnothing$, of item (b) from \Cref{prop_orth_functions} and the branching rule~\eqref{eq:branching rules schur} for skew Schur functions.
\end{proof}

Next we have a summation identity that generalizes the Cauchy-Littlewood identity (\Cref{cauchy_littlewood_prop_2}) and includes both universal symplectic characters and down-up Schur functions.

\begin{thm}\label{gen_cauchy_littlewood_prop_2}
Let $\la\in\Y$ and let $x,y$ be two sequences of variables. Then
\begin{equation} \label{eq:skew_cauchy_littlewood_general}
\sum_\mu{\SP_\mu(x) T_{\la,\mu}(y)} =
G(y) H(x;y)\,\SP_\la(x),
\end{equation}
where $H(x;y)$ and $G(y)$ were defined in \eqref{eq: H} and \eqref{eq: G}, respectively.
\end{thm}
\begin{proof}
By the Cauchy-Littlewood identity (\Cref{cauchy_littlewood_prop_2}) and the definition~\eqref{eq:SP SP expansion} of the Newell-Littlewood coefficients, we can transform the right hand side of \eqref{eq:skew_cauchy_littlewood_general} through the following chain of equalities
\begin{align*}
\left( \sum_{\nu}{\SP_\nu(x)s_\nu(y)}\right) \cdot \SP_\la(x)
&= \sum_{\nu}{\SP_\nu(x) \SP_\la(x)} s_\nu(y)\\
&= \sum_{\nu} \sum_{\mu}{d^{\mu}_{\nu,\la}\, \SP_\mu(x)} s_\nu(y)\\
&= \sum_{\mu} \left( \sum_{\nu}{d^{\mu}_{\nu,\la}s_\nu(y)} \right) \cdot \SP_\mu(x),\end{align*}
where in the last line we recognize the defining equation~\eqref{eq:def T} of the $T_{\la,\mu}$ functions, thus finishing the proof.
\end{proof}

Projecting the identity of \Cref{gen_cauchy_littlewood_prop_2} to the space of finite variables $x_1,\dots,x_n$, $y_1,\dots,y_k$, we obtain the following.

\begin{cor} \label{gen_cauchy_littlewood_prop}
Let $x^\pm=(x_1^\pm,\dots,x_n^\pm)$, $y=(y_1,\dots, y_k)$ be such that $|x_iy_j|,|x_i^{-1}y_j|<1$, for all $i,j$, and let $\la\in\Y$ be such that $\ell(\la)+k\le n$. Then we have
    \begin{equation}\label{eqn:generalization_CL_identity}
        \sum_{\mu\,:\,\ell(\mu)\le\ell(\la)+k}{\sp_\mu(x^\pm)\,T_{\la,\mu}(y)} =
        \frac{\prod_{1\le i<j \le k}(1 - y_iy_j)}{\prod_{i= 1}^n \prod_{j= 1}^k  (1 - x_iy_j) (1 - x_i^{-1}y_j)}\,\sp_{\la}(x^\pm).
    \end{equation}
\end{cor}
\begin{proof}
The desired identity follows by applying both the canonical projection $\pi^\infty_k$ to the $y$-variables and the projection $\Pi^\infty_n$ (see~\eqref{eqn:big_pi}) to the $x$-variables in Equation~\eqref{eq:skew_cauchy_littlewood_general}.
For the right hand side, we need to observe that $H(x;y)$ turns into $\prod_{i,j}(1-x_iy_j)^{-1}(1-x_i^{-1}y_j)^{-1}$, whereas $G(y)$ turns into $\prod_{i<j}(1-y_iy_j)$, as shown in \Cref{rem:evaluation H and G}, while $\SP_\la$ turns into $\sp_\la(x^\pm)$ by Equation~\eqref{SP_equations}, because $\ell(\la)\le n$.
For the left hand side, we also need item (c) from \Cref{prop_orth_functions}, which allows us to restrict the sum from the set of all partitions $\mu$ to those satisfying $\ell(\mu)\le\ell(\la)+k$.
By the condition $\ell(\la)+k\le n$, it then follows that $\ell(\mu)\le n$, and so $\SP_\mu$ turns into $\sp_\mu(x^\pm)$, thus completing the proof.
\end{proof}

\begin{rem}
    Identity \eqref{eqn:generalization_CL_identity} (as well as some other finite-variable versions of our identities) has appeared in a slightly more general form in~\cite[Equation (5.24)]{JLW-2024}. There, the authors expressed the symmetric polynomials $T_{\la,\mu}(y_1,\dots,y_k)$ (using the notation $S_{\la/\mu}^*(y_1,\dots,y_k)$) as the Weyl-type determinant
    \begin{equation}\label{jlw_formula}
   \det_{1\le i,j\le\ell(\mu)+k}\left[ \mathbf{1}_{1\le i \le k} \, \cdot   y_i^{\la_j-j+k} + \mathbf{1}_{k+1\le i \le k+\ell(\mu)} \,\cdot h_{i-k-\la_j + j -i}(y_1^\pm ,\dots,y_k^\pm)  \right] \prod_{1\le i < j \le k}(y_i-y_j)^{-1}.
    \end{equation}
    \Cref{main_df_fns}, or equivalently Equation~\eqref{ortho_formula_3}, offer different and more general characterizations of these skew-dual variants of the symplectic Schur polynomials.
    We remark that it is not evident that formula~\eqref{jlw_formula} defines a symmetric polynomial and much less that it has the stability property that leads to a symmetric function.
\end{rem}

We also report the following dual variant of the generalized Cauchy-Littlewood identity from \Cref{gen_cauchy_littlewood_prop_2}.
\begin{prop}
Let $\la\in\Y$ and let $x,y$ be two sequences of variables. Then
    \begin{equation} \label{eq:dual Littlewood identity}
        \sum_\mu{\SP_\mu(x) T_{\la',\mu'}(y)} = \overline{G}(y) E(x;y) \,\SP_{\la}(x),
    \end{equation}
    where 
    \begin{equation*}
        \overline{G}(y) := \exp \left\{ - \sum_{k\ge 1} \frac{p_k(y)^2 + p_{2k}(y)}{2k} \right\},
        \qquad
        E(x;y) :=
        \exp \left\{ \sum_{k\ge 1} (-1)^{k-1} \frac{p_k(x) p_k(y)}{k} \right\}.
    \end{equation*}
\end{prop}
\begin{proof}
    Consider the Hall involution $\omega:\Lambda_y \to \Lambda_y$ defined by the relation $\omega (p_k) = (-1)^{k-1}p_k$, for all $k\in\Z_{\ge 1}$. It is known (see \cite[Ch.~I.5, Example~30(c)]{M-1995}) that the involution acts on skew Schur functions as  $\omega(s_{\la/\kappa}(y))= s_{\la'/\kappa'}(y)$. This implies, by the down-up formula~\eqref{ortho_formula_3}, that
    \begin{equation*}
        \omega(T_{\la,\mu}) = T_{\la',\mu'}.
    \end{equation*}
    Then, applying the involution $\omega$ to the left hand side of identity \eqref{eq:skew_cauchy_littlewood_general}, we find
    \begin{equation*}
        \omega\left( \sum_\mu{\SP_\mu(x)\,T_{\la,\mu}(y)} \right) = \sum_\mu{\SP_\mu(x)\,T_{\la',\mu'}(y)}.
    \end{equation*}
    Focusing on the right hand side of \eqref{eq:skew_cauchy_littlewood_general}, we find
    \begin{equation*}
    \begin{split}
        \omega\left( G(y) \right) = \overline{G}(y),
        \qquad
        \omega\left( H(x;y) \right) = E(x;y).
    \end{split}
    \end{equation*}
    Combining the above relations shows that the action of $\omega$ transforms the right hand side of  \eqref{eq:skew_cauchy_littlewood_general} into that of \eqref{eq:dual Littlewood identity}, completing the proof.
\end{proof}

\section{The symplectic Schur process}

\subsection{Specializations of the algebra of symmetric functions}

\begin{df}[Specializations of $\Lambda$]\label{specs_sym}
\begin{enumerate}[label=(\alph*)]
	\item Any unital algebra homomorphism $\a:\Lambda\to\R$ is said to be a \emph{specialization of $\Lambda$}.
For any symmetric function $f\in\Lambda$, we denote $\a(f)$ by $f(\a)$.

	\item Since $\Lambda=\R[p_1,p_2,\dots]$, a specialization $\a:\Lambda\to\R$ is determined by the values $p_k(\a)$, $k\ge 1$. If $\a_1,\a_2:\Lambda\to\R$ are any two specializations, then the \emph{union specialization} $\a_1\cup \a_2:\Lambda\to\R$ is defined by $p_k(\a_1\cup\a_2):=p_k(\a_1)+p_k(\a_2)$, for all $k\ge 1$. Similarly, for any number of specializations $\a_1,\dots,\a_n$, one can define the union $\a_1\cup\dots\cup\a_n$.

	\item A specialization $\a:\Lambda\to\R$ is said to be \emph{Schur-positive} if $s_\la(\a)\ge 0$, for all $\la\in\Y$.
\end{enumerate}
\end{df} 

Since the skew Schur functions are $s_{\la/\mu} := \sum_\nu{c^{\la}_{\mu,\nu}s_\nu}$ and the Littlewood-Richardson coefficients are nonnegative, then $s_{\la/\mu}(\a)\ge 0$, for all Schur-positive specializations $\a:\Lambda\to\R$.
Similarly, since the down-up Schur functions are $T_{\la,\mu} := \sum_\nu{d^\la_{\mu,\nu}\,s_\nu}$ and, by \Cref{thm_koike}, the Newell-Littlewood coefficients are nonnegative, then $T_{\la,\mu}(\a)\ge 0$, for all Schur-positive specializations $\a:\Lambda\to\R$ and all $\la,\mu\in\Y$.

\begin{exam}[Empty specialization]
The \emph{empty specialization}, to be denoted by $\varnothing:\Lambda\to\R$, is defined by $p_k(\varnothing)=0$, for all $k\ge 1$. Evaluating skew Schur functions with the empty specialization yields $s_{\la/\mu}(\varnothing) = \mathbf{1}_{\la=\mu}$, while for the universal symplectic characters the result is far less trivial; for instance, from the expansion \eqref{eq:symplectic Schur frobenius}, we deduce that $\SP_{\varnothing}(\varnothing)=1$ and $\SP_{(1,1)}(\varnothing)=-1$.
\end{exam}

\begin{exam}[Variable specializations of $\Lambda$]\label{exam_var_spec}
For any $n\in\Z_{\ge 0}$, let $y_1,\dots,y_n\in\R$ be arbitrary.
Denote $\y:=(y_1,\dots,y_n)$ and by the same cursive letter denote the specialization $\y:\Lambda\to\R$ given by
\begin{equation*}
p_k(\y) := y_1^k+\cdots+y_n^k,\text{ for all $k\ge 1$}.
\end{equation*}
Then $\y$ will be called the \emph{variable specialization} associated to $(y_1,\dots,y_n)$.

Recall that, for any $\la\in\Y$, $s_\la(\y)=s_\la(y_1,\dots,y_n)$ is a polynomial in $y_1,\dots,y_n$ with nonnegative coefficients, because of the combinatorial formula in terms of semistandard Young tableaux~\cite[Ch.~I,~(5.12)]{M-1995}.
As a result, if $\y=(y_1,\dots,y_n)\in(\R_{\ge 0})^n$, then the corresponding variable specialization is Schur-positive.
\end{exam}

\begin{exam}[Laurent variable specializations of $\Lambda$]\label{exam_var_spec_2}
For any $n\in\Z_{\ge 0}$, let $x_1,\dots,x_n\in\R^*$ be any nonzero real numbers.
Denote $\x^\pm:=(x_1^\pm,\dots,x_n^\pm)$ and by the same letter denote the specialization $\x^\pm:\Lambda\to\R$ given by
\begin{equation*}
p_k(\x^\pm) := x_1^k+x_1^{-k}+\cdots+x_n^k+x_n^{-k},\text{ for all $k\ge 1$}.
\end{equation*}
Then $\x^\pm$ will be called the \emph{Laurent variable specialization} associated to $(x_1^\pm,\dots,x_n^\pm)$.

We have seen in \Cref{subs:SP} that for all $\la\in\Y$ such that $\ell(\la)\le n$, then $\SP_\la(\x^\pm) = \sp_\la(x_1^\pm,\dots,x_n^\pm)$ is a Laurent polynomial in $x_1,\dots,x_n$ with nonnegative coefficients, because of the combinatorial formula in terms of Young tableaux~\cite{King75}.
However, following
\Cref{rem:SP for long partitions}, if $\ell(\la) > n$, then sometimes $\SP_\la(x^\pm)<0$, even if $x_1,\dots,x_n>0$.
\end{exam}

\subsection{The partition function} \label{subs:partition function}

Assume that $k\in\Z_{\ge 1}$ and that we have $(2k)$ specializations $\a^1,\dots,\a^k,\b^1,\dots,\b^k:\Lambda\to\R$.
It will be convenient to set
\begin{equation*}
\veca := (\a^1,\dots,\a^k),\quad \vecb := (\b^1,\dots,\b^k).
\end{equation*}
It will be convenient to use the notation for union specializations: 
\begin{equation} \label{eq:notation union specialization}
    \a^{[i,j]} := \bigcup_{i\le\ell\le j}{\a^\ell},\qquad
\b^{[i,j]} := \bigcup_{i\le\ell\le j}{\b^\ell}.
\end{equation}

\begin{df}\label{input_process}
Let $\la^{(1)},\dots,\la^{(k)},\mu^{(1)},\dots,\mu^{(k-1)}$ be any $(2k-1)$ partitions and set
\begin{equation*}
\lavec := (\la^{(1)},\dots,\la^{(k)}),\quad\muvec := (\mu^{(1)},\dots,\mu^{(k-1)}).
\end{equation*}
Then to the $(2k-1)$-tuple $(\lavec,\muvec)$, associate the weight $\mathcal{W}(\lavec,\muvec \mid \veca,\vecb)$ given by the product:
\begin{align} \label{eq: W}
    \mathcal{W}\!\left( \lavec,\muvec \,\big|\, \veca,\vecb \right) :=
    s_{\la^{(1)}}(\b^1)\cdot \prod_{j=1}^{k-1}{\left[ s_{\la^{(j)}/\mu^{(j)}}(\a^j)\, T_{\mu^{(j)},\la^{(j+1)}}(\b^{j+1}) \right]} \cdot \SP_{\la^{(k)}}(\a^k).
\end{align}
The sum 
\begin{equation*}
\mathcal{Z}(\veca,\vecb) := \sum_{\lavec,\muvec}{\mathcal{W}\!\left( \lavec,\muvec \,\big|\, \veca,\vecb \right)}
\end{equation*}
over all tuples $(\lavec,\muvec)=(\la^{(1)},\dots,\la^{(k)},\mu^{(1)},\dots,\mu^{(k-1)})\in\Y^{2k-1}$ will be called the \emph{partition function}.
\end{df}

\begin{rem}\label{rem:tensors}
It is possible that the partition function $\mathcal{Z}(\veca,\vecb)$ diverges, but we can avoid this unpleasant scenario by regarding the weight $\mathcal{W}(\lavec,\muvec \mid \veca,\vecb)$ not as a real number, but as an element of the tensor product $\Lambda_{\a^1}\otimes\cdots\otimes\Lambda_{\a^k}\otimes\Lambda_{\b^1}\otimes\cdots\otimes\Lambda_{\b^k}$, where each $\Lambda_{\a^i}, \Lambda_{\b^j}$ is a copy of $\Lambda$.
By using this interpretation, \Cref{partition_function_1} below, which gives an explicit product formula for the partition function, is always true in a suitable completion of the tensor product of $\Lambda$'s.\footnote{This convention is implicitly used to write the classical Cauchy identity for Schur functions or the Cauchy-Littlewood identity for universal symplectic characters in \Cref{cauchy_littlewood_prop_2}.}
\end{rem}

\begin{rem} \label{rem:support W}
Note that the weight $\mathcal{W}(\lavec,\muvec \mid \veca,\vecb)$ vanishes, unless $\mu^{(j)}\subseteq\la^{(j)}$, for all $j=1,2,\dots,k-1$.
As a result, the partition function can be calculated by adding only over $(2k-1)$-tuples of partitions $(\lavec,\muvec)$ with this additional restriction.
\end{rem}

\begin{thm}\label{partition_function_1}
    For any $k\in\Z_{\ge 1}$ and $k$-tuples of specializations $\veca=(\a^1,\dots,\a^k)$, $\vecb=(\b^1,\dots,\b^k)$, the partition function (according to \Cref{input_process}) is equal to 
    \begin{equation}\label{eqn_partition_function}
        \mathcal{Z}(\veca,\vecb) := \sum_{\lavec,\muvec}{ \mathcal{W} \!\left( \lavec,\muvec \,\big|\, \veca,\vecb \right)} = \prod_{1\le p\le q\le k}{H(\a^q;\b^p)}\cdot\prod_{r=1}^k{G(\b^r)},
    \end{equation}
    where the functions $H,G$ were defined in \eqref{eq: H}, \eqref{eq: G}.
\end{thm}

\begin{proof}
    Using the skew down-up Cauchy identity~\eqref{eq:skew cauchy up down} and the branching rule for skew Schur functions~\eqref{eq:branching rules schur} we find 
    \begin{equation*}
        \sum_{\la^{(1)}, \la^{(2)}, \mu^{(1)}} \mathcal{W}\left( \lavec,\muvec \,\big|\, \veca,\vecb \right) = \frac{H(\b^1;\b^2)}{H(\a^1;\b^2)}\cdot\sum_{\rho^{(1)}} \mathcal{W} \left( \lavec',\muvec' \,\big|\, \veca',\vecb' \right),
    \end{equation*}
    where, recalling notation \eqref{eq:notation union specialization},
    \begin{gather*}
        \veca' := (\a^{[1,2]},\a^3,\dots,\a^k),\quad\vecb' := (\b^{[1,2]},\b^3,\dots,\b^k),\\
        \lavec' := (\rho^{(1)},\la^{(3)},\dots,\la^{(k)}),\quad\muvec' := (\mu^{(2)},\dots,\mu^{(k-1)}).
    \end{gather*}
    By repeating the same steps $(k-1)$ times in total, we end up with:
    \begin{equation} \label{part_fn_previous}
        \begin{split}
            \sum_{\lavec,\muvec}{\mathcal{W} \!\left( \lavec,\muvec \,\big|\, \veca,\vecb \right)}
            = & \frac{H(\b^1 ; \b^2)}{H(\a^1 ; \b^2)}
            \frac{H(\b^{[1,2]} ; \b^3)}{H(\a^{[1,2]} ; \b^3)}
            \cdots\frac{H(\b^{[1,k-1]} ; \b^k)}{H(\a^{[1,k-1]} ; \b^k)}
            \\
            & \qquad \qquad \qquad \times \sum_{\rho^{(k-1)}}{s_{\rho^{(k-1)}}(\b^{[1,k]})\SP_{\rho^{(k-1)}}(\a^{[1,k]})}.
        \end{split}
    \end{equation}
    We point out that in the last $(k-1)$-th step, we used the branching rule for universal symplectic characters (\Cref{branching_SP_thm}) rather then the branching rule~\eqref{eq:branching rules schur} for skew Schur functions.
    Then by the Cauchy-Littlewood identity for universal symplectic characters (\Cref{gen_cauchy_littlewood_prop_2}), we conclude from~\eqref{part_fn_previous} that
    \begin{equation}\label{partition_part_1}
        \sum_{\lavec,\muvec}{\mathcal{W}\!\left( \lavec,\muvec \,\big|\, \veca,\vecb \right)}
        = \frac{H(\b^1;\b^2)}{H(\a^1;\b^2)}
        \frac{H(\b^{[1,2]};\b^3)}{H(\a^{[1,2]};\b^3)}
        \cdots\frac{H(\b^{[1,k-1]};\b^k)}{H(\a^{[1,k-1]};\b^k)}
        \cdot H(\a^{[1,k]};\b^{[1,k]})\,G(\b^{[1,k]}).
    \end{equation}
    We can next simplify the $H$'s in~\eqref{partition_part_1} as follows:
    \begin{equation}\label{partition_part_2}
        \begin{gathered}
            H(\b^{[1,j]} ; \b^{j+1}) = \prod_{i=1}^j{H(\b^i ; \b^{j+1})},\qquad
            H(\a^{[1,j]} ; \b^{j+1}) = \prod_{i=1}^j{H(\a^i ; \b^{j+1})},\\
            H(\a^{[1,k]} ; \b^{[1,k]}) = \prod_{1\le p,q\le k}{H(\a^q;\b^p)}.
        \end{gathered}
    \end{equation}
    Moreover, from the definition of the generating function $G$ in~\eqref{eq: G}, we deduce
    \begin{equation}\label{partition_part_3}
        G(\b^{[1,k]}) = \frac{1}{\prod_{1\le i<j\le k}{H(\b^i;\b^j)}}\cdot\prod_{r=1}^k{G(\b^r)}.
    \end{equation}
    Plugging \eqref{partition_part_2} and \eqref{partition_part_3} into \eqref{partition_part_1} concludes the proof of the theorem.
\end{proof}

\subsection{General definition of the symplectic Schur process}\label{subs:symplectic Schur}

\begin{df}[Symplectic Schur process] \label{symplectic_process_1} 
    Let $\veca = (\a^1,\dots,\a^k)$, $\vecb=(\b^1,\dots,\b^k)$ be $k$-tuples of specializations and for any $(2k-1)$-tuple of partitions $(\lavec,\muvec)= (\la^{(1)},\dots,\la^{(k)},\mu^{(1)},\dots,\mu^{(k-1)})$ define the \emph{symplectic Schur process} $\Pr^\SSP\!\left(\lavec,\muvec \,\big|\, \veca,\vecb \right)$ as
    \begin{equation} \label{eqn:general_ssp}
        \Pr^\SSP\!\left(\lavec,\muvec \,\big|\, \veca,\vecb \right) := \frac{ \mathcal{W} \left(\lavec,\muvec \,\big|\, \veca,\vecb \right)}{\mathcal{Z} (\veca,\vecb)},
    \end{equation}
    where the weight $\mathcal{W}$ and the partition function $\mathcal{Z}$ are given in \eqref{eq: W} and \eqref{eqn_partition_function} respectively.
\end{df}

As mentioned in \Cref{rem:tensors}, the weight $\mathcal{W}$ can be regarded as an element of the tensor product 
$$
    \Lambda_{\veca} \otimes \Lambda_{\vecb} := \Lambda_{\a^1} \otimes  \cdots  \otimes \Lambda_{\a^k} \otimes \Lambda_{\b^1} \otimes \cdots  \otimes \Lambda_{\b^k},
$$ 
while the partition function $\mathcal{Z}$ is a formal power series in $\Lambda_{\veca} \otimes \Lambda_{\vecb}$ by virtue of \Cref{partition_function_1}.
As a result, in general we can view $\Pr^\SSP\!\left( \,\cdot\,\big|\, \veca,\vecb \right)$ as a measure with values being formal power series in the graded tensor product $\Lambda_{\veca} \otimes \Lambda_{\vecb}$ --- this interpretation always makes sense without any convergence or positivity requirement.

Following \Cref{rem:support W}, the support of the weight $\mathcal{W} \left(\lavec,\muvec \,\big|\, \veca,\vecb \right)$ and hence of the symplectic Schur process are the sequences $(\lavec,\muvec)$ such that $\mu^{(j)}\subseteq\la^{(j)}$, for all $j=1,\dots,k-1$; see the diagram in \Cref{fig:support}. From \Cref{partition_function_1}, the measure $\Pr^\SSP(\,\cdot\mid \veca,\vecb)$ has total mass equal to $1$.
So if $\a^1,\dots,\a^k,\b^1,\dots,\b^k$ are such that $\Pr^\SSP(\lavec,\muvec \mid \veca,\vecb)\ge 0$, for all $\lavec,\muvec$, then $\Pr^\SSP(\,\cdot \mid \veca,\vecb)$ is a probability measure; we will examine a class of such choices of specializations in \Cref{sec:special_cases}.

\begin{rem}\label{rem_schur_pos}
If $\a^1,\dots,\a^{k-1},\b^1,\dots,\b^k$ are Schur-positive specializations, then 
\begin{equation*}
s_{\la^{(1)}}(\b^1)\cdot \prod_{j=1}^{k-1}{\left[ s_{\la^{(j)}/\mu^{(j)}}(\a^j)\, T_{\mu^{(j)},\la^{(j+1)}}(\b^{j+1}) \right]}\ge 0,
\end{equation*}
for all $\lavec,\muvec$. However, the classification of specializations $\a^k$ such that $\SP_\la(\a^k)\ge 0$, for all $\la\in\Y$, is unknown. For instance, as shown in \Cref{exam_var_spec_2}, finite Laurent variable specialization with positive variables $(x_1^\pm,\dots,x_n^\pm)$ do not yield $\SP_\la(x_1^\pm,\dots,x_n^\pm)>0$, for all $\la\in\Y$.
This shows that the classification problem for $\SP$-positive specializations of~$\Lambda$ (compared with the classification problem for Schur-positive specializations) appears to be more difficult.
\end{rem}

\begin{rem}
    The symplectic Schur process resembles the Schur process~\cite{OR-2003} discovered by Okounkov-Reshetikhin, and recalled in Equation~\eqref{eq:Schur process}, but neither one of them is more general than the other. In fact, while down-up Schur functions $T_{\la,\mu}$ appearing in the weight $\mathcal{W}(\lavec,\muvec \,|\, \veca, \vecb)$ can be expressed as traces of skew Schur functions, as in \eqref{ortho_formula_3}, the same is not true for the universal symplectic character $\SP_{\la^{(k)}}$, as shown by \eqref{eq:symplectic Schur frobenius}.
\end{rem}

When $k=1$, the symplectic Schur process turns into a measure on $\Y$ that coincides with the \emph{symplectic Schur measure}, studied by Betea~\cite{B-2018}. It is worth to establish a notation for this simplest case.

\begin{df}[Symplectic Schur measure]\label{symplectic_measure}
Let $\a^1,\b^1$ be two specializations of $\Lambda$.
The \emph{symplectic Schur measure} is the measure on $\Y$, denoted by $\M^\SS(\,\cdot\mid\a^1,\b^1)$, given by the formula:
\begin{equation*}
\M^\SS\!\left( \la\,\big|\,\a^1,\b^1 \right) := \frac{s_\la(\b^1)\SP_\la(\a^1)}{H(\a^1;\b^1)G(\b^1)},\quad\la\in\Y.
\end{equation*}
\end{df}

\subsection{Marginals of the symplectic Schur process}

Let $\a^1,\dots,\a^k,\b^1,\dots,\b^k$ be specializations and set $\veca=(\a^1,\dots,\a^k)$, $\vecb=(\b^1,\dots,\b^k)$.

\begin{prop}[Marginals of symplectic Schur processes are symplectic Schur processes]\label{prop:marginals}
Assume that $k\ge 2$.
Also let $(\lavec,\muvec)=(\la^{(1)},\dots,\la^{(k)},\mu^{(1)},\dots,\mu^{(k-1)})$ be a random $(2k-1)$-tuple distributed according to the symplectic Schur process $\Pr^\SSP(\,\cdot\mid\veca,\vecb)$ from Equation~\eqref{eqn:general_ssp}. Then

\begin{enumerate}
    \item[(i)] the $(2k-3)$-tuple
    \begin{equation*}
        (\lavec\setminus\{\la^{(k)}\},\,\muvec\setminus\{\mu^{(k-1)}\}) := (\la^{(1)},\dots,\la^{(k-1)},\mu^{(1)},\dots,\mu^{(k-2)})
    \end{equation*}
    is distributed according to $\Pr^\SSP(\,\cdot\mid\a',\b')$, where
    \begin{equation*}
        \a' := (\a^1,\dots,\a^{k-1}\cup\a^k),\qquad \b' := (\b^1,\dots,\b^{k-2},\b^{k-1}).
    \end{equation*}
    \item[(ii)] the $(2k-3)$-tuple
    \begin{equation*}
        (\lavec\setminus\{\la^{(1)}\},\,\muvec\setminus\{\mu^{(1)}\}) := (\la^{(2)},\dots,\la^{(k)},\mu^{(2)},\dots,\mu^{(k-1)})
    \end{equation*}
    is distributed according to $\Pr^\SSP(\,\cdot\mid\a'',\b'')$, where
    \begin{equation*}
        \a'' := (\a^2,\dots,\a^k),\qquad\b'' := (\b^1\cup\b^2,\b^3,\dots,\b^k).
    \end{equation*}
    \end{enumerate}
\end{prop}
\begin{proof}

    Combining \Cref{gen_cauchy_littlewood_prop_2} and the branching rule of universal symplectic characters of \Cref{branching_SP_thm}, we have
    \begin{multline*}
        \sum_{\mu^{(k-1)},\la^{(k)}}{ s_{\la^{(k-1)}/\mu^{(k-1)}}(\a^{k-1})\, T_{\mu^{(k-1)},\la^{(k)}}(\b^k)\, \SP_{\la^{(k)}}(\a^k) }\\
        = H(\a^k;\b^k)G(\b^k)\,\SP_{\la^{(k-1)}}(\a^{k-1}\cup\a^k),
    \end{multline*}
    which proves (i) upon inspecting the \Cref{symplectic_process_1} of $\Pr^\SSP(\,\cdot\mid\a,\b)$. Similarly, using identities \eqref{eq:skew cauchy id} and \eqref{eq: skew cauchy T simplified}, we have
    \begin{equation*}
        \sum_{\mu^{(1)},\la^{(1)}}{ s_{\la^{(1)}}(\b^1) s_{\la^{(1)}/\mu^{(1)}}(\a^1) T_{\mu^{(1)},\la^{(2)}}(\b^2) } = H(\a^1;\b^1)H(\b^1;\b^2)\, s_{\la^{(2)}}(\b^1\cup\b^2),
    \end{equation*}
    which implies (ii).
\end{proof}

One-dimensional marginals of any symplectic Schur process are the symplectic Schur measures of \Cref{symplectic_measure}, as we show next.

    \begin{cor}
        Consider the symplectic Schur process, adopting the notation of \Cref{symplectic_process_1}. Then the distribution of $\la^{(j)}$ is the symplectic Schur measure $\M^\SS(\,\cdot\mid\a^{[j,k]},\b^{[1,j]})$ from \Cref{symplectic_measure}, for all $j=1,\dots, k$. Similarly, the distribution of $\mu^{(j)}$ is the symplectic Schur measure $\M^\SS(\,\cdot\mid\a^{[j+1,k]},\b^{[1,j]})$, for all $j=1,\dots, k-1$. In formulas, we have for all $\eta\in\Y$ that
        \begin{align*}
        \Pr^\SSP(\la^{(j)} = \eta \mid\veca,\vecb) &= \M^\SS(\,\eta \mid \a^{[j,k]},\b^{[1,j]}),\\
        \Pr^\SSP( \mu^{(j)} = \eta \mid\veca,\vecb) &= \M^\SS( \,\eta \mid \a^{[j+1,k]},\b^{[1,j]}).
        \end{align*}
    \end{cor}
    \begin{proof}
        The law of the marginal $\la^{(j)}$ of the symplectic Schur process can be computed by applying \Cref{prop:marginals} $(k-1)$ times to trace out the pairs $(\la^{(1)},\mu^{(1)}),\dots,(\la^{(j-1)},\mu^{(j-1)})$, $(\mu^{(j)},\la^{(j+1)}),\dots,(\mu^{(k-1)},\la^{(k)})$, yielding the claim. 
        
        As for $\mu^{(j)}$, apply \Cref{prop:marginals} $(k-2)$ times to trace out the pairs $(\la^{(1)},\mu^{(1)})$, $\dots$, $(\la^{(j-1)},\mu^{(j-1)})$, $(\mu^{(j+1)},\la^{(j+2)}),\dots,(\mu^{(k-1)},\la^{(k)})$: the end result is that $(\la^{(j)},\mu^{(j)},\la^{(j+1)})$ is distributed as the symplectic Schur process $\Pr^\SSP(\,\cdot\mid \a',\b')$, where $\a' := (\a^j,\a^{[j+1,k]})$ and $\b' := (\b^{[1,j]},\b^{j+1})$.
        Finally, to find the law of the marginal $\mu^{(j)}$, we use the identity
        \begin{multline*}
            \sum_{\la^{(j)},\la^{(j+1)}} s_{\la^{(j)}}(\b^{[1,j]}) s_{\la^{(j)}/\mu^{(j)}}(\a^j) T_{\mu^{(j)},\la^{(j+1)}}(\b^{j+1}) \SP_{\la^{(j+1)}}(\a^{[j+1,k]}) 
            \\
            = H(\a^j;\b^{[1,j]})\, H(\a^{[j+1,k]};\b^{j+1})\, G(\b^{j+1})\cdot s_{\mu^{(j)}}(\b^{[1,j]})\,\SP_{\mu^{(j)}}(\a^{[j+1,k]}),
        \end{multline*}
        which follows by the skew Cauchy identity~\eqref{eq:skew cauchy id} and \Cref{gen_cauchy_littlewood_prop_2}. This completes the proof.
    \end{proof}

\subsection{Symplectic Schur process with Laurent variable specializations.}\label{sec:special_cases}

In this subsection, we describe a class of Laurent variable specializations for which the symplectic Schur process from \Cref{input_process} is an honest probability measure (and not just a signed measure).

Let $k\in\Z_{\ge 1}$ and let $A_1,B_1,\dots,A_k,B_k\in\Z_{\ge 0}$.
Let $\x^{1,\pm},\dots,\x^{k,\pm}$, $\y^1,\dots,\y^k$ be the following $(2k)$ sequences with finitely many real variables:
\begin{equation}\label{xy_seqs}
\x^{j,\pm} := (x^j_1,(x^j_1)^{-1},\dots, x^j_{A_j}, (x^j_{A_j})^{-1}),\quad \y^j := (y^j_1,\dots,y^j_{B_j}),\text{ for }j=1,\dots,k.
\end{equation}
We assume that all the variables $x^j_i$ are nonzero.
If some of the quantities $A_j$ or $B_j$ is zero, our convention is that the corresponding specialization is empty.
It is convenient to denote
\begin{equation}\label{xpm_y}
\vec{\x}^\pm := (\x^{1,\pm},\dots,\x^{k,\pm}),\qquad\vec{\y} := (\y^1,\dots,\y^k).
\end{equation}

    \begin{df}[Symplectic Schur process with variable specializations] \label{def:ssp with laurent variable spec}
        Let $k\in\Z_{\ge 1}$ and let $\vec{\x}=(\x^{1,\pm},\dots,\x^{k,\pm}),\,\vec{\y}=(\y^1,\dots,\y^k)$ be finite sequences as in \eqref{xy_seqs}, such that 
        \begin{equation} \label{eq:assumption Ak}
            A_k \ge B_1+\cdots +B_k
        \end{equation}
        and
        \begin{equation} \label{eq:conditions_variable_specs}
            x^q_i, y^p_j > 0,\quad y^p_j<\min\left( x^q_i, (x^q_i)^{-1} \right),\quad y^p_i y^p_j<1,
        \end{equation}
        for all $p,q,i,j$. The \emph{symplectic Schur process with variable specializations} is the measure on $(2k-1)$-tuples of partitions $(\lavec,\muvec) = (\la^{(1)},\dots,\la^{(k)},\mu^{(1)},\dots,\mu^{(k-1)})$, given by
        \begin{equation*}
            \Pr^\SSP\!\left( \lavec,\muvec \,\big|\, \vec{\x}^\pm,\vec{\y}\right) := \frac{ \mathcal{W} \left( \lavec,\muvec \,\big|\, \vec{\x}^\pm,\vec{\y}\right)}{ \mathcal{Z} (\vec{\x}^\pm,\vec{\y})},
        \end{equation*}
        where the weight $\mathcal{W}$ was defined in \eqref{eq: W} and the partition function $\mathcal{Z}$ is given by \eqref{eqn_partition_function}.
        The particular case where we set $A_1=\cdots=A_{k-1}=0$ will be called \emph{oscillating symplectic Schur process}.\footnote{cf.~the ascending Macdonald process from \cite{BC-2014}.}
    \end{df}
    \begin{thm} \label{thm:SSP is a probability measure}
        The symplectic Schur process with variable specializations of \Cref{def:ssp with laurent variable spec} is a \emph{probability measure} (not just a signed measure) on the set of $(2k-1)$-tuples of partitions $(\lavec,\muvec)$. Moreover with these choices of specializations, the terms $H,G$ in the partition function evaluate to
        \begin{equation}\label{eqn_HG}
            H(\x^{q,\pm};\y^p) = \prod_{i=1}^{A_q}\prod_{j=1}^{B_p} {\frac{1}{(1 - x_i^q y_j^p)(1 - (x_i^q)^{-1} y_j^p)}},\qquad G(\y^r) = \prod_{1\le i<j\le B_r}(1 - y_i^r y_j^r).
        \end{equation}
    \end{thm}
    \begin{proof}
    By inspecting the weight $ \mathcal{W} (\lavec,\muvec\mid\vec{\x}^\pm,\vec{\y})$, we see that the presence of the factors $s_{\la^{(i-1)}/\mu^{(i-1)}}(\x^{i,\pm})$ and $T_{\la^{(i)},\mu^{(i-1)}}(\y^{i})$ imply, by item (c) of \Cref{prop_orth_functions}, that
    \begin{equation*}
        \mathcal{W} (\lavec,\muvec\mid\vec{\x}^\pm,\vec{\y}) = 0 \qquad \text{if }\, \ell(\la^{(i-1)}) < \ell(\mu^{(i-1)}), \quad \text{or if }\, \ell(\la^{(i)}) > \ell(\mu^{(i-1)})+ B_i,
    \end{equation*}
    for all $i=1,\dots,k$. This shows that $\mathcal{W}(\lavec,\muvec\mid\vec{\x}^\pm,\vec{\y})$ vanishes unless $\ell(\la^{(i)}) \le \ell(\mu^{(i-1)})+B_i \le \ell(\la^{(i-1)})+B_i$, and inductively
    \begin{equation}\label{claim_weight}
        \mathcal{W} (\lavec,\muvec\mid\vec{\x}^\pm,\vec{\y})\ne 0 \,\text{ implies that }\, \ell(\la^{(i)})\le B_1+\dots+B_i,\text{ for all }i=1,\dots, k.
    \end{equation} 
    We can then assume that $\ell(\la^{(k)})\le B_1+\dots+B_k \le A_k$, using condition \eqref{eq:assumption Ak}. Under this assumption, the universal symplectic character reduces to the symplectic Schur polynomial
    \begin{equation*}
        \SP_{\la^{(k)}}(\x^{k,\pm}) = \sp_{\la^{(k)}}\big( x^k_1,(x^k_1)^{-1},\dots,x^k_{A_k},(x^k_{A_k})^{-1} \big),
    \end{equation*}
    which is positive when $x^k_1,\dots,x^k_{A_k}>0$, as discussed in \Cref{exam_var_spec_2}. Since all specializations $\x^{1,\pm},$$\dots,$$\x^{k-1,\pm},$
    $\y^1,$ $\dots,$$\y^k$ are Schur-positive, by \Cref{exam_var_spec} and \Cref{exam_var_spec_2}, we have $\mathcal{W} (\lavec,\muvec\mid\vec{\x}^\pm,\vec{\y})\ge 0$, for all $(\lavec,\muvec)$. Finally, under the conditions \eqref{eq:conditions_variable_specs}, the partition function $\mathcal{Z}$ is absolutely convergent and it can be computed through \Cref{partition_function_1} and \Cref{rem:evaluation H and G}; its evaluation coincides with \eqref{eqn_HG}.
    This completes the proof.
\end{proof}

\begin{rem}\label{rem:possible negativity}
The condition $A_k\ge B_1+\cdots+B_k$ in \Cref{def:ssp with laurent variable spec} is subtle.
It ensures that $\ell(\la^{(k)})$ is at most the number of variables in the specialization $\x^{k,\pm}$, thus implying that $\SP_{\la^{(k)}}(\x^{k,\pm})=\sp_{\la^{(k)}}(\x^{k,\pm}) > 0$, whenever $\Pr^\SSP(\lavec,\muvec \mid\vec{\x}^\pm,\vec{\y}) \neq 0$. On the other hand, if $A_k<B_1+\cdots+B_k$, then the weights $\mathcal{W} (\lavec,\muvec\mid\vec{\x}^\pm,\vec{\y})$ might be negative for certain $(\lavec,\muvec)$.
\end{rem}

\subsection{Correlation functions}

In this subsection, partitions will be regarded as infinite monotone decreasing sequences of nonnegative integers with only finitely many being nonzero, e.g.~a partition of $10$ is $(5,3,1,1,0,0,\cdots)$ and the empty partition is $(0,0,0,\cdots)$.

We will consider the most general symplectic Schur process on $(2k-1)$-tuples of partitions $(\lavec,\muvec) = (\la^{(1)},\dots,\la^{(k)},\mu^{(1)},\dots,\mu^{(k-1)})$, taking values on $\Lambda_{\a^1}\otimes\cdots\otimes\Lambda_{\a^k}\otimes\Lambda_{\b^1}\otimes\cdots\otimes\Lambda_{\b^k}=\Lambda^{\otimes (2k)}$, given by
\begin{equation*}
\Pr^\SSP\!\left(\lavec,\muvec \,\big|\, \veca,\vecb \right) := \frac{1}{ \mathcal{Z} (\vec{\a},\vec{\b})}\cdot s_{\la^{(1)}}(\b^1)\cdot \prod_{j=1}^{k-1}{\left[ s_{\la^{(j)}/\mu^{(j)}}(\a^j)\, T_{\mu^{(j)},\la^{(j+1)}}(\b^{j+1}) \right]}\cdot\SP_{\la^{(k)}}(\a^k),
\end{equation*}
where $ \mathcal{Z} (\vec{\a},\vec{\b})=\prod_{1\le p\le q\le k}{H(\a^q;\b^p)}\cdot\prod_{r=1}^k{G(\b^r)}$. 

Let us consider the alphabet $\{1, 1', \dots, (k-1), (k-1)', k\}$ and for any $(2k-1)$-tuple of partitions $(\lavec,\muvec)$ as above, assign the following point configuration (subset) of $\{1,1',\dots,(k-1),(k-1)',k\}\times\Z$:
\begin{multline}\label{L_configuration}
\mathcal{L}(\lavec,\muvec) := 
\left\{ \left( 1, \la^{(1)}_i - i \right) \right\}_{i\ge 1}\cup
\left\{ \left( 1', \mu^{(1)}_i - i \right) \right\}_{i\ge 1}\cup
\ \cdots\\
\cdots\ \cup
\left\{ \left( (k-1)', \mu^{(k-1)}_i - i \right) \right\}_{i\ge 1}\cup
\left\{ \left( k, \la^{(k)}_i - i \right) \right\}_{i\ge 1}.
\end{multline}

\begin{thm}[The symplectic Schur process is determinantal]\label{thm:det}
For any collection of $S\ge 1$ distinct points $(i_1,u_1),\dots,(i_S,u_S)$ of $\{1,1',\dots,(k-1),(k-1)',k\}\times\Z$, we have the following identity
\begin{equation}\label{determinantal_property}
\sum_{\{(i_1,u_1),\dots,(i_S,u_S)\}\subset\mathcal{L}(\lavec,\muvec)}{ \Pr^\SSP\!\left(\lavec,\muvec \,\big|\, \veca,\vecb \right) }
= \det_{1\le s,t\le S}\left[ K^{\SSP}(i_s,u_s;i_t,u_t) \right],
\end{equation}
where
\begin{multline}\label{final_kernel}
    K^\SSP(i,u;j,v) =\\
    \frac{1}{(2\pi\mathrm{i})^2}\oint\oint \frac{H(\a^{[j,k]};z) H(\b^{[1,i]};w)H(\b^{[1,i]};w^{-1})}{H(\b^{[1,j]};z)H(\b^{[1,j]};z^{-1})H(\a^{[i,k]};w)} \frac{(1-w^2)}{(1-zw)(1-z^{-1}w)} \frac{\dd z\dd w}{z^{v+2}w^{-u}},
\end{multline}
whenever $i,j\in\{1,\dots,k\}$.
The $z$- and $w$-contours are simple closed counterclockwise circles centered at the origin. If $i\le j$, the radii satisfy $|w|<\min\{ |z|,|z|^{-1} \}$, so that we can expand $\frac{1}{(1-zw)(1-z^{-1}w)}=\sum_{m,n\ge 0}{z^{m-n}w^{m+n}}$; if $i>j$, then $|z|<|w|<|z|^{-1}$, so that $\frac{1}{(1-zw)(1-z^{-1}w)} = -\sum_{m,n\ge 0}{z^{m+n+1}w^{m-n-1}}$. Functions $H$ in the integrand are defined in \eqref{eq: H}.

The formula for $K^{\SSP}(i,u;j',v)$, where $i\in\{1,\dots,k\}$ and $j'\in\{1',\dots,(k-1)'\}$, is similar except that one should replace $\a^{[i,k]}\mapsto\a^{[i+1,k]}$ in the integrand; the contours are the same as in the $i,j \in \{1,\dots,k\}$ case.
For $K^{\SSP}(i',u;j,v)$, replace $\a^{[j,k]}\mapsto\a^{[j+1,k]}$; the contours satisfy $|w|<\min\{|z|,|z|^{-1}\}$ if $i<j$, and $|z|<|w|<|z|^{-1}$ if $i\ge j$.
Finally, for $K^{\SSP}(i',u;j',v)$, replace both $\a^{[i,k]}\mapsto\a^{[i+1,k]}$ and $\a^{[j,k]}\mapsto\a^{[j+1,k]}$; the contours are the same as in the $i,j \in \{1,\dots,k\}$ case.
\end{thm}

\begin{rem}\label{rem:expand}
The equations \eqref{determinantal_property} and \eqref{final_kernel} should be interpreted as equalities between formal power series with values in $\Lambda_{\a^1}\otimes\cdots\otimes\Lambda_{\a^k}\otimes\Lambda_{\b^1}\otimes\cdots\otimes\Lambda_{\b^k}$.
For the right hand side of \eqref{final_kernel}, expand the integrand around $z=0,w=0$, e.g.~by using that
\begin{equation*}
H(\a^{[j,k]};z) = \exp\left\{ \sum_{n=0}^\infty \sum_{\ell=j}^k \frac{p_n(\a^\ell)}{n}z^n \right\},
\end{equation*}
and similarly for the other five $H$ terms; then pick up the residue at $z=w=0$ to obtain $K^{\SSP}(i,u; j,v)$.
\end{rem}

\begin{rem}
The identities \eqref{determinantal_property} and \eqref{final_kernel} turn into numerical equalities for appropriate choices of specializations $\a^1,\dots,\b^k$ and contours $z,w$.
For example, one possible set of admissible conditions, inspired by~\cite[Thm.~1]{B-2020}, is the following. Let $r_\a,r_\b>0$ be such that $\min(1,1/r_\a)>r_\b>0$. Moreover, to ensure convergence of all terms of the integrand in \eqref{final_kernel}, assume that
\begin{equation*}
\frac{p_n(\a^\ell)}{n}=O(r_\a^n),\qquad \frac{p_n(\b^\ell)}{n}=O(r_\b^n),\text{ as }n\to\infty,
\end{equation*}
for all $\ell=1,\dots,k$. Then if $i\le j$, take contours in \eqref{final_kernel} satisfying
\begin{equation*}
\min\{ 1/|w|,\, 1/r_\a,\, 1/r_\b \} > |z| > |w| > r_\b.
\end{equation*}
Otherwise, if $i>j$, take contours with
\begin{equation*}
\min\{ 1/|w|,\, 1/r_\a,\, 1/r_\b \} > |w| > |z| > r_\b.
\end{equation*}

If, additionally, the specializations $\a^1,\dots,\b^k$ are such that $\Pr^\SSP(\,\cdot\!\mid\vec{\a},\vec{\b})$ is a valid probability measure, i.e.~$\Pr^\SSP(\lavec,\muvec\mid\vec{\a},\vec{\b})\ge 0$, for all $(\lavec,\muvec)$ (see \Cref{sec:special_cases} for examples of such specializations), then the theorem implies that the point process on $\{1,1',\dots,(k-1),(k-1)',k\}\times\Z$ deriving from $\Pr^\SSP(\,\cdot\mid\vec{\a},\vec{\b})$, via the map $(\lavec,\muvec)\mapsto\mathcal{L}(\lavec,\muvec)$ in \eqref{L_configuration}, is a \emph{determinantal point process}; see e.g.~\cite{S-2000} and references therein. The Schur measure~\cite{O-2001}, Schur process~\cite{OR-2003}, as well as the symplectic and orthogonal Schur measures~\cite{B-2018},\cite{B-2020}, all lead to determinantal point processes, too.
\end{rem}

\begin{rem}
When $k=1$, our theorem matches the double contour integral formula for the symplectic Schur measures obtained in \cite[Thm.~1]{B-2020}, up to the deterministic shifts $u\mapsto u-\frac{1}{2}$ and $v\mapsto v-\frac{1}{2}$.
\end{rem}

\begin{proof}[Proof of \Cref{thm:det}]
This proof is an adaptation of the proof that the Schur process is determinantal, given in \cite[Thm.~2.2]{BR-2005}.

\smallskip

{\bf Step 1: Reformulation of the problem.}
It will suffice to prove~\eqref{determinantal_property} when $\a^k$ is a Laurent variable specialization with a finite (but arbitrarily large) number of variables $x_1^\pm,\dots,x_p^\pm$ and $\b^1$ is the variable specialization corresponding to $y_1,\dots,y_p$.
Making these reductions will allow us to apply the Eynard-Mehta theorem in Step 2 below.
Also, do the following change of coordinates for the $(2k-1)$-tuples of partitions $(\lavec,\muvec)$:
\begin{equation}\label{eq_coords}
\begin{gathered}
l^{(r)}_i := \la^{(r)}_i - i,\quad r=1,\dots,k,\ i\ge 1;\\
m^{(r)}_i := \mu^{(r)}_i - i,\quad r=1,\dots,(k-1),\ i\ge 1.
\end{gathered}
\end{equation}

If $p\ge\ell(\la^{(k)})$, then the factor $\SP_{\la^{(k)}}(\a^k)$ turns into
\begin{equation*}
\sp_{\la^{(k)}}(x_1^\pm,\dots,x_p^\pm) =
\frac{\det_{1\le i,j\le p}\Big(x_i^{l_j^{(k)}+p+1} - x_i^{-l_j^{(k)}-p-1}\Big)}
{\prod_{1\le i<j\le p}(x_i-x_j)(1-x_i^{-1}x_j^{-1})\prod_{i=1}^p(x_i - x_i^{-1})}
\end{equation*}
and $s_{\la^{(1)}}(\b^1)$ becomes
\[
s_{\la^{(1)}}(y_1,\dots,y_p) = \frac{\det_{1\le i,j\le p}\Big( y_i^{l_j^{(1)}+p} \Big)}{\prod_{1\le i<j\le p}(y_i-y_j)} = \frac{\prod_{i=1}^p{y_i^p}\cdot\det_{1\le i,j\le p}\Big( y_i^{l_j^{(1)}} \Big)}{\prod_{1\le i<j\le p}(y_i-y_j)}.
\]
Moreover, all $s_{\la^{(r)}/\mu^{(r)}}(\a^r)$ and $T_{\mu^{(r)},\la^{(r+1)}}(\b^{r+1})$ are functions of the coordinates~\eqref{eq_coords}.
Indeed, by the Jacobi-Trudi formula, $s_{\la^{(r)}/\mu^{(r)}}(\a^r)=\det_{1\le i,j\le N}{V_r(l_i^{(r)}, m_j^{(r)})}$, as long as $N\ge\max\{\ell(\la^{(r)}),\ell(\mu^{(r)})\}$, where
\begin{equation}\label{V_form}
V_r(l,m) := h_{l-m}(\a^r).
\end{equation}
By \Cref{combinatorial_S_star}, we have $T_{\mu^{(r)},\la^{(r+1)}} = \sum_{\nu}{ s_{\mu^{(r)}/\nu}\, s_{\la^{(r+1)}/\nu} }$. Then by the Jacobi-Trudi and the Cauchy-Binet formulas, we deduce $T_{\mu^{(r)},\la^{(r+1)}}(\b^{r+1})=\det_{1\le i,j\le N}{U_r(m_i^{(r)},l_j^{(r+1)})}$, where $N\ge\max\{\ell(\mu^{(r)}),\ell(\la^{(r+1)})\}$, and $[U_r(x,y)]_{x,y\in\Z}$ is the Toeplitz matrix with symbol
\begin{equation}\label{U_form}
\sum_{m\in\Z}{U_r(x+m,x)z^m} = H(\b^{r+1};z)H(\b^{r+1};z^{-1}).
\end{equation}
(That is, $[U_r(x,y)]_{x,y\in\Z}$ is an infinite matrix such that $U_r(x,y) = U_r(x+1,y+1)$, for all $x,y\in\Z$ and with coefficients $U_r(x+m,x)=U_r(m,0)$ defined by the generating function \eqref{U_form}.)
Hence, for any $(\lavec,\muvec)$ we can write (set the previous size $N$ of the matrices equal to $p$)
\[
\Pr^\SSP(\lavec,\muvec\mid\vec{\a},\vec{\b}) = \frac{\mathcal{W} (\lavec,\muvec\mid\vec{\a},\vec{\b})}{ \mathcal{Z} (\vec{\a},\vec{\b})},
\]
where $ \mathcal{Z} (\vec{\a},\vec{\b})=\prod_{1\le i\le j\le k}{H(\a^j,\b^i)}\cdot\prod_{r=1}^k{G(\b^r)}$ and 
\begin{multline}\label{W_dets}
\mathcal{W} (\lavec,\muvec\mid\vec{\a},\vec{\b}) = \frac{\prod_{i=1}^p{y_i^p}}{\prod_{1\le i<j\le p}(x_i-x_j)(1-x_i^{-1}x_j^{-1})(y_i-y_j) \prod_{i=1}^p(x_i - x_i^{-1})}
\times\det_{1\le i,j\le p}\Big( y_i^{l_j^{(1)}} \Big)\\
\times\prod_{r=1}^{k-1}{ \det_{1\le i,j\le p}{V_r(l_i^{(r)}, m_j^{(r)})}
\!\det_{1\le i,j\le p}{U_r(m_i^{(r)}, l_j^{(r+1)})} }\times
\det_{1\le i,j\le p}\Big(x_i^{l_j^{(k)}+p+1} - x_i^{-l_j^{(k)}-p-1}\Big),
\end{multline}
for any $p\ge\max\{\max_i{\ell(\la^{(i)})},\,\max_j{\ell(\mu^{(j)})}\}$.
The goal is then to prove, for any $S\ge 1$:
\[
\sum_{\{(i_1,u_1),\dots,(i_S,u_S)\}\subset\mathcal{L}(\lavec,\muvec)}{ \mathcal{W} (\lavec,\muvec\mid\vec{\a},\vec{\b}) }
= \mathcal{Z} (\vec{\a},\vec{\b})\cdot\det_{1\le s,t\le S}\left[ K^{\SSP}(i_s,u_s;i_t,u_s) \right].
\]

{\bf Step 2: Application of Eynard-Mehta theorem.}
The next step is to apply the theorem of Eynard-Mehta~\cite{EM-1998}, as stated in~\cite{BR-2005}.
The setting is the following.
Let $\Phi$ be the $p\times\Z$ matrix with entries
\[
\Phi_{i,l}:=y_i^l,\quad 1\le i\le p,\ l\in\Z;
\]
let $W_1,W_2,\dots,W_{2k-2}$ be $\Z\times\Z$ matrices with entries
\begin{equation}\label{W_form}
(W_{2s-1})_{l,m} := V_s(l,m),\qquad (W_{2s})_{m,l} := U_s(m,l),\qquad\text{for all $m,l\in\Z,\ 1\le s\le k-1$};
\end{equation}
and let $\Psi$ be the $\Z\times p$ matrix with entries
\[
\Psi_{l,j}:=x_j^{l+p+1}-x_j^{-l-p-1},\quad 1\le j\le p,\ l\in\Z.
\]
Further, let $M$ be the product $M:=\Phi\cdot W_1\cdots W_{2k-2}\cdot\Psi$.
Then \cite[Lem.~1.3]{BR-2005} shows that
\begin{equation*}
    \sum \mathcal{W}(\lavec,\muvec\mid\a,\b) = \det M\cdot
    \frac{\prod_{i=1}^p{y_i^p}}{\prod_{1\le i<j\le p}(x_i-x_j)(1-x_i^{-1}x_j^{-1})(y_i-y_j) \prod_{i=1}^p(x_i - x_i^{-1})}.
\end{equation*}

On the other hand, \Cref{partition_function_1} provides a closed formula for the partition function $\mathcal{Z}(\a,\b)$. By comparing these results:
\begin{multline}\label{detM_eq1}
\frac{\prod_{i=1}^p{y_i^p}}{\prod_{1\le i<j\le p}(x_i-x_j)(1-x_i^{-1}x_j^{-1})(y_i-y_j)
\prod_{i=1}^p(x_i - x_i^{-1})}\det M = \prod_{1\le i\le j\le k}{H(\a^j;\b^i)} \cdot\prod_{r=1}^k{G(\b^r)}\\
= \prod_{i=1}^p{H(\a^{[1,k-1]};y_i) H(\b^{[2,k]};x_i) H(\b^{[2,k]};x_i^{-1}) }\\
\times\frac{\prod_{1\le i<j\le p}(1-y_iy_j)\prod_{2\le i\le j\le k-1}{H(\a^j;\b^i)}\prod_{r=2}^k{G(\b^r)}}{\prod_{i,j=1}^p{(1-x_iy_j)(1-x_i^{-1}y_j)}}.
\end{multline}

Further, \cite[Thm.~1.4]{BR-2005} shows that the process with formal weights $\mathcal{W}(\lavec,\muvec \mid \a,\b)$ is a determinantal point process, and therefore so is the measure $\Pr(\lavec,\muvec \mid \a,\b)$. That theorem also gives a precise formula for the correlation kernel. For that, we need the following notations:
\[
W_{[a,b)} := \begin{cases}
W_a W_{a+1}\cdots W_{b-1},&\text{ if }a<b,\\
\mathbf{Id},&\text{ if }a=b,\\
0,&\text{ if }a>b,
\end{cases}
\qquad
\accentset{\circ}{W}_{[a,b)} := \begin{cases}
W_a W_{a+1}\cdots W_{b-1},&\text{ if }a<b,\\
0,&\text{ if }a\ge b,
\end{cases}
\]
Then the correlation kernel will be a $(2k-1)\times(2k-1)$ block matrix with rows and columns for the blocks labeled by $1,1',\dots,(k-1),(k-1)',k$, and the blocks are given by:
\begin{equation}\label{K_cases}
\begin{aligned}
K_{i,j} &= W_{[2i-1,2k-1)}\Psi M^{-1}\Phi W_{[1,2j-1)} - \accentset{\circ}{W}_{[2i-1,2j-1)},\\
K_{i,j'} &= W_{[2i-1,2k-1)}\Psi M^{-1}\Phi W_{[1,2j)} - \accentset{\circ}{W}_{[2i-1,2j)},\\
K_{i',j} &= W_{[2i,2k-1)}\Psi M^{-1}\Phi W_{[1,2j-1)} - \accentset{\circ}{W}_{[2i,2j-1)},\\
K_{i',j'} &= W_{[2i,2k-1)}\Psi M^{-1}\Phi W_{[1,2j)} - \accentset{\circ}{W}_{[2i,2j)},
\end{aligned}
\end{equation}
for all $i,j\in\{1,2,\dots,k\}$ and $i',j'\in\{1',\dots,(k-1)'\}$.\footnote{A word of caution: the results employed here were stated in \cite[Thm.~1.4]{BR-2005} when $\Z$ is replaced by a finite set $\mathfrak{X}$.
However, as in the proof of \cite[Thm.~2.2]{BR-2005} (this is the theorem that computes the correlation kernel of the Schur process), we can overcome this difficulty by focusing on terms $\mathcal{W}(\lavec,\muvec \mid \a,\b)$ of small degree and restricting partitions $\la^{(r)},\mu^{(r)}$ to have bounded $\ell(\la^{(r)}),\ell(\mu^{(r)}),\la^{(r)}_1,\mu^{(r)}_1$; for example, equality~\eqref{detM_eq1} is true only up to terms of high degree. As the details are well-explained in the previous reference, we ignore them here and focus on the underlying calculations.}

\medskip
{\bf Step 3: Double contour integral representation of the correlation kernel.}
Equation~\eqref{detM_eq1} can be used to calculate the determinant of the matrix $M$ without $a$-th row and $b$-th column, denoted $M\left( \substack{1\cdots\hat{a}\cdots p\\1\cdots\hat{b}\cdots p} \right)$, except that we need to remove the variables $y_a$ and $x_b$ from that formula:
\begin{equation*}
\begin{gathered}
\frac{\prod_{i\ne a}{y_i^p}}{\prod_{1\le i<j\le p,\,i,j\ne b}(x_i-x_j)(1-x_i^{-1}x_j^{-1})\prod_{1\le i<j\le p,\,i,j\ne a}(y_i-y_j)
\prod_{i\ne b}(x_i - x_i^{-1})}\det M\left( \substack{1\cdots\hat{a}\cdots p\\1\cdots\hat{b}\cdots p} \right)\\
=\prod_{i=1}^p{H(\a^{[1,k-1]};y_i) H(\b^{[2,k]};x_i) H(\b^{[2,k]};x_i^{-1}) }
\frac{\prod_{1\le i<j\le p}(1-y_iy_j)\prod_{2\le i\le j\le k-1}{H(\a^j,\b^i)}}{\prod_{i,j=1}^p{(1-x_iy_j)(1-x_i^{-1}y_j)}}\\
\times\frac{\prod_{r=2}^k{G(\b^r)}\,\prod_{i=1}^p{(1-x_by_i)(1-x_b^{-1}y_i)(1-x_iy_a)(1-x_i^{-1}y_a)}}
{H(\a^{[1,k-1]};y_a)H(\b^{[2,k]};x_b)H(\b^{[2,k]};x_b^{-1})(1-x_by_a)(1-x_b^{-1}y_a)\prod_{j\ne a}{(1-y_jy_a)}}.
\end{gathered}
\end{equation*}
By dividing the previous equality by \eqref{detM_eq1}, we obtain the entries of the inverse matrix $M^{-1}$:
\begin{multline*}
(M^{-1})_{ba} = (-1)^{a+b}\,\frac{\det M\left( \substack{1\cdots\hat{a}\cdots p\\1\cdots\hat{b}\cdots p} \right)}{\det M}
= \frac{x_b^{1-p}y_a}{(x_b-x_b^{-1})\prod_{i\ne b}(1-x_ix_b^{-1})(1-x_i^{-1}x_b^{-1})}\\
\times\frac{1}{H(\a^{[1,k]};y_a)H(\b^{[1,k]};x_b)H(\b^{[1,k]};x_b^{-1})(1-x_by_a)(1-x_b^{-1}y_a)\prod_{j\ne a}{(1-y_jy_a)(1-y_jy_a^{-1})}}.
\end{multline*}
The next step is to calculate the entries of $\Psi M^{-1}\Phi$; it turns out that we can now represent these entries as formal double contour integrals:
\begin{multline}\label{pre_contour}
(\Psi M^{-1}\Phi)_{uv} =  \sum_{a,b=1}^p\frac{1}{H(\a^{[1,k]};y_a)H(\b^{[1,k]};x_b)H(\b^{[1,k]};x_b^{-1})}\\
\times\frac{1}{(1-x_by_a)(1-x_b^{-1}y_a)\prod_{j\ne a}{(1-y_jy_a)(1-y_jy_a^{-1})}}\\
\times \left( \frac{x_b^{u+1}y_a^{v+1}}{(1-x_b^{-2})\prod_{i\ne b}(1-x_ix_b^{-1})(1-x_i^{-1}x_b^{-1})}
+ \frac{x_b^{-u-1}y_a^{v+1}}{(1-x_b^2)\prod_{i\ne b}(1-x_ix_b)(1-x_i^{-1}x_b)} \right)\\
= \frac{1}{(2\pi\mathrm{i})^2}\oint\oint \frac{H(\a^k;z^{-1}) H(\b^1;w) H(\b^1;w^{-1})}{H(\b^{[1,k]};z) H(\b^{[1,k]};z^{-1})H(\a^{[1,k]};w)}\frac{z^u w^v(1-w^2)}{(1-zw)(1-z^{-1}w)} \dd z\dd w,
\end{multline}
where we assume that the double contour integral picks up the residues at points $(z,w)=(x_b,y_a)$ and $(z,w)=(x_b^{-1},y_a)$, but does not pick up the residues at $(z,w)=(x_b,y_a^{-1})$ and $(z,w)=(x_b^{-1},y_a^{-1})$; we also assume that the product $(1-zw)(1-z^{-1}w)$ in the denominator does not produce any residues.
With these specifications on the $z$- and $w$-contours, the last equality above follows from residue calculus; the actual existence of the contours satisfying the required conditions will be shown in Step 4 below.

All four formulas in the right sides of~\eqref{K_cases} involve $\Psi M^{-1}\Phi$.
In the remainder, we only calculate the entries $K^{\SSP}(i,u;j,v)$ of the block $K_{i,j}$, for any $i,j\in\{1,2,\dots,k\}$, as the other three block-types are similar.
From~\eqref{pre_contour}, \eqref{W_form}, \eqref{V_form} and \eqref{U_form}, we deduce:
\begin{multline}\label{first_contour}
(W_{[2i-1,2k-1)}\Psi M^{-1}\Phi W_{[1,2j-1)})_{uv}\\
= \frac{1}{(2\pi\mathrm{i})^2}\oint\oint \frac{H(\a^{[i,k]};z^{-1}) H(\b^{[1,j]};w) H(\b^{[1,j]};w^{-1})}{H(\b^{[1,i]};z) H(\b^{[1,i]};z^{-1})H(\a^{[j,k]};w)}\frac{z^u w^v(1-w^2)}{(1-zw)(1-z^{-1}w)} \dd z\dd w,
\end{multline}
where the interpretation of this integral is the same as before (in terms of picking up certain residues).

\medskip
{\bf Step 4: Analytic considerations and existence of the desired contours.}

From \eqref{K_cases}, it follows that $K_{i,j}=W_{[2i-1,2k-1)}\Psi M^{-1}\Phi W_{[1,2j-1)}$, when $i\ge j$, so the entry $K^{\SSP}(i,u;j,v)$ of the desired correlation kernel is exactly given by the double contour integral in \eqref{first_contour}.
By making a simple change of variable $z\mapsto z^{-1}$, we conclude that:
\begin{multline}\label{second_contour}
K^{\SSP}(i,u;j,v)=\\
\frac{1}{(2\pi\mathrm{i})^2}\oint\oint \frac{H(\a^{[i,k]};z) H(\b^{[1,j]};w) H(\b^{[1,j]};w^{-1})}{H(\b^{[1,i]};z) H(\b^{[1,i]};z^{-1})H(\a^{[j,k]};w)}
\frac{(1-w^2)}{(1-zw)(1-z^{-1}w)} \frac{\dd z\dd w}{z^{u+2}w^{-v}},
\end{multline}
in the case that $i\ge j$. Note that now the poles $x_b$ and $x_b^{-1}$ should fall outside the $z$-contour.

In the case that the variables $x_b, y_a$ satisfy:
\begin{equation*}
|y_a| < \min\{|x_b|,|x_b|^{-1}\},\text{ for all }a,b,
\end{equation*}
we can choose the $z$- and $w$-contours appropriately. Indeed, find $R,r>0$ such that
\begin{equation}\label{condition_R}
R^{-1}>\min\left\{ \min_b|x_b|,\, \min_b|x_b|^{-1} \right\} > R > r > \max_a|y_a|,
\end{equation}
and then choose the contours $|z|=R$, $|w|=r$.
Then the $w$-contour does not enclose points $z$ or $z^{-1}$ from the $z$-contour.
Moreover, the $z$-contour excludes all $x_b,x_b^{-1}$.
Finally,~\eqref{condition_R} implies $1>R>r>|y_a|$, for all $a$, so the $w$-contour encloses all $y_a$, but none of $y_a^{-1}$.
In the case when all remaining $\a^1,\dots,\a^{k-1},\b^2,\dots,\b^k$ are variable specializations such that all $H$ terms in~\eqref{second_contour} admit the analytic expansions \eqref{eq: H}, one can alternatively calculate the integral by using the Taylor expansion $\frac{1}{(1-zw)(1-z^{-1}w)}=\sum_{m,n\ge 0}{z^{m-n}w^{m+n}}$ and then picking up the residue at $z=0,w=0$.
Since this is true for variable specializations, it is true in general.

On the other hand, when $i<j$, note that the residue of the $w$-integral in~\eqref{first_contour} at $w=z$ is
\begin{equation*}
-\frac{1}{2\pi\mathrm{i}} \oint{ H(\a^{[i,j)};z) H(\b^{(i,j]};z) H(\b^{(i,j]};z^{-1}) \frac{\dd z}{z^{u-v+1}} }
= -(\accentset{\circ}{W}_{[2i-1,2j-1)})_{u,v}.
\end{equation*}
But also, when $i<j$, we have $K_{i,j}=W_{[2i-1,2k-1)}\Psi M^{-1}\Phi W_{[1,2j-1)} - \accentset{\circ}{W}_{[2i-1,2j-1)}$, therefore $K^{\SSP}(i,u;j,v)$ equals the double contour integral in \eqref{second_contour} with the caveat that the $w$-contour must now enclose the $z$-contour.
In this case, we can choose the contours $|z|=r$, $|w|=R$, where $1>R>r>0$ satisfy~\eqref{condition_R}.
As before, when $\a^1,\dots,\a^{k-1},\b^2,\dots,\b^k$ are suitable variable specializations, we can find the value of this integral by Taylor expanding around $z=0,w=0$, using $\frac{1}{(1-zw)(1-z^{-1}w)} = -\sum_{m,n\ge 0}{z^{m+n+1}w^{m-n-1}}$, and this must be also true when $\a^i,\b^j$ are general specializations.

Hence, we have proved~\eqref{second_contour}, up to an appropriate interpretation of the contours.\footnote{The formula in the theorem statement is the conjugate kernel $\tilde{K}^\SSP(i,u;j,v) := K^{\SSP}(j,v;i,u)$, which also serves as correlation kernel for the same point process, but is slightly more suitable for our asymptotic analysis in the upcoming sections.}
We have also proved the existence of the contours satisfying our desired specifications, completing the proof of the theorem in one of the four cases.
The formulas in the other three cases from~\eqref{K_cases} are obtained in a similar fashion.
\end{proof}

\section{Dynamics of the symplectic Schur process} \label{sec:dynamics ssp}

We show that our combinatorial identities from Sections~\ref{sec:symplectic_functions} and~\ref{sec:du_schur} can be applied to construct Markov dynamics preserving the family of symplectic Schur processes. Our construction is based on~\cite{B-2011}, where similar dynamics were constructed for Schur processes; that paper attributes its main ideas to~\cite{DF-1990}. Later, in \Cref{sec:Berele} we will consider a different Markov dynamics that still preserves (a variant of) the symplectic Schur processes.

\subsection{Dynamics on single partitions}

Let $x,y,z,t$ be four Schur-positive specializations of $\Lambda$.
Let us denote the union specializations by commas, e.g.~the union of specializations $x\cup y\cup z$ will be denoted simply by $(x,y,z)$.
For any Schur-positive specialization $\rho$, also define its \emph{support} by
\begin{equation*}
\Y(\rho) := \{ \la\in\Y : s_\la(\rho) > 0\}.
\end{equation*}

\begin{df}[Transition probabilities]\label{def:single_transitions}
When $H(y;z)<\infty$, define
\begin{equation}\label{eqn:up_arrow}
q_{\mu\la}^\uparrow(y;z) := \frac{s_\la(y)s_{\la/\mu}(z)}{s_\mu(y)H(y;z)},\quad\text{for }\mu,\la\in\Y(y).
\end{equation}
Also denote by $q^\uparrow(y;z)$ the infinite $\Y(y)\times\Y(y)$ matrix with entries $q_{\mu\la}^\uparrow(y;z)$.
Likewise, define
\begin{equation}\label{eqn:sim_arrow}
q_{\mu\la}^\rcurvearrowright(y;t) := \frac{s_\la(y)T_{\la,\mu}(t)}{s_\mu(y,t)H(y;t)},\quad\text{for }\mu\in\Y(y,t),\,\la\in\Y(y),
\end{equation}
when $H(y;t)<\infty$. Also denote by $q^\rcurvearrowright(y;t)$ the $\Y(y,t)\times\Y(y)$ matrix with entries $q_{\mu\la}^\rcurvearrowright(y;t)$.
\end{df}
We point out that $q^\uparrow(y;z)$ was defined in \cite[Sec.~9]{B-2011}, but $q^\rcurvearrowright(y;t)$ is new.

Since $y,z$ are Schur-positive, then $H(y;z)=\sum_{\la\in\Y}{s_\la(y)s_\la(z)}\ge s_\varnothing(y)s_\varnothing(z)=1$, so the denominators in~\eqref{eqn:up_arrow},~\eqref{eqn:sim_arrow} are strictly positive. Moreover, by the skew Cauchy identity~\eqref{eq:skew cauchy id} and \Cref{cor:new_skew_cauchy_thm}, we deduce:
\begin{equation}\label{eqn:prob_one}
\sum_{\la\in\Y(y)}{q_{\mu\la}^\uparrow(y;z)} = 1,\text{ for all $\mu\in\Y(y)$},\qquad\qquad\sum_{\la\in\Y(y)}{q_{\mu\la}^\rcurvearrowright(y;t)} = 1,\text{ for all $\mu\in\Y(y,t)$}.
\end{equation}
By Equation~\eqref{eqn:prob_one}, both $q^\uparrow(y;z)$ and $q^\rcurvearrowright(y;t)$ are stochastic matrices that lead to Markov chains $\Y(y)\dashrightarrow\Y(y)$ and $\Y(y,t)\dashrightarrow\Y(y)$, respectively. Thus, $q^\uparrow(y;z)$ and $q^\rcurvearrowright(y;t)$ map probability measures on $\Y(y)$ and $\Y(y,t)$, respectively, to probability measures on $\Y(y)$.
The symplectic Schur measures $\M^\SS(x;y)$ and $\M^\SS(x;y,t)$ can be regarded as probability measures on $\Y(y)$ and $\Y(y,t)$ because they contain the factors $s_\la(y)$ and $s_\la(y,t)$, respectively (recall \Cref{symplectic_measure}).
Moreover,
\begin{equation}\label{eqn:preserving_schurs}
\begin{aligned}
\M^\SS(x;y)q^\uparrow(y;z) &= \M^\SS(x,z;y),\\
\M^\SS(x;y,t)q^\rcurvearrowright(y;t) &= \M^\SS(x;y),
\end{aligned}
\end{equation}
follow from \Cref{branching_SP_thm} and \Cref{gen_cauchy_littlewood_prop_2}, respectively.

\subsection{Dynamics on tuples of partitions}

The following commutativity relations turn out to be essential for the construction of dynamics on tuples of partitions.

\begin{prop}\label{prop:commutativity}
We have the following equalities between stochastic matrices:
\begin{align*}
q^\uparrow(y;z_1)q^\uparrow(y;z_2) &= q^\uparrow(y;z_2)q^\uparrow(y;z_1),&\text{($\Y(y)\times\Y(y)$ matrices),}\\
q^\rcurvearrowright(y,t;z)q^\rcurvearrowright(y;t) &= q^\rcurvearrowright(y,z;t)q^\rcurvearrowright(y;z),&\text{($\Y(y,t,z)\times\Y(y)$ matrices),}\\
q^\uparrow(y,t;z)q^\rcurvearrowright(y;t) &= q^\rcurvearrowright(y;t)q^\uparrow(y;z),&\text{($\Y(y,t)\times\Y(y)$ matrices).}
\end{align*}
\end{prop}

The first of these equalities is exactly the first item from~\cite[Prop.~18]{B-2011}.
The other two identities are new, but their proofs use the same idea as the proof of~\cite[Prop.~18]{B-2011}, so we omit them.
We simply point out that the key identities used are the branching rule for down-up Schur functions (\Cref{branching_star}) and the skew down-up Cauchy identity (\Cref{new_skew_cauchy_thm}).

Next, we aim to write down explicit formulas for the dynamics on tuples of partitions. Given Schur-positive specializations $x,y$ and fixed partitions $\mu,\nu$, define
\begin{align}
P_{x,y}(\mu,\nu \rcurvearrowright\,\uparrow \la) &:= \textrm{const}_1\cdot T_{\mu,\la}(x)\, s_{\la/\nu}(y),\label{trans_1}\\
P_{x,y}(\mu,\nu \downarrow\,\uparrow \la) &:= \textrm{const}_2\cdot s_{\mu/\la}(x)\, s_{\la/\nu}(y),\label{trans_2}\\
P_{x,y}(\mu,\nu \rcurvearrowright\,\rcurvearrowright \la) &:= \textrm{const}_3\cdot T_{\mu,\la}(x)\, T_{\nu,\la}(y),\label{trans_3}\\
P_{x,y}(\mu,\nu \downarrow\,\rcurvearrowright \la) &:= \textrm{const}_4\cdot s_{\mu/\la}(x)\, T_{\nu,\la}(y),\label{trans_4}
\end{align}
where the constants are positive real numbers chosen so that the right hand sides are probability measures on $\la$.
For this, we need to assume that the set of $\la$'s for which the right hand sides are nonzero is nonempty. Additionally, for \eqref{trans_1} and \eqref{trans_3}, we must assume $H(x;y)<\infty$ for the normalization constant to exist.
For example, by \Cref{new_skew_cauchy_thm}:
\begin{equation*}
\frac{1}{\textrm{const}_1} = \sum_\la {T_{\mu,\la}(x)\, s_{\la/\nu}(y)} = H(x,y) \sum_\rho {T_{\nu,\rho}(x)s_{\mu/\rho}(y)}.
\end{equation*}

Let $k\in\Z_{\ge 1}$ and $\a^1,\dots,\a^k,\b^1,\dots,\b^k$ be $(2k)$ Schur-positive specializations such that, if we set $\veca:=(\a^1,\dots,\a^k)$, $\vecb:=(\b^1,\dots,\b^k)$, the symplectic Schur process $\Pr^\SSP\big(\,\cdot \, \big|\, \veca,\vecb\, \big)$ is a probability measure (and not just a signed measure); see \Cref{sec:special_cases} for examples.
Given this setting, consider the set of tuples
\begin{multline}\label{def_X}
\mathcal{X}(\veca,\vecb) := \Big\{ (\lavec,\muvec) \ \big|\  \lavec = (\la^{(1)},\dots,\la^{(k)}),\, \muvec = (\mu^{(1)},\dots,\mu^{(k-1)}),\, \\
\textrm{such that } s_{\la^{(1)}}(\b^1)\cdot \prod_{j=1}^{k-1}{\left[ s_{\la^{(j)}/\mu^{(j)}}(\a^j)\, T_{\mu^{(j)},\la^{(j+1)}}(\b^{j+1}) \right]} \ne 0 \Big\}.
\end{multline}
Note that $\mathcal{X}(\veca,\vecb)$ contains the support of $\Pr^\SSP\big(\,\cdot \, \big|\, \veca,\vecb\, \big)$.

For convenience, let us omit the arrows on top of the tuples of partitions, so that elements of $\mathcal{X}(\veca,\vecb)$ will be denoted $(\la,\mu)$ and not $(\lavec,\muvec)$.
Given another Schur-positive specialization $\pi$, define the $\mathcal{X}(\veca,\vecb)\times\mathcal{X}(\veca,\vecb)$ stochastic matrix $\mathcal{Q}_\pi^\uparrow$ with entries
\begin{multline*}
\mathcal{Q}_\pi^\uparrow \Big( (\la,\mu), (\widetilde{\la},\widetilde{\mu}) \Big)
:= P_{\b^1,\pi}(\varnothing,\la^{(1)}\rcurvearrowright\,\uparrow\widetilde{\la}^{(1)})\\
\times\prod_{j=1}^{k-1}
\Big( P_{\a^j,\pi}( \widetilde{\la}^{(j)},\mu^{(j)} \downarrow\,\uparrow \widetilde{\mu}^{(j)} )\,
P_{\b^{j+1},\pi}( \widetilde{\mu}^{(j)},\la^{(j+1)} \rcurvearrowright\,\uparrow \widetilde{\la}^{(j+1)} ) \Big).
\end{multline*}
We point out that $\mathcal{Q}_\pi^\uparrow$ defines a Markov chain $\mathcal{X}(\veca,\vecb)\dashrightarrow\mathcal{X}(\veca,\vecb)$ that can be described as a sequential update: first use $\la^{(1)}$ to determine $\widetilde{\la}^{(1)}$, then use $\widetilde{\la}^{(1)}$ and $\mu^{(1)}$ to determine $\widetilde{\mu}^{(1)}$, then use $\widetilde{\mu}^{(1)}$ and $\la^{(2)}$ to determine $\widetilde{\la}^{(2)}$, and so on.

Next, let $\sigma$ be a specialization such that $\b^1=\sigma\cup\b'$, for some other specialization $\b'$, that will be denoted $\b'=:\b^1\setminus\sigma$.
Let us assume that both $\sigma$ and $\b^1\setminus\sigma$ are Schur-positive; this can be arranged, for example, by choosing $\sigma=\b^1$, since then $\b^1\setminus\sigma=\varnothing$ is the empty specialization.
Define also ${}^\#\vecb:=(\b^1\setminus\sigma,\b^2,\dots,\b^k)$ and $\mathcal{X}(\veca,{}^\#\vecb)$ via Equation~\eqref{def_X} with the obvious modification: $\b^1\setminus\sigma$ should replace $\b^1$.
Then define the $\mathcal{X}(\veca,\vecb)\times\mathcal{X}(\veca,{}^\#\vecb)$ stochastic kernel $\mathcal{Q}_\sigma^\rcurvearrowright$ with entries
\begin{multline*}
\mathcal{Q}_\sigma^\rcurvearrowright \Big( (\la,\mu), (\widetilde{\la},\widetilde{\mu}) \Big) := P_{\b^1\setminus\sigma,\,\sigma}(\varnothing,\la^{(1)}\rcurvearrowright\,\rcurvearrowright\widetilde{\la}^{(1)})\\
\times\prod_{j=1}^{k-1}
\Big( P_{\a^j,\sigma}( \widetilde{\la}^{(j)},\mu^{(j)} \downarrow\,\rcurvearrowright \widetilde{\mu}^{(j)} )\,
P_{\b^{j+1},\sigma}( \widetilde{\mu}^{(j)},\la^{(j+1)} \rcurvearrowright\,\rcurvearrowright \widetilde{\la}^{(j+1)} ) \Big).
\end{multline*}

\begin{thm}[Dynamics preserving symplectic Schur processes]
Assume that $k\in\Z_{\ge 1}$ and $\a^1,\dots,\a^k,\b^1,\dots,\b^k$ are as above; also set $\veca=(\a^1,\dots,\a^k)$, $\vecb=(\b^1,\dots,\b^k)$.
Let $\pi,\sigma,\b^1\setminus\sigma$ be three other specializations such that $\b^1=\sigma\cup(\b^1\setminus\sigma)$.

\begin{enumerate}[label=(\alph*)]

    \item Set $\veca^{\#} := (\a^1,\dots,\a^{k-1},\a^k\cup\pi)$.
    If $\pi$ is Schur-positive, then
    \begin{equation*}
    \Pr^\SSP\big(\,\cdot \, \big|\, \veca,\vecb\, \big)\, \mathcal{Q}_\pi^\uparrow
    = \Pr^\SSP\big(\,\cdot \, \big|\, \veca^{\#},\vecb\, \big).
    \end{equation*}

    \item Set
    ${}^{\#}\!\vecb:=(\b^1\setminus\sigma,\,\b^2,\dots,\b^k)$.
    If $\sigma$ and $\b^1\setminus\sigma$ are Schur-positive, then
    \begin{equation*}
    \Pr^\SSP\big(\,\cdot \, \big|\, \veca,\vecb\, \big)\, \mathcal{Q}_\sigma^\rcurvearrowright
    = \Pr^\SSP\big(\,\cdot \, \big|\, \veca,{}^\#\!\vecb\, \big).
    \end{equation*}

\end{enumerate}

\end{thm}

The previous theorem is similar to \cite[Thms.~10--11]{B-2011}, so we do not present a proof here.
However, we comment that it is based on a general formalism stated in \cite[Prop.~17]{B-2011} (originating from ideas in \cite{DF-1990} and further generalized in \cite{BF-2014}) that pastes together the transition probabilities $q^\uparrow_{\mu\la}$, $q^\rcurvearrowright_{\mu\la}$ on single partitions to construct the transition probabilities $\mathcal{Q}^\uparrow_\pi\big( (\la,\mu), (\widetilde{\la},\widetilde{\mu}) \big)$, $\mathcal{Q}^\rcurvearrowright_\sigma\big( (\la,\mu), (\widetilde{\la},\widetilde{\mu}) \big)$ on tuples of partitions.
This general construction is possible whenever we have commutativity relations as in \Cref{prop:commutativity}.
In our particular case, the result follows because the probabilities from Equations~\eqref{trans_1}--\eqref{trans_4}, that are constituents of $\mathcal{Q}^\uparrow_\pi$ and $\mathcal{Q}^\rcurvearrowright_\sigma$, are certain products of the quantities $q^\uparrow_{\mu\la}$, $q^\rcurvearrowright_{\mu\la}$, while simultaneously the symplectic Schur process turns out to be distributed like a random trajectory starting from the symplectic Schur measure and fueled by the transition probabilities $q^\uparrow_{\mu\la}$ and $q^\rcurvearrowright_{\mu\la}$:
\begin{equation*}
\Pr^\SSP\!\Big(\lavec,\muvec \ \Big|\ \veca,\vecb \Big) = \M^\SS\Big( \la^{(k)} \,\Big|\, \a^{k}, 
\b^{[1,k]} \Big)\prod_{j=1}^{k-1}
\Big( q^\rcurvearrowright_{\la^{(j+1)}\mu^{(j)}}\big( \b^{[1,j]};\b^{j+1} \big)\,
q^\uparrow_{\mu^{(j)}\la^{(j)}}\big( \b^{[1,j]};\a^j \big) \Big).
\end{equation*}

The stochastic matrices $\mathcal{Q}_\pi^\uparrow$ can lead to growth processes and perfect sampling algorithms for some instances of the symplectic Schur process, see e.g.~\cite{B-2011,BG-2009}.

\section{Sampling the symplectic Schur process through the Berele insertion} \label{sec:Berele}

\subsection{Symplectic tableaux}

For any $n\in\Z_{\ge 1}$, we introduce the alphabet of $2n$ ordered symbols
\begin{equation*}
\mathsf{A}_n:=\{1 < \overline{1} < \cdots < n < \overline{n} \}.
\end{equation*}
The set of finite length words in this alphabet is denoted by $\mathsf{A}_n^*$, while the set of weakly increasing words $w=a_1 \dots a_k$, i.e.~such that $a_i \le a_{i+1}$ for all $1\le i <k$, is denoted by $\mathsf{A}_n^{\uparrow}$. Given a word $w \in \mathsf{A}_n^*$ we define its length $\ell(w)$ to be the number of its letters and its weight $\mathrm{wt}(w)$ to be the vector in $\Z^n$ with entries
\begin{equation*}
    \mathrm{wt}(w)_i := m_i(w) - m_{\overline{i}}(w), \qquad \text{for } i=1,\dots, n,
\end{equation*}
where $m_j(w)$ is the multiplicity of the letter $j$ in $w$.

A \emph{symplectic tableau} \cite{King75} $P$ of shape $\la$ is a filling of the Young diagram of $\la$ (which is a partition) with entries in $\mathsf{A}_n$ such that
\begin{enumerate}
    \item entries of $P$ are strictly increasing column-wise and weakly increasing row-wise;
    \item entries at row $r$ are all greater or equal than $r$, for all $1\le r \le n$.
\end{enumerate}
At times the second condition above is referred to as the \emph{symplectic condition}. We denote the set of symplectic tableaux with entries in $\mathsf{A}_n$ and shape $\la$ by $\mathrm{SpTab}_n(\la)$. The weight of a symplectic tableau $\mathrm{wt}(P)$ is the vector in $\Z^n$ defined by
\begin{equation*}
    \mathrm{wt}(P)_i := m_i(P) - m_{\overline{i}}(P), \qquad \text{for } i=1,\dots, n,
\end{equation*}
where $m_j(P)$ is the multiplicity of the letter $j$ in $P$. When $\ell(\la)\le n$, the symplectic Schur polynomial $\sp_\la(x_1^\pm,\dots,x_n^\pm)$ given in \Cref{def:schur poly} turns out to be the generating function of symplectic tableaux (see \cite{King75}):
\begin{equation}\label{sp_tableaux}
\sp_\la(x_1^\pm,\dots,x_n^\pm) = \sum_{P\in\,\mathrm{SpTab}_n(\la)} x^{\mathrm{wt}(P)},
\end{equation}
where $x^{\mathrm{wt}(P)}$ stands for the monomial $x_1^{\mathrm{wt}(P)_1}\cdots\, x_n^{\mathrm{wt}(P)_n}$.

\begin{rem}
Symplectic tableaux belonging to $\mathrm{SpTab}_n(\la)$ are in bijection with half-triangular interlacing arrays with $2n$ levels and top level being $\la$ (recall that half-triangular interlacing arrays or symplectic Gelfand-Tsetlin patterns were discussed in \Cref{subs:probabilistic interpretation}). In fact, this bijection associates $P\in\mathrm{SpTab}_n(\la)$ to the array
    \[
        \left( \mathsf{x}_{k,i} : 1\le i \le \lceil k/ 2\rceil \le n \right),
    \]
defined by
    \[
        \mathsf{x}_{k,i} := 
        \begin{cases}
            \#\,\text{of cells of $P$ at row $i$ with label $\le (k+1)/2$}, \qquad &\text{if $k$ is odd},
            \\
            \#\,\text{of cells of $P$ at row $i$ with label $\le \overline{k/2}$}, \qquad &\text{if $k$ is even}.
        \end{cases}
    \]
Then, in particular, it is clear that $\mathsf{x}_{2n,i}=\la_i$.
The array presented in the left panel of \Cref{fig:dynamics example} corresponds to the tableau
\begin{equation*}
    \ytableausetup{centertableaux}
    \begin{ytableau}
    1 & 1 & 1 & 2 & \overline{2} & \overline{2} & 3 & \overline{3} \\
    2 & \overline{2} & \overline{2} & 3 & 3 & \overline{3} \\
    3 & 3 & 3 & \overline{3}
    \end{ytableau}.
\end{equation*}
\end{rem}

\subsection{Berele insertion}

Given a symplectic tableau $P \in \mathrm{SpTab}_n(\la)$ and a letter $x\in\mathsf{A}_n$, there exists an operation, known as \emph{Berele insertion}~\cite{B-1986}, that produces a new symplectic tableau $P'$ and we will denote it using the notation $P'=(P\xleftarrow[]{\mathcal{B}}x)$. The Berele insertion is described by the following algorithm.
\begin{enumerate}
    \item Set $x_1=x$ and $r=1$.
    \item Find the smallest label in row $r$ of $P$ to be strictly greater than $x_r$. If such an element exists, call it $x_{r+1}$. If such an element does not exist, then create a new cell at the right of row $r$ of $P$ and allocate $x_r$ there, calling the resulting tableau $P'$.
    \item If $x_{r+1}$ is different than $\overline{r}$, then replace in row $r$ of $P$ the leftmost occurrence of $x_{r+1}$ with $x_r$ and repeat step (2) with $r\to r+1$.
    \item If $x_{r+1}=\overline{r}$, then necessarily $x_{r}=r$ and in such case relabel the leftmost $\overline{r}$-cell in row $r$ of $P$ with the symbol $\star$.
    \item If the tableau $P$ has a $\star$-cell, call its location $(i,j)$. If $(i,j)$ is a corner cell, meaning that both cells $(i+1,j)$ and $(i,j+1)$ fall outside of the shape of $P$, then erase cell $(i,j)$ and call the resulting tableau $P'$.
    \item If $(i,j)$ is not a corner cell and $P(i+1,j) \le P(i,j+1)$, then swap the labels of $P$ at cells $(i,j)$ and $(i+1,j)$. Vice versa, if $P(i+1,j) > P(i,j+1)$, then swap the labels of $P$ at cells $(i,j)$ and $(i,j+1)$. Here, by convention we give cells outside of the shape of $P$ the label $+\infty$. Finished this step, go back to step (5).
\end{enumerate}
It is clear that the previous algorithm terminates at some point either at some iteration of step (2) or at some iteration of step (5). In the latter case we say that the Berele insertion \emph{causes a cancellation}, since the shape of $P'$ will differ from that of $P$ by the removal of a corner cell. The procedure described in steps (5) and (6) of dragging out a $\star$ cell from a tableau takes the name of \emph{jeu de taquin extraction} \cite{Schutzemberger77}, \cite[Def.~3.7.2]{Sagan_book}.

Here is an example of a Berele insertion that does not cause a cancellation

\begin{equation*}
    \includegraphics[width=0.7\linewidth]{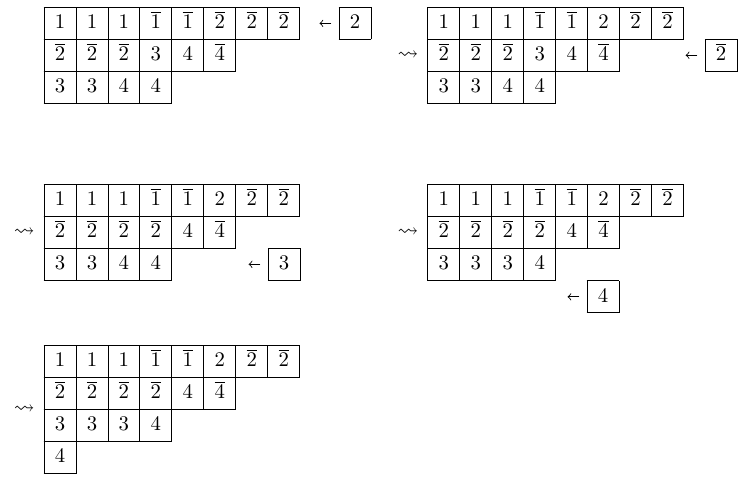}
\end{equation*}
while the next is an example of a Berele insertion that does cause a cancellation
\begin{equation*}
    \includegraphics[width=0.7\linewidth]{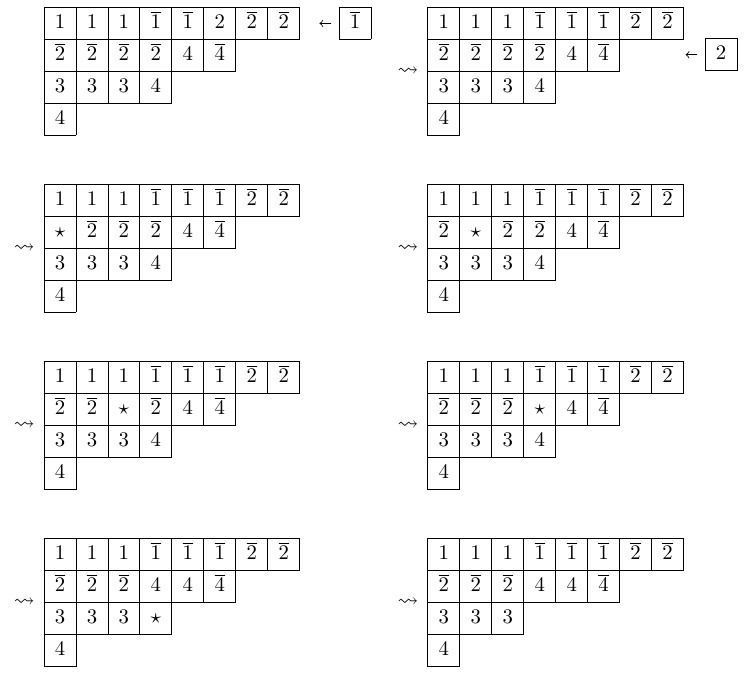}
\end{equation*}

We can extend the notion of Berele insertion to words. For a word $w = a_1 \dots a_k \in \mathsf{A}_n^*$ and a symplectic tableau $P$ we define
\begin{equation*}
    P\xleftarrow[]{\mathcal{B}}w := ( \,\cdots\, (P\xleftarrow[]{\mathcal{B}} a_1 ) \,\cdots\, ) \xleftarrow[]{\mathcal{B}} a_k.
\end{equation*}

\subsection{Bijective proof of Corollary~\ref{gen_cauchy_littlewood_prop}}

In this subsection, we present a proof of Equation~\eqref{eqn:generalization_CL_identity} by means of the Berele insertion algorithm: the argument essentially follows the work of Sundaram \cite{S-1990}.
By the branching rule of the down-up Schur functions (\Cref{branching_star}), we only need to verify the identity \eqref{eqn:generalization_CL_identity} in the simplest case where the alphabet $y$ consists of a single variable denoted by $y$. In this case, the desired identity is
\begin{equation*}
\sp_\la(x^\pm)\cdot\frac{1}{\prod_{i=1}^n(1-x_i y) (1-x_i^{-1}y )} = \sum_{\rho\,:\,\ell(\rho)\le\ell(\la)+1} \sp_{\rho}(x^\pm) T_{\rho,\la}(y) ,
\end{equation*}
where $\ell(\la)\le n-1$ and $T_{\rho,\la}$ is, by \Cref{combinatorial_S_star},
\begin{equation*}
T_{\rho,\la}(y) = \sum_{\kappa} s_{\la/\kappa}(y) s_{\rho/\kappa}(y)
= \sum_{\kappa\,:\,\la/\kappa,\,\rho/\kappa \ \text{hor.~strips}} y^{|\la/\kappa| + |\rho/\kappa|}.
\end{equation*}
In one line, the desired identity is therefore
\begin{equation}\label{eq:CI_simplest}
\sp_\la(x^\pm)\cdot\frac{1}{\prod_{i=1}^n(1-x_i y) (1-x_i^{-1}y )}
= \sum_{\rho,\kappa\,:\,\la/\kappa,\,\rho/\kappa \ \text{hor.~strips}} {\sp_\rho(x^\pm)\,y^{|\la/\kappa| + |\rho/\kappa|}},
\end{equation}
for any $\la$ with $\ell(\la)\le n-1$.
To produce a bijective proof of \eqref{eq:CI_simplest}, we regard the left hand side as the generating function of pairs $(P,w)$, where $P$ is a symplectic tableau of shape $\la$ and $w\in \mathsf{A}_n^\uparrow$ is an increasing word coming from the expansion
\[
    \frac{1}{\prod_{i=1}^n(1-x_i y) (1-x_i^{-1}y )} = \sum_{m_1, m_{\overline{1}} \ge 0 } \cdots \sum_{ m_n, m_{\overline{n}} \ge 0} \prod_{i=1}^n (x_i y)^{m_i } (x_i^{-1} y)^{ m_{\overline{i}} } = \sum_{w \in \mathsf{A}_n^\uparrow} \prod_{i=1}^n x_i^{\mathrm{wt}_i(w)}\cdot y^{\ell(w)}.
\]
On the other hand, the right hand side can be regarded as the generating function of triples $(P',\la/\kappa,\rho/\kappa)$, where $P'$ is a symplectic tableau of some shape $\rho$ and $\la/\kappa,\rho/\kappa$ are horizontal strips, for some $\kappa$. (Recall that a skew Young diagram $\la/\kappa$ is a horizontal strip if $\kappa$ and $\lambda$ interlace, i.e. $\lambda_1 \ge \kappa_1 \ge \lambda_2 \ge \kappa_2 \ge \cdots$.)
Using the Berele insertion algorithm we obtain the following result that, together with~\eqref{sp_tableaux}, immediately proves the desired Equation~\eqref{eq:CI_simplest}.

\begin{prop}\label{prop:bijection}
Given any partition $\la$ with $\ell(\la)\le n-1$, there exists a bijection
\begin{equation}\label{eqn:bijection}
(P,w) \leftrightarrow (P', \kappa, \rho)
\end{equation}
between the following sets of tuples. On the left side of~\eqref{eqn:bijection}, we have a symplectic tableau $P\in \mathrm{SpTab}_n(\la)$ and a weakly increasing word $w\in\mathsf{A}_n^{\uparrow}$. On the right side, we have two partitions $\kappa,\rho$ such that $\la/\kappa,\,\rho/\kappa$ are both horizontal strips and $P'\in\mathrm{SpTab}_n(\rho)$ is of shape $\mathrm{sh}(P')=\rho$.
Moreover,
\begin{equation}\label{restrictions_w}
\mathrm{wt}(P) + \mathrm{wt}(w) = \mathrm{wt}(P'),\qquad\ell(w)=|\la/\kappa|+|\rho/\kappa|,
\end{equation}
and $P'=(P\xleftarrow[]{\mathcal{B}}w)$ is the result of applying the Berele insertion to insert $w$ to $P$.
\end{prop}

\begin{cor}\label{cor:bounded_length}
If $\ell(\la)\le n-1$, $P\in\mathrm{SpTab}_n(\la)$, $w\in\mathsf{A}_n^\uparrow$ and $P'=(P\xleftarrow[]{\mathcal{B}}w)$, then the shape $\mathrm{sh}(P')=\rho$ satisfies $\ell(\rho)\le\ell(\la)+1$.
\end{cor}

The proof of the proposition is based on the next technical statements about the Berele insertion.

\begin{lem}\label{lem:properties_Berele_deletion}
    Let $x\le x'\in \mathsf{A}_n$ and let $P$ be a symplectic tableau. Call $P'=(P\xleftarrow[]{\mathcal{B}}x)$ and  assume that the insertion $P' \xleftarrow[]{\mathcal{B}}x'$ causes a cancellation. Then
    \begin{enumerate}[label=(\arabic*)]
        \item the insertion $P\xleftarrow[]{\mathcal{B}}x$ also causes a cancellation;
        \item call $\gamma$ and $\gamma'$ the jeu de taquin paths followed by the $*$-cells during the insertion of $x$ and $x'$ respectively in $P$ and $P'$. Then $\gamma'$ lies weakly southwest of $\gamma$ and the endpoint of $\gamma'$ lies strictly to the west of the endpoint of $\gamma$. 
    \end{enumerate}
\end{lem}

\begin{proof}
The statement (1) is proven in \cite[Lem.~3.2]{S-1990}, while (2) is proven in \cite[Lem.~3.3]{S-1990}.
\end{proof}

\begin{lem}\label{lem:properties_Berele_creation}
    Let $x\le x'\in \mathsf{A}_n$ and let $P$ be a symplectic tableau. Call $P'= (P\xleftarrow[]{\mathcal{B}}x)$ and assume that the insertion $P \xleftarrow[]{\mathcal{B}}x$ does not cause cancellations. Then
    \begin{enumerate}[label=(\arabic*)]
        \item the insertion $P'\xleftarrow[]{\mathcal{B}}x'$ also does not cause a cancellation;
        \item call $p$ and $p'$ the cells added during the insertion of $x$ and $x'$ respectively in $P$ and $P'$. Then $p'$ lies strictly east of $p$. 
    \end{enumerate}
\end{lem}

\begin{proof}
The statement (1) follows from \Cref{lem:properties_Berele_deletion} (1). To prove the statement (2) notice that if a Berele insertion does not cause a cancellation then it is simply a Schensted insertion. For the Schensted insertion, it is a known fact (see e.g.~\cite[Lem.~7.11.2]{S-1999}) that the consecutive insertion of letters in increasing order changes the shape of the tableau by adding cells in rows of decreasing order.
\end{proof}

\begin{proof}[Proof of \Cref{prop:bijection}]
Let $w=a_1\dots a_k$, so that $\ell(w)=k$.
Define $P^{(0)}:=P$ and $P^{(i)}:=(P \xleftarrow[]{\mathcal{B}}a_1\dots a_i)$, for $i=1,\dots,k$.
Denoting by $\mu^{(i)}$ the shape of $P^{(i)}$ and $P':=P^{(k)}$ produces the map
\begin{equation*}
(P,w) \mapsto (P',\{ \mu^{(i)} \}_{i=0,\dots,k}),
\end{equation*}
which is invertible since the Berele insertion of any letter can be inverted if the location of the cell added or removed is known; moreover, the algorithm implies $\mathrm{wt}(P) + \mathrm{wt}(w) = \mathrm{wt}(P')$.
By Lemmas~\ref{lem:properties_Berele_deletion} and~\ref{lem:properties_Berele_creation}, the sequence $\{ \mu^{(i)} \}_{i=0,\dots,k}$ has the property that
\begin{equation*}
\la = \mu^{(0)} \supset \mu^{(1)} \supset \cdots \supset \mu^{(s)} \subset \mu^{(s+1)} \subset \cdots \subset\mu^{(k)}=:\rho,
\end{equation*}
for some $0\le s\le k$.
Moreover, for $i\le s$, $\mu^{(i-1)}/\mu^{(i)}$ consists of a single cell strictly west of $\mu^{(i-2)}/\mu^{(i-1)}$, while for $j\ge s$, $\mu^{(j+1)}/\mu^{(j)}$ consists of a single cell strictly west of $\mu^{(j+2)}/\mu^{(j+1)}$. Calling $\kappa:=\mu^{(s)}$, this implies in particular that $\la/\kappa,\,\rho/\kappa$ are horizontal strips with $|\la/\kappa|+|\rho/\kappa|=k=\ell(w)$, and also that the map $(P',\{ \mu^{(i)} \}_{i=0,\dots,k})\mapsto (P',\kappa,\rho)$ is invertible.
Hence, the composition map~$(P,w)\mapsto (P', \kappa, \rho)$, stated in~\eqref{eqn:bijection}, is injective.
The inverse map $(P', \kappa, \rho)\mapsto(P,w)$ can be constructed by applying the inverse step of the Berele insertion to $P'$ and the boxes of $\rho/\kappa$ from east to west and the boxes of $\la/\kappa$ from west to east.
The fact that the resulting $w$ is weakly increasing, i.e.~$w\in\mathsf{A}_n^\uparrow$, follows from \cite[Exer.~3.12.2]{Sagan_book} and \cite[Lemmas~3.4--3.5]{S-1990} --- these are the converses to Lemmas~\ref{lem:properties_Berele_deletion} and~\ref{lem:properties_Berele_creation}.
\end{proof}

\subsection{The Berele insertion process: sampling algorithm for oscillating symplectic Schur process} 

Let us use the Berele insertion to build up random sequences of partitions. Fix a natural number $n$ as above and for $j=1,\dots, n$, construct random words $w_j \in \mathsf{A}_n^\uparrow$ of the form
\begin{equation} \label{eq:random word}
    w_j = 1^{m_{1,j}}\,\overline{1}^{\, \overline{m}_{1,j}} \cdots\, n^{m_{n,j}}\,\overline{n}^{\,\overline{m}_{n,j}},
\end{equation}
where $m_{i,j}, \overline{m}_{i,j}$ are independent random variables sampled with the geometric laws
\begin{equation*}
\Pr\left(m_{i,j} = k\right) = (1-x_iy_j)(x_iy_j)^k,\qquad
\Pr\left(\overline{m}_{i,j} = k\right) = (1-x_i^{-1}y_j)(x_i^{-1}y_j)^k,\qquad k=0,1,\dots,
\end{equation*}
for some reals $x_1,\dots , x_n,y_1,\dots , y_n$ satisfying $0 < x_i y_j, x_i^{-1} y_j < 1$, for all $i,j$.
Once we sample the random words $w_1,\dots, w_n$, we produce the sequence of random symplectic tableaux
\begin{equation} \label{eq:construction Pj}
    P^{(0)} = \varnothing,
    \qquad \qquad
    P^{(j)} = (P^{(j-1)} \xleftarrow[]{\mathcal{B}} w_j)
\end{equation}
and we denote their shapes by $\mathrm{sh}(P^{(j)})=\la^{(j)}$.
We call this random growth process on partitions the \emph{Berele insertion process}.
As a consequence of \Cref{cor:bounded_length}, $\ell(\la^{(j)})\le j$, for all $j=1,\dots,n$, in particular $\la^{(1)}$ is always a row partition.

The next proposition describes the joint law of this sequence of partitions.

\begin{prop} \label{prop:symplectic Schur berele}
Consider the random sequence of Young diagrams $\{ \la^{(j)} \}_{j= 1,\dots ,n}$ built as described above.
Then its law is given by the oscillating symplectic Schur process from \Cref{def:ssp with laurent variable spec} with $\y^j=(y_j)$, for $1\le j\le n$, and $\x^0=x^\pm:=(x_1,x_1^{-1},\dots,x_n,x_n^{-1})$.
In other words,
\begin{equation}\label{prop:general_berele_process}
\mathbb{P} \left( \{ \la^{(j)} \}_{j= 1,\dots ,n} \right) = \prod_{i,j=1}^n (1-x_i y_j) (1-x_i^{-1} y_j) \cdot s_{\la^{(1)}}(y_1)\cdot\prod_{j=2}^{n} T_{\la^{(j-1)},\la^{(j)}} (y_j)\cdot\sp_{\la^{(n)}}(x^{\pm}).
\end{equation}
\end{prop}
\begin{proof}
Using \Cref{prop:bijection} repeatedly in the sequence of insertions~\eqref{eq:construction Pj}, we have the bijections
    \begin{equation*}
        (P^{(j-1)} ,w_j ) \leftrightarrow (P^{(j)},\kappa^{(j)},\la^{(j)})
    \end{equation*}
    so that the full Berele insertion process is encoded by the sequence of data
    \begin{equation}\label{eq:sequence P lambdas and kappa}
        P^{(n)}, \kappa^{(n)} , \la^{(n-1)}, \kappa^{(n-1)} ,\dots , \la^{(2)}, \kappa^{(2)}, \la^{(1)},\kappa^{(1)}.
    \end{equation}
    The joint probability density of the sequence of words $w_1,\dots,w_n$ is, up to normalization constant,
    \begin{equation*}
\prod_{i,j=1}^n (x_iy_j)^{m_{i,j}} (x_i^{-1} y_j)^{\overline{m}_{i,j}} = \left(\prod_{j=1}^n y_j^{\mathrm{\ell}(w_j)}  \right)\cdot\left(\prod_{i,j=1}^n x_i^{\mathrm{wt}(w_j)_i}  \right).
    \end{equation*}
The pushforward leads to the following joint probability density on the sequence~\eqref{eq:sequence P lambdas and kappa} (set $\la^{(n)}=\mathrm{sh}(P^{(n)})$)
\begin{equation}\label{density_1}
        \left( \prod_{j=1}^n y_j^{|\la^{(j-1)}/\kappa^{(j)}|+|\la^{(j)}/\kappa^{(j)}|} \right)\cdot\left(\prod_{i=1}^n x_i^{\mathrm{wt}(P^{(n)})_i}  \right),
\end{equation}
thanks to the properties of the Sundaram bijection of \Cref{prop:bijection}; see Equation~\eqref{restrictions_w}.
Note that $\la^{(0)}=\varnothing$ and $\la^{(0)}/\kappa^{(1)}$ being a horizontal strip imply that $\kappa^{(1)}=\varnothing$, therefore~\eqref{density_1} equals
	\begin{equation}\label{density_2}
		y_1^{|\la^{(1)}|}\cdot\left( \prod_{j=2}^n y_j^{|\la^{(j-1)}/\kappa^{(j)}|+|\la^{(j)}/\kappa^{(j)}|} \right)\cdot\left(\prod_{i=1}^n x_i^{\mathrm{wt}(P^{(n)})_i}  \right).
	\end{equation}
Summing over all partitions $\kappa^{(n)},\dots,\kappa^{(2)}$ and all symplectic tableaux $P^{(n)}$ of shape $\la^{(n)}$ gives the marginal of $\{\la^{(j)}\}_{j=1,\dots,n}$: this gives the desired~\eqref{prop:general_berele_process} by virtue of \Cref{combinatorial_S_star} and identity~\eqref{sp_tableaux}.
Note that $y_1^{|\la^{(1)}|}=s_{\la^{(1)}}(y_1)$, because $\la^{(1)}$ is a row partition.
Finally, the normalization constant in \eqref{prop:general_berele_process} was computed in greater generality in \Cref{partition_function_1}, \Cref{thm:SSP is a probability measure}. The proof is now complete.
\end{proof}

\begin{rem}
    A construction of random partitions based on the Berele insertion was already considered in \cite{Nte-18}. Earlier, in \cite{WW-2009,Def-2012}, the authors considered sampling of random symplectic tableaux based on Pieri rules of symplectic Schur polynomials. In these cases, the dynamics can be recovered from the one treated in the current paper through a continuous time scaling limit (although \cite{Nte-18} also defines a more general dynamics that depends on an additional parameter $q$, and specializes to the classical Berele insertion algorithm when $q=0$). In none of these articles an asymptotic analysis was performed.
\end{rem}

\begin{rem}
    The Berele insertion dynamics is different from those described in \Cref{sec:dynamics ssp}. Such difference is parallel to that between the classical RSK insertion (see \cite{BBBCCV-2014} and references therein) and the Borodin-Ferrari dynamics \cite{B-2011,BF-2014}, which are variants of Propp's domino shuffling \cite{Pro-2003}. The former, like the Berele insertion, uses less randomness, as partitions are generated deterministically from the choice of a random word of the form~\eqref{eq:random word}.
\end{rem}

\section{Asymptotics of the oscillating symplectic Schur process} \label{sec:asymptotics}

In this section, we consider the limits of a particular case of the oscillating symplectic Schur process, namely the probability measure on $k$-tuples of partitions $(\la^{(1)},\dots,\la^{(k)})$ with law
\begin{equation} \label{eq:Symplectic process berele}
    \mathbb{P} ( \{ \la^{(i)} \}_{i=1,\dots,k} ) \propto s_{\la^{(1)}}(1/2)  \cdot\prod_{i=2}^k T_{\la^{(i)},\la^{(i-1)}}(1/2) \cdot \sp_{\la^{(k)}}(1^{2n}),
\end{equation}
where $(1^{2n})=(1,1,\dots,1)$ is a collection of $2n$  variables specialized at $1$. Recall that \eqref{eq:Symplectic process berele} is a probability measure only when $k\le n$, by \Cref{rem:possible negativity}. Moreover, for $k\le n$, the measure \eqref{eq:Symplectic process berele} can be sampled through the Berele insertion, as shown in \Cref{prop:symplectic Schur berele}.
For this sampling, construct random words
\begin{equation} \label{eq:Berele_words_2}
    w_j = 1^{m_{1,j}}\,\overline{1}^{\, \overline{m}_{1,j}} \cdots\, n^{m_{n,j}}\,\overline{n}^{\,\overline{m}_{n,j}},\qquad j=1,2,\dots,k,
\end{equation}
where $m_{i,j}, \overline{m}_{i,j}$ are mutually independent random variables with geometric laws
\begin{equation*}
\Pr\left(m_{i,j} = \ell\right) = \Pr\left(\overline{m}_{i,j} = \ell\right) = 2^{-\ell-1},\qquad \ell=0,1,\dots.
\end{equation*}
Once we sample the random words $w_1,\dots w_k$, we inductively produce the sequence of random partitions starting from the empty symplectic tableau as follows:
\begin{equation*}
    P^{(0)} := \varnothing,
    \qquad \qquad
    P^{(j)} := (P^{(j-1)}\xleftarrow[]{\mathcal{B}} w_j),
	\qquad \qquad
	\la^{(j)} := \mathrm{sh}(P^{(j)}).
\end{equation*}
We will be interested in the discrete point process $\mathcal{L}(\lavec)=\{(j, \la^{(j)}_i-i): 1\le i \le n, \, 1\le j \le k \}$ on $\{1,\dots,k\}\times\Z$.
The results of a simulation for large values of $k\le n$ are shown in \Cref{fig:symplectic schur}.
Note that the Berele insertion process can be run even for steps $k>n$, but the distribution of $\{\la^{(i)}\}_{i=1,\dots,k}$ would no longer be distributed according to a symplectic Schur process.

In principle, the asymptotic analysis we perform below could be performed for more general variants of the process \eqref{eq:Symplectic process berele}, allowing several parameter to enter the game as in \eqref{prop:general_berele_process}, but calculations are already rather cumbersome, so we decided to stick with the simplest non-trivial situation.\footnote{However, we expect that the asymptotic behaviors described here are universal for random partitions sampled via the Berele insertion process with more general parameters.}

\begin{figure}
\floatbox[{\capbeside\thisfloatsetup{capbesideposition={right,top},capbesidewidth=5cm}}]{figure}[\FBwidth]
{\caption{
    (a) A sample of the point process $\mathcal{L}(\lavec)=\{(j, \la^{(j)}_i-i): 1\le i\le n, \, 1\le j\le k \}$ obtained through the Berele insertion of $k$ random words with law \eqref{eq:Berele_words_2} in the alphabet $\mathsf{A}_n$. Here we chose $n=200$, $k=300$. Red and orange curves denote respectively the upper and lower branch of the limit shape, computed in \eqref{eq:limit shape}. The verical green line has abscissa $n$; the process to its left has the law of the oscillating symplectic Schur process \eqref{eq:Symplectic process berele}.
    \\
    (b) A close-up on the top edge of the liquid region.
    \\
    (c) A close-up on the bulk of the liquid region.
    \\
    (d) A close-up of the point process around the point $(4n/5,-4n/5)$. Green squares denote the first few levels of the complemented point configuration described in \Cref{subs:GUE corner processes}, which asymptotically converges to a GUE corners process. 
    \\
    (e) Blue points represent holes of the process $\mathcal{L}(\lavec)$. Portrayed is a close-up of the proximity of the point $(9n/8,-n)$, where the lower branch of the limit shape touches the horizontal line of ordinate $-n$.
    }\label{fig:symplectic schur}}
{\includegraphics[width=11cm]{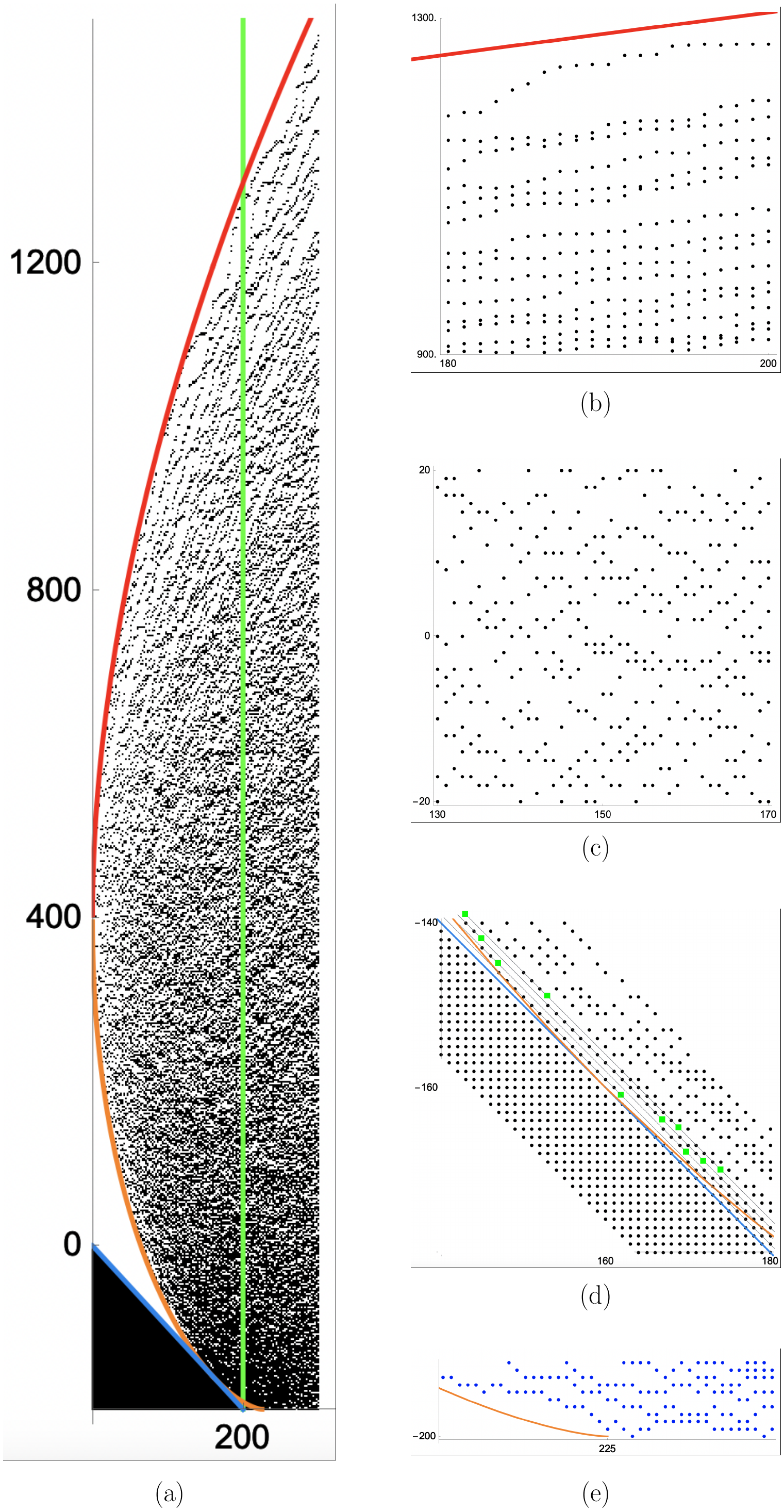}}
\end{figure}

\subsection{Limit shape}\label{subs:limit shapes}

Let us describe the macroscopic limit shape produced by the random point configuration $\mathcal{L}(\lavec)=\{(j, \la^{(j)}_i-i): 1\le i\le n, \, 1\le j\le k \}$, where $\lavec=\{\la^{(i)}\}_{i=1,\dots,k}\subset\{1,\dots,k\}\times\Z$ is distributed according to the oscillating symplectic Schur process \eqref{eq:Symplectic process berele} after rescaling microscopic coordinates $(i,j)$ as $(nx,ny)$.
First, we observe that $\la^{(i)}_j=0$ almost surely for all $n\ge j>i$, as a result of the Berele insertion mechanism. This implies that points $(i,-j)$ for $n \ge j > i$ belong almost surely to the configuration $\mathcal{L}(\lavec)$. Thus we refer to the set of macroscopic points $\{(x,y): x \in (0,1), \, y\in (-1,0), \, x+y\le 0 \}$ as the \emph{frozen region}.

To describe the non-trivial part of the limit shape, we consider the correlation kernel obtained in \Cref{thm:det}.
Up to a simple gauge transformation,\footnote{We use that if $K^{\SSP}(i,u;j,v)$ is a correlation kernel, then $(-2)^{i-j}\cdot K^{\SSP}(i,u;j,v)$ is a correlation kernel for the same determinantal point process.} the correlation kernel~\eqref{final_kernel}, with the choice of specializations prescribed by \eqref{eq:Symplectic process berele}, describing the point configuration $\mathcal{L} (\Vec{\la})$ equals
\begin{equation} \label{eq:kernel berele}
        K^{\SSP}(i,u;j,v) =
        \begin{dcases}
            \frac{1}{(2\pi \mathrm{i})^2} \oint_{\Gamma_{1-\delta}} \frac{\dd z}{ z} \oint_{\Gamma_{\frac{1}{2}+\delta}} \dd w\,  \frac{1-w^2}{(1-zw)(z-w)} \frac{G_{j,v}(z)}{G_{i,u}(w)}, \qquad &\text{if } i\le j,
            \\
            \frac{1}{(2\pi \mathrm{i})^2} \oint_{\Gamma_{\frac{1}{2}+\delta}} \frac{\dd z}{ z} \oint_{\Gamma_{1-\delta}} \dd w\,  \frac{1-w^2}{(1-zw)(z-w)} \frac{G_{j,v}(z)}{G_{i,u}(w)},
            \qquad &\text{if } i > j.
        \end{dcases}
\end{equation}
In this formula, we take any $\delta\in(0,1/4)$ and used the notations
\begin{align*}
    \Gamma_{\rho} &:= \{ z\in \C \, : \, |z|=\rho \},\nonumber\\
    G_{j,v}(z) &:= (1-z)^{-2n} [(1-z/2)(1-2z)]^{j}\, z^{-j-v}.
\end{align*}
To understand limit shapes and limiting processes arising while scaling around various regions of the random field $\mathcal{L}(\Vec{\la})$, we will perform an asymptotic analysis of the kernel $K^{\SSP}$ by employing the saddle point method, as it has been done previously in the literature~\cite{OR-2003,OR-2006,OR-2007}. Under the scaling 
\begin{equation*}
    j = nx,
    \qquad 
    v = ny, 
\end{equation*}
with $x\in (0,1), y\in\R$, and in the liquid region $x+y>0$, we observe that the function $G_{j,v}$ becomes
\begin{equation*}
    G_{j,v}(z) = \exp \left\{ n h_{x,y}(z) \right\},
\end{equation*}
where the function $h_{x,y}$ in the exponent is
\begin{equation} \label{eq:h}
    h_{x,y}(z) = -2 \log (1-z) + x \log \left(1-\frac{z}{2}\right)+x \log \left(1-2 z \right)-(x+y) \log (z).
\end{equation}
The large $n$ behavior of $K^{\SSP}$ is determined by the critical points of $h_{x,y}(z)$, i.e.~by the roots of
\begin{equation} \label{eq:h'}
    h_{x,y}'(z)=-\frac{x+y}{z}+\frac{2 x}{2 z-1}-\frac{x}{2-z}+\frac{2}{1-z},
\end{equation}
or, collecting the terms,
\begin{equation} \label{eq:h'_collected}
    h_{x,y}'(z) = \frac{z^3 (2 x-2 y-4) - z^2 (2x - 7y -10)-z (2 x+7 y+4)+2(x+y)}{(z-2) (z-1) z (2 z-1)}.
\end{equation}
Then $h_{x,y}'(z)=0$ yields a cubic equation with coefficients parameterized by $x,y$. Studying the expression \eqref{eq:h'} we see that it always has a real zero, which we denote by $z_0$, in the half line $(2,+\infty)$; such zero will always be irrelevant for our asymptotic analysis. The two remaining zeros, which we denote by $z_\pm$, can have one of the following topologies, explained by \Cref{fig:plots h'}:
\begin{enumerate}
    \item for $y \ge y_+(x)$, $z_\pm \in (1/2,1)$;
    \item for $y_-(x) < y < y_+(x)$, $z_\pm$ are complex conjugate with non-zero imaginary part;
    \item for $-x \le y \le y_-(x)$, $z_\pm \in (0,1/2)$;
    \item for $y<-x$, $z_- <0$ and $z_+ \in (0,1/2)$.
\end{enumerate}
\begin{figure}[htbp]
\centering

\begin{subfigure}[b]{0.4\linewidth}
\includegraphics[width=\linewidth]{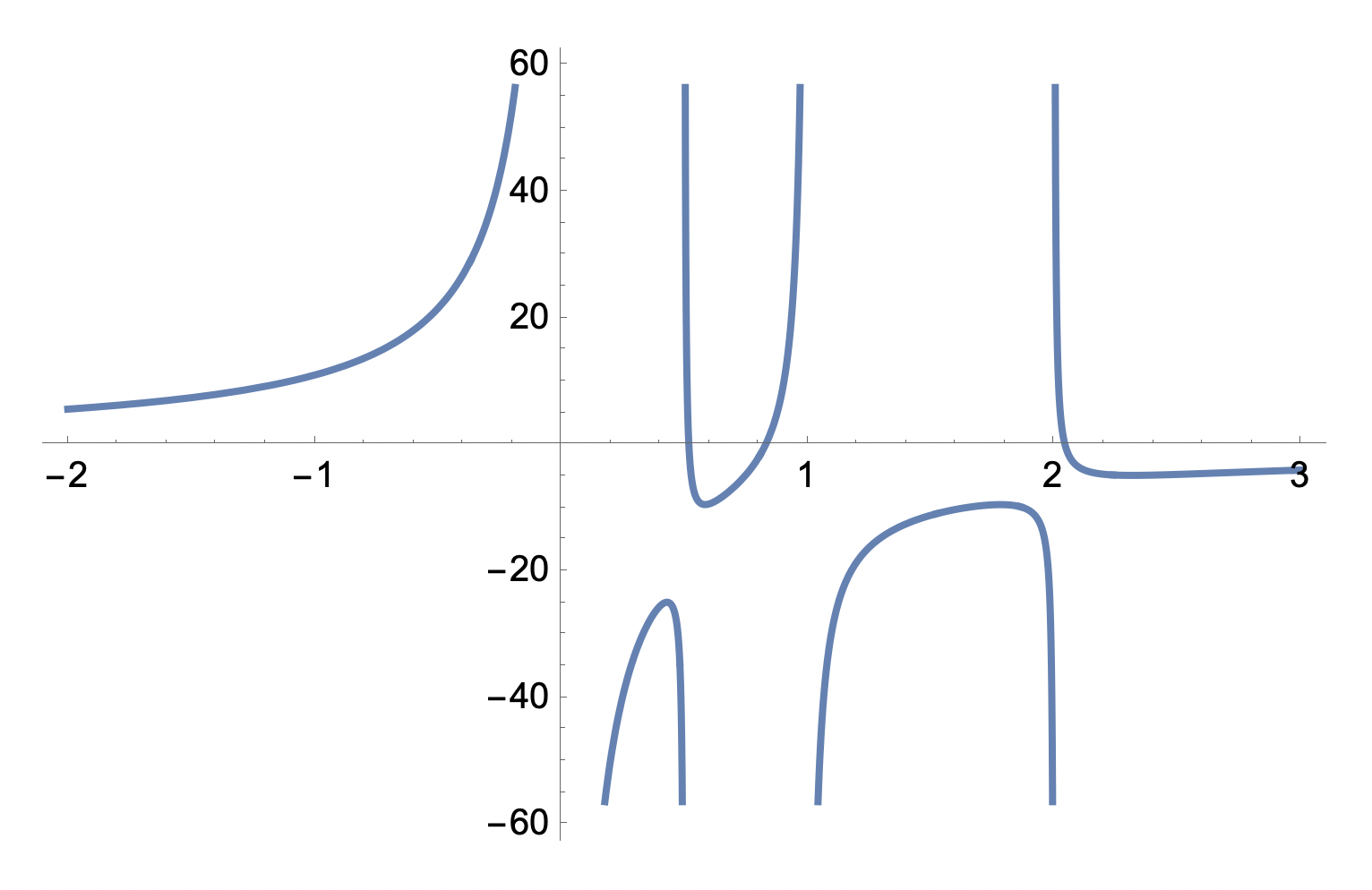}
\end{subfigure}
\hfill
\begin{subfigure}[b]{0.4\linewidth}
\includegraphics[width=\linewidth]{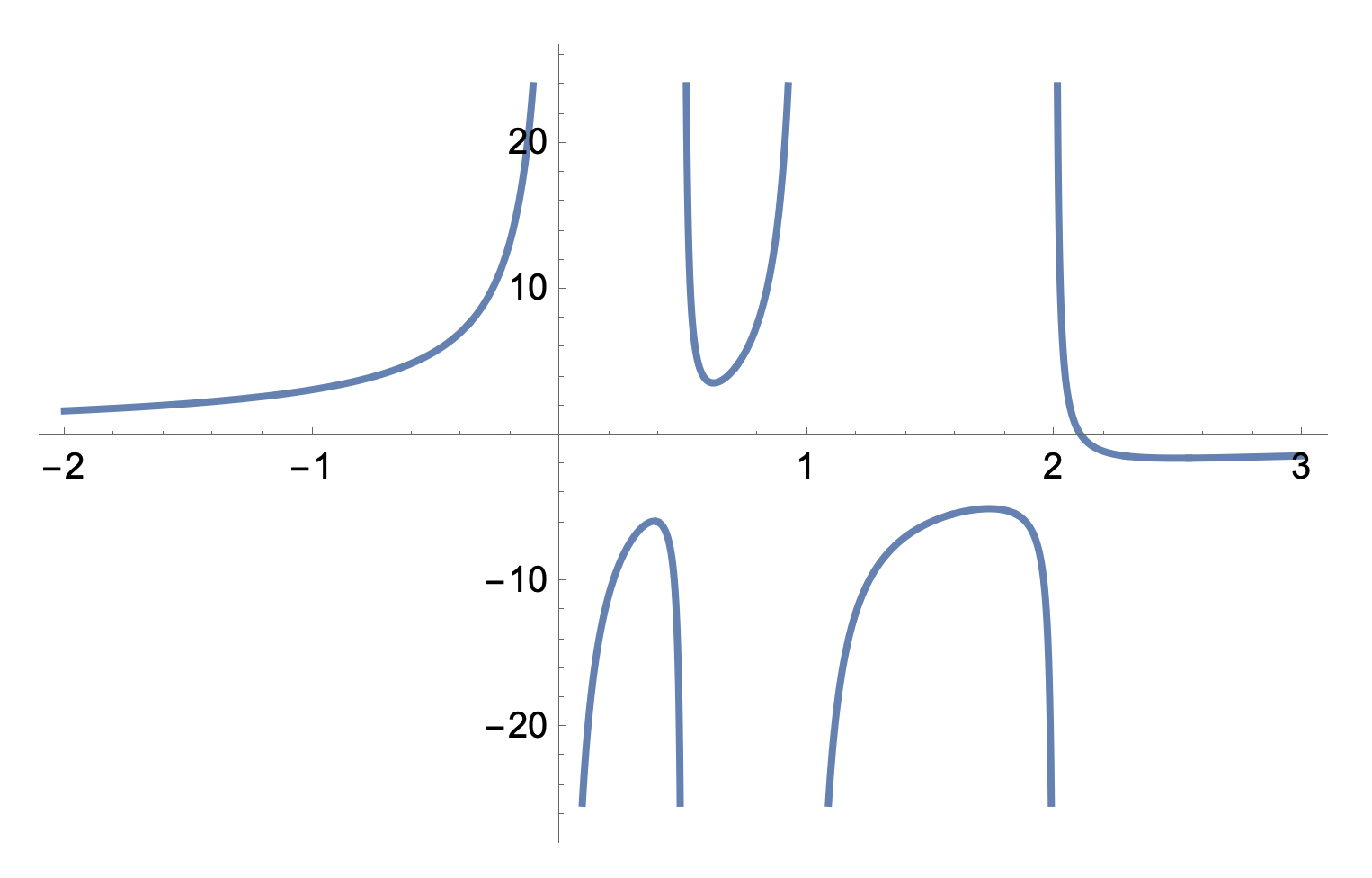}
\end{subfigure}

\begin{subfigure}[b]{0.4\linewidth}
\includegraphics[width=\linewidth]{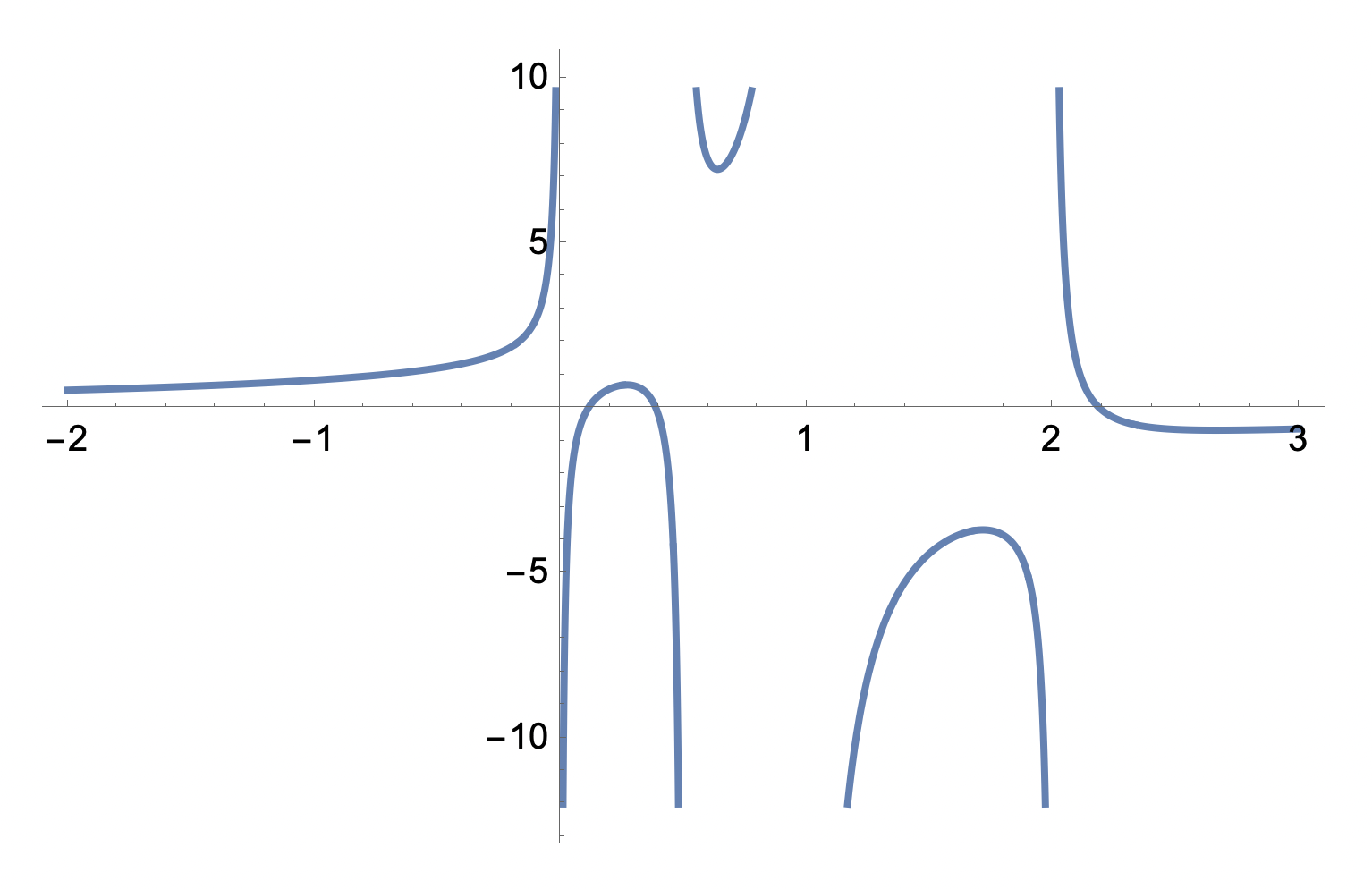}
\end{subfigure}
\hfill
\begin{subfigure}[b]{0.4\linewidth}
\includegraphics[width=\linewidth]{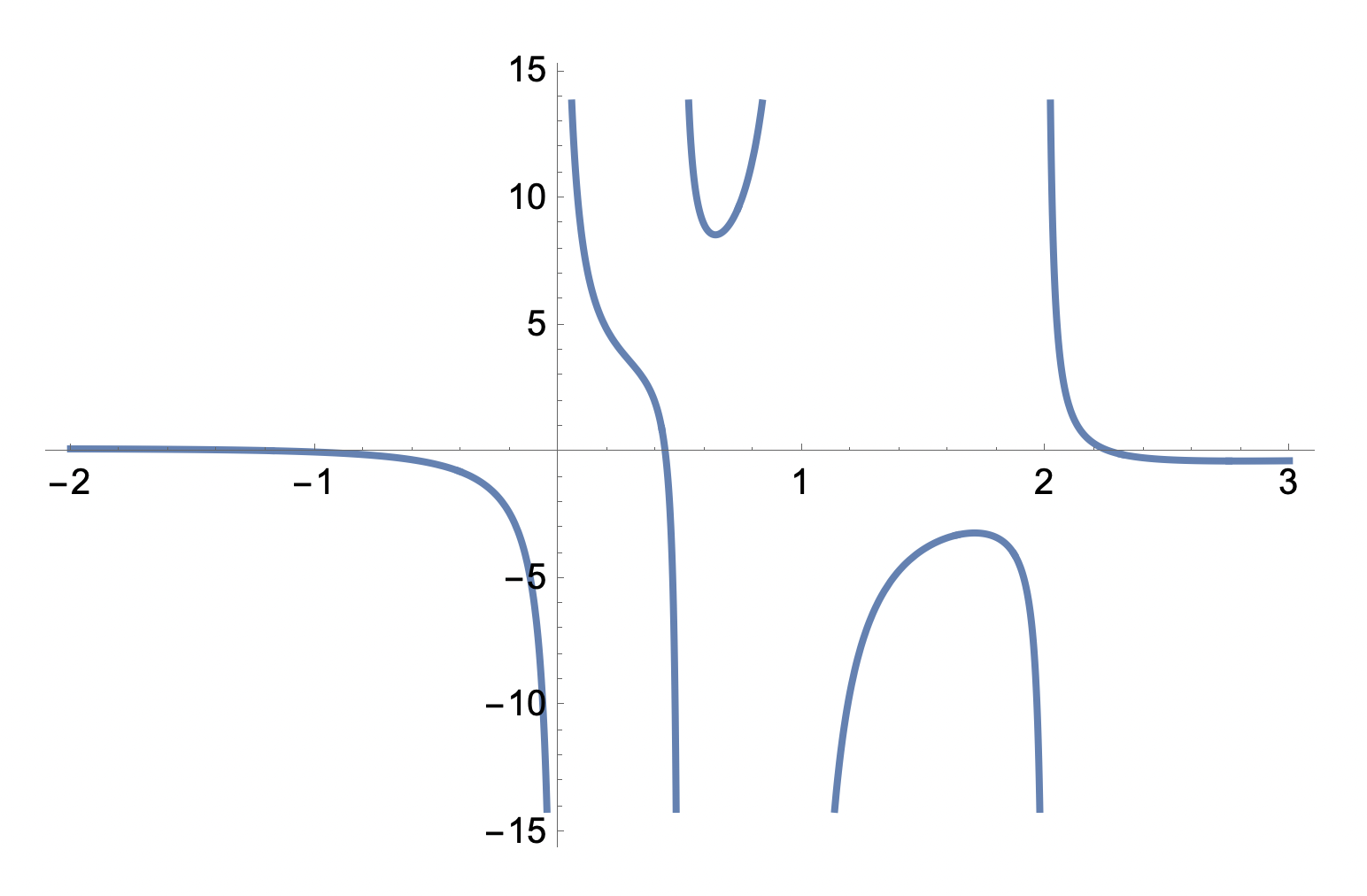}
\end{subfigure}

\caption{Plots of $h'_{x,y}(z)$ for real $z$ and different values of $x,y$. From top-left to bottom-right these topologies correspond to the cases enumerated after Equation~\eqref{eq:h'_collected}.}
\label{fig:plots h'}
\end{figure}
The functions $y_\pm(x)$ above are the unique real numbers such that there exists some $z_0\in\R$ with
\begin{equation}\label{def_yx}
h_{x,y(x)}'(z_0) = h_{x,y(x)}''(z_0) = 0.
\end{equation}
We label $y_+(x)$ the larger and $y_-(x)$ the smaller of the two, so that $y_+(x)\ge y_-(x)$.
Moreover, we label the corresponding $z_0$ in \eqref{def_yx} by $z_{\mathrm{up}}(x)$ and $z_{\mathrm{down}}(x)$, respectively.

To find explicit expressions, let us directly solve the system of equations $h_{x,y}'(z)=0$, $h_{x,y}''(z)=0$, which, after simple algebra and assuming $z\neq 0$, can be recast as the system:
\begin{gather} 
    \frac{1}{(1-z)^2}-  \frac{x }{(z-2)^2}-\frac{x}{(1-2z)^2} = 0, \label{eq:z}
    \\
    y(x,z)=  \frac{2z }{1- z}-\frac{x z}{2- z} +\frac{x}{2z-1}. \label{eq: f}
\end{gather}
The two relevant solutions of \eqref{eq:z} are given by
\begin{equation*}
	z_{\mathrm{up}}(x)=\frac{\sqrt{-18 \sqrt{x^3 (x+9)}-x (18 x+11)+20 \sqrt{x (x+9)}+36}}{10 x-8}+\frac{9 x-\sqrt{x (x+9)}-10}{10 x-8},
\end{equation*}
\begin{equation*}
	z_{\mathrm{down}}(x)=\frac{9 x+\sqrt{x (x+9)}-10}{10 x-8}-\frac{\sqrt{18 \sqrt{x^3 (x+9)}-x (18 x+11)-20 \sqrt{x (x+9)}+36}}{10x-8},
\end{equation*}
so that plugging these expressions into \eqref{eq: f}, we obtain
\begin{equation}\label{eq:limit shape}
    y_+(x)=y(x,z_{\mathrm{up}}(x)), 
    \qquad
    y_-(x)=y(x,z_{\mathrm{down}}(x)),
\end{equation}
and their plots produce the red and orange curves portrayed in \Cref{fig:symplectic schur}.
The lower branch of the limit shape $y_-$ is tangent to the line $y=-x$, which delimits the frozen region of fully packed particles. The tangency point is
\begin{equation*}
    x=4/5,
    \qquad
    y_-(4/5)= -4/5.
\end{equation*}

\subsection{Limit process in the bulk: an extended Sine process}

Fix $x\in(0,1)$ and $y\in (y_-(x),y_+(x))$.
Let $z_\pm=z_\pm(x,y)$ be the complex conjugate roots of $h_{x,y}'(z)=0$. Set  $\gamma_+(z_+)$ as a circle arc connecting $z_-$ to $z_+$ counterclockwise, crossing the segment $(1/2,2)$ and $\gamma_-(z_+)$ as a circle arc connecting $z_-$ to $z_+$ clockwise, crossing $(-\infty,1/2)$. For the following result, define the kernel $\mathcal{K}_{\mathrm{bulk}}^{(z_+)} : \Z^2\times\Z^2 \to\R$ by
\begin{equation} \label{eq:extended_Sine_kernel}
    \mathcal{K}_{\mathrm{bulk}}^{(z_+)}(\tau,\alpha;\tau', \alpha') = 
    \begin{dcases}
        \frac{1}{2\pi\mathrm{i}} \int_{\gamma_+(z_+)} (5/2-w-1/w)^{\tau'-\tau}\, w^{\alpha-\alpha'-1} \dd w,
        \qquad & \text{if } \tau \le \tau',\\
        \frac{1}{2\pi \mathrm{i}} \int_{\gamma_-(z_+)} (5/2-w-1/w)^{\tau'-\tau}\, w^{\alpha-\alpha'-1} \dd w, \qquad & \text{if } \tau > \tau'.
    \end{dcases}
\end{equation}

\begin{thm}\label{prop:limit Sine}
    Consider the scaling
    \begin{equation*}
        i(\tau) = \lfloor nx\rfloor + \tau,
        \qquad
        u(\alpha) = \lfloor ny\rfloor + \alpha.
    \end{equation*}
    Then, for any $k\in\Z_{\ge 1}$ and any pair of $k$-tuples $\tau_1,\dots,\tau_k,\alpha_1,\dots,\alpha_k\in\Z$, denoting $i_\ell=i(\tau_\ell)$, $u_\ell=u(\alpha_\ell)$, we have
    \begin{equation*}
        \lim_{n\to\infty} \,\det_{1\le\ell,\ell'\le k} \left[ K^{\SSP}(i_\ell,u_\ell;i_{\ell'},u_{\ell'}) \right] = \det_{1\le\ell,\ell'\le k} \left[\mathcal{K}_{\mathrm{bulk}}^{(z_+)}(\tau_\ell, \alpha_\ell; \tau_{\ell'}, \alpha_{\ell'} ) \right].
    \end{equation*}
\end{thm}

The proof is analogous to \cite{OR-2003} and therefore we do not provide details here.

\subsection{Density in the bulk}
    The diagonal terms of our extended sine kernel $\mathcal{K}_{\mathrm{bulk}}^{(z_+)}$ determine the density $\rho(x,y)$ of points in the process~\eqref{eq:Symplectic process berele}:
    \begin{equation}\label{eq:limiting density}
        \rho(x,y) = \mathcal{K}_{\mathrm{bulk}}^{(z_+)}(0,0;0,0) = \frac{\arg\left( z_+(x,y) \right)}{\pi}.
    \end{equation}
	Using Cardano's formula, we find $z_+(x,y)$ explicitly as the solution to the cubic equation~\eqref{eq:h'_collected} with $\Im(z_+(x,y))>0$; it is
    \begin{equation*}
        z_+(x,y) = \frac{1}{6 a(x,y)}\left( 2 b(x,y) -c(x,y) - \frac{\Delta_0(x,y)}{c(x,y)} + \mathrm{i} \sqrt{3} \left| c(x,y) - \frac{\Delta_0(x,y)}{c(x,y)} \right| \right),
    \end{equation*}
    where
    \begin{gather*}
        a(x,y) = 4-2x+2y, 
        \qquad
        b(x,y) = 10-2x+7y,
        \\
        c(x,y) = \sqrt[3]{\frac{\Delta_1(x,y) + \sqrt{\Delta_1(x,y)^2-4 \Delta_0(x,y)^3}}{2}},
    \end{gather*}
    and
    \begin{gather*}
        \Delta_0(x,y) = 16 x^2+2 x y-40 x+7 y^2+32 y+52,
        \\
        \Delta_1(x,y) = 128 x^3+24 x^2 y-264 x^2+78 x y^2-132 x y-48 x+20 y^3+276 y^2+816 y+560.
    \end{gather*}
    Within the limit shape, $c(x,y)$ is a real number, assuming $\sqrt[3]{},\sqrt[2]{}$ to be the principal branches of the cube and square root. In fact, the limit shape, in the region $y\ge -1$, is characterized by the implicit equation
    \begin{equation*}
        \Delta_1(x,y)^2-4 \Delta_0(x,y)^3 = 0.
    \end{equation*}

\subsection{Edge asymptotics: Airy process}

We can also comment on the nature of the scaled process around the limit shape. Following \cite[Eqn.~(60)]{OR-2007}, recall that the extended Airy kernel $K_{\mathrm{Airy}}:\R^2\times\R^2\to\R$ is:
\begin{equation*}
    \begin{split}
        K_{\mathrm{Airy}}(\tau,\alpha ; \tau', \alpha') := &
        -\mathbf{1}_{\tau>\tau'}\,\frac{e^{-\frac{(\alpha-\alpha')^2}{2(\tau-\tau')}}}{\sqrt{2 \pi (\tau - \tau')}} 
        \\
        & \qquad + \int_{e^{-\mathrm{i}\pi/3}\infty}^{e^{\mathrm{i}\pi/3}\infty} \frac{\dd\xi}{2\pi\mathrm{i}} \int_{e^{-\mathrm{i}2\pi/3}\infty}^{e^{\mathrm{i}2\pi/3}\infty} \frac{\dd\eta}{2\pi\mathrm{i}} \,\frac{1}{\xi - \eta}\, e^{\frac{\xi^3-\eta^3}{3} - \frac{\xi^2 \tau' - \eta^2 \tau}{2} - \xi \alpha' + \eta \alpha }.
    \end{split}
\end{equation*}
Alternatively, this kernel can be written in terms of Airy functions, hence the name; for example, see~\cite{FNH-1999,M-1994,PS-2002,J-2003}.

\begin{thm}\label{prop:limit Airy}
    Fix $x\in (0,1)$ and let $(y,z_*)$ be either $(y_-(x), z_\mathrm{down}(x))$ or $(y_+(x), z_\mathrm{up}(x))$. Define parameters 
    \begin{equation*}
        d_1 := \left( \frac{h_{x,y}'''(z_*)}{2}\right)^{\!\!\frac{1}{3}}z_*,
        \qquad
        c_1 := \frac{(2 z_*-1)^2 (z_*-2)^2}{2z_*\left(5 z_*^2-8 z_*+5\right)} \, d_1^{2},
        \qquad
        c_2 := \frac{2 (z_*^2-1)}{(z_*-2) (2 z_*-1)} \, c_1,
    \end{equation*}
    and consider the scaling
    \begin{equation*}
        i(\tau)= \left\lfloor xn +c_1 \tau n^{\frac{2}{3}} \right\rfloor,
        \qquad
        u(\tau,\alpha) = \left\lfloor yn + c_2 \tau n^{\frac{2}{3}} + d_1 \alpha n^{\frac{1}{3}} \right\rfloor.
    \end{equation*}
    Then, for any $k\in\Z_{\ge 1}$ and any pair of $k$-tuples $\tau_1,\dots,\tau_k,\alpha_1,\dots,\alpha_k \in\R$, denoting $i_\ell=i(\tau_\ell)$, $u_\ell=u(\tau_\ell,\alpha_\ell)$, we have:
    \begin{itemize}
        \item if $y=y_+(x)$ and $x\in (0,1)$ or $y=y_-(x)$ and $x\in (0,\frac{4}{5})$, then
        \begin{equation}\label{eq: limit kernel airy}
            \lim_{n\to\infty} \,\det_{1\le\ell,\ell'\le k} \left[ n^{\frac{1}{3}}|d_1|\,  K^{\SSP}(i_\ell,u_\ell;i_{\ell'},u_{\ell'}) \right] 
            = \det_{1\le\ell,\ell'\le k} \left[ K_{\mathrm{Airy}}(\tau_\ell,\alpha_\ell ; \tau_{\ell'}, \alpha_{\ell'}) \right];
        \end{equation}
        \item if $y=y_-(x)$ and $x\in (\frac{4}{5},1)$, then
        \begin{equation}\label{eq: limit kernel airy holes}
            \lim_{n\to\infty} \,\det_{1\le\ell,\ell'\le k} \left[ n^{\frac{1}{3}}|d_1|
\!\left( \delta_{i_\ell,i_{\ell'}} \delta_{u_\ell,u_{\ell'}} - K^{\SSP}(i_\ell,u_\ell;i_{\ell'},u_{\ell'}) \right) \right] 
            = \det_{1\le\ell,\ell'\le k} \left[ K_{\mathrm{Airy}}(\tau_\ell,\alpha_\ell ; \tau_{\ell'}, \alpha_{\ell'}) \right].
        \end{equation}
    \end{itemize}
\end{thm}

The proof is analogous to \cite{OR-2003}, therefore we do not provide details here.

\begin{rem}
If $x\in(0,1)$ and $(y,z_*)=(y_+(x),z_{\textrm{up}}(x))$, then $z_*\in(\frac{1}{2},1)$, $h_{x,y}'''(z_*)>0$, implying that $d_1,c_1,c_2>0$.
If $x\in (0,\frac{4}{5})$ and $(y,z_*)=(y_-(x),z_{\textrm{down}}(x))$, then $z_*\in(0,\frac{1}{2})$, $h_{x,y}'''(z_*)<0$, implying $c_1>0>d_1,c_2$.
Finally, if $x\in(\frac{4}{5},1)$ and $(y,z_*)=(y_-(x),z_{\textrm{down}}(x))$, then $z_*\in(-\frac{1}{2},0)$, $h_{x,y}'''(z_*)>0$, implying $c_2>0>d_1,c_1$.
\end{rem}

\begin{rem}
    The left hand side of \eqref{eq: limit kernel airy} differs from that of \eqref{eq: limit kernel airy holes} since in the latter, the extended Airy process arises from considering the ``holes" of the point process rather than the points themselves.
\end{rem}

\begin{rem} \label{rem:symplectic schur and ordinary schur}
The kernel $K^{\SSP}$ in \eqref{eq:kernel berele} differs from that of the ordinary Schur process \cite{O-2001} by the presence of the cross factor $(1-w^2)/(1-wz)(z-w)$, which replaces $1/(z-w)$. From the standpoint of saddle point asymptotic analysis, such factor transforms, after scaling integration variables $z=z_*+\varepsilon\xi$ and $w = z_*+\varepsilon\eta$ around a saddle point $z_*$, as
\begin{equation*}
	\frac{(1-w^2)}{(1-wz)(z-w)} =  \frac{1}{\varepsilon}\cdot\frac{1}{\xi-\eta} + O(1),\,\text{ as }\epsilon\to 0^+,
\end{equation*}
whenever $z_*\neq \pm 1$.  Therefore, for saddle points $z_* \neq \pm 1$, the asymptotic behaviour of $K^{\SSP}$, and hence of the symplectic Schur process, is equivalent to the ordinary Schur process.
\end{rem}

\begin{rem} \label{rem:singular points at z=-1}
The argument of \Cref{rem:symplectic schur and ordinary schur} holds whenever the dominant saddle point $z_*$ of the function $h_{x,y}(z)$ is different than $\pm 1$. From \eqref{eq:h'}, we see that $h_{x,y}'(z)$ always has a singularity at $z=1$ (and so $z=1$ is never a saddle point), while
\begin{equation*}
	h_{x,y}'(-1) = 1+y.
\end{equation*}
    Therefore $h_{x,y}$ has a saddle point at $z=-1$ for $y=-1$ and evaluating higher order derivatives we see
    \begin{equation}\label{eqn:higher_derivatives}
        h_{x,-1}''(-1) = -\frac{1}{2} + \frac{4}{9}x,
        \qquad
        h_{x,-1}'''(-1) = -\frac{3}{2} + \frac{4}{3}x.
        \qquad
        h_{x,-1}''''(-1) = -\frac{21}{4} + \frac{128}{27} x,
    \end{equation}
    which implies that second order singularities at $z=-1$ can only take place setting $x=9/8, y=-1$, which falls outside of the probabilistic region ($0<x \le 1$) for the symplectic Schur process. Interestingly, at $(x,y)=(9/8,-1)$, the point $z=-1$ is a singularity of fourth order, as can be seen from~\eqref{eqn:higher_derivatives}.
    The asymptotic limit of the kernel $K^{\SSP}$ is not the extended Airy kernel in this case, but instead a Pearcey-like kernel that will be reported in \Cref{sec:asymptotics at 1.125}.
\end{rem}

\subsection{Tangency points asymptotics: GUE corners process} \label{subs:GUE corner processes}

We analyze the point process in the neighborhood of the tangency points of the limit shape with the frozen regions.
Following~\cite[Eqn.~(14)]{OR-2006}, see also~\cite{J-2005,JN-2006}, define the kernel $\mathcal{K}_{\mathrm{GUE cor}}:\left( \Z_{>0} \times \R \right) \times \left( \Z_{>0} \times \R \right) \to \R$ by
\begin{equation*}
\begin{split}
    \mathcal{K}_{\mathrm{GUE cor}}(i,\alpha; j,\alpha') := &- \mathbf{1}_{i>j} \mathbf{1}_{\alpha>\alpha'} \,\frac{(\alpha-\alpha')^{i-j-1}}{(i-j-1)!}
    \\
    & \qquad
    + \int_{ \mathrm{i} \R +1} \frac{\dd\xi}{2\pi\mathrm{i}} \oint_{ \Gamma_{\frac{1}{2}} } \frac{\dd\eta}{2\pi\mathrm{i}} \,\frac{1}{\xi-\eta} \frac{\xi^j}{\eta^i}\, e^{\frac{\xi^2 - \eta^2}{2} - \alpha'\xi + \alpha\eta}.
\end{split}
\end{equation*}
The point process having correlation kernel $\mathcal{K}_{\mathrm{GUE cor}}$ is the \emph{GUE corners process}, which describes the eigenvalues of corners of an infinite Gaussian Hermitian random matrix, as in~\cite{OV-1996}.
Equivalently, this means that the projection of the kernel $\mathcal{K}_{\mathrm{GUE cor}}$ to $\{1,\dots,N\}\times\R$ describes the eigenvalues of principal minors of an $N\times N$ GUE random matrix.

\begin{thm}[GUE corners process at the tangency point $(0,2)$]\label{prop:limit GUE corner}
    Consider the scaling
    \begin{equation*}
        u(\alpha)=2n + \lfloor 2 \alpha \sqrt{n} \rfloor.
    \end{equation*}
    For any $k\in\Z_{\ge 1}$ and any pair of $k$-tuples $i_1,\dots,i_k\in\Z_{>0}$, and pairwise distinct $\alpha_1,\dots,\alpha_k\in\R$, let us denote $u_\ell=u(\alpha_\ell)$, then:
    \begin{equation*}
        \lim_{n\to\infty} \,\det_{1\le\ell,\ell'\le k} \left[ 2 \sqrt{n} \, K^{\SSP}(i_\ell, u_\ell ; i_{\ell'}, u_{\ell'} ) \right] = \det_{1\le\ell,\ell'\le k} \left[ \mathcal{K}_{\mathrm{GUE cor}}(i_\ell,\alpha_\ell; i_{\ell'},\alpha_{\ell'}) \right].
    \end{equation*}
\end{thm}

The GUE corners process can also be observed at the tangency point $(x,y)=(4/5,-4/5)$ after a certain complementation of the point configuration, which we briefly explain next.
The point process $\mathcal{L}(\lavec)$, around the microscopic location $(4n/5,-4n/5)$, tends to occupy diagonal segments of the form $(i,j), \dots, (i+k,j-k)$ with length $k = O(\sqrt{n})$.
Then, it is a certain ``complemented process'' that in the large $n$ limit converges to the GUE corners process, as depicted in \Cref{fig:symplectic schur}(d); by definition, this new process consists of the points
\[
    \text{$(i+1,j-1)$, if $(i,j)$ is occupied, but $(i+1,j-1)$ is not}
\]
(this happens when $(i,j)$ is the lower right endpoint of a filled diagonal segment) and points 
\[
    \text{$(i,j+1)$, if $(i,j)$ is occupied, but $(i-1,j+1)$ is not}
\]
(this happens when $(i,j)$ is the upper left endpoint of a filled diagonal segment).
This complementation procedure consists, in the language of lozenge tilings (which are in correspondence with interlacing partitions), in considering different types of lozenges in the neighborhood of the tangency point; see e.g.~\cite{Gor-2021}.
Finally, we point out that the point configuration $\mathcal{L}(\lavec)$ originating from the symplectic Schur process does not consist of interlacing partitions, nevertheless asymptotically, around the tangency point, this identification can be made.

\subsection{Bottom edge asymptotics at $x=9/8$.} \label{sec:asymptotics at 1.125}

In the previous subsections, we provided arguments to show that, for any $x\in(0,1)$, the symplectic Schur process converges asymptotically to a multilevel extension of the Sine process in the bulk and of the Airy process at the edge. We also gave explicit formulas for the limit shape.
In this subsection, we do the asymptotic analysis in a neighborhood of the macroscopic point $(x,y)=(9/8,-1)$.
This would correspond \emph{formally} to the region where the bottom branch of the limit shape touches the horizontal line $y=-1$; see \Cref{fig:symplectic schur}~(e).
We say formally because, as stated in \Cref{prop:symplectic Schur berele}, the symplectic Schur process~\eqref{eq:Symplectic process berele} only describes the Berele insertion process when $k \le n$. For $k>n$, corresponding to the scaled region when $x>1$, the symplectic Schur process fails to be a probability measure in general; see \Cref{rem:possible negativity}.


To state our result, we introduce the kernel $\mathcal{K}:(\R\times\R_{>0})\times(\R\times\R_{>0})\to\R$ given by
\begin{equation} \label{eq:symplectic Percey}
\begin{split}
    \mathcal{K}(\tau,\alpha;\tau',\alpha') := 
    &
    \,\mathbf{1}_{\tau<\tau'}\,
    \frac{e^{-\frac{(\alpha'-\alpha)^2}{2(\tau'-\tau)}}}{\sqrt{2\pi(\tau'-\tau)}}
    \\
    & 
    \qquad
    +
    \int_{e^{-3\mathrm{i} \pi/4} \infty}^{e^{3\mathrm{i}\pi/4} \infty} \frac{\dd\xi}{2\pi\mathrm{i}} \int_{-\mathrm{i} \infty}^{+\mathrm{i} \infty} \frac{\dd\eta}{2\pi\mathrm{i}} \,\frac{2\eta}{\eta^2-\xi^2} \,e^{ \frac{\xi^4-\eta^4}{24} + \frac{\tau'\xi^2 - \tau\eta^2}{2} - \alpha'\xi + \alpha\eta},
\end{split}
\end{equation}
where the $\xi$-contour lies in the complex half-plane $\Re(\xi) \le 0$.

\begin{thm}\label{prop:pearcey}
    Consider the scaling
    \begin{equation} \label{eq:scaling Pearcey}
        i(\tau)=\left\lfloor\frac{9n}{8} +\sqrt{\frac{27n}{64}}\,\tau\right\rfloor,
\qquad u(\alpha)= -n + \left\lfloor\left( \frac{n}{12} \right)^{\!\frac{1}{4}} \alpha\right\rfloor.
    \end{equation}
    Then, for any $k\in\Z_{\ge 1}$ and any pair of $k$-tuples $\tau_1,\dots,\tau_k\in\R$, $\alpha_1,\dots,\alpha_k\in\R_{+}$, denoting $i_\ell=i(\tau_\ell)$, $u_\ell=u(\alpha_\ell)$, we have:
    \begin{equation} \label{eq:limit kernel Pearcey}
        \lim_{n\to\infty} \,
        \det_{1\le\ell,\ell'\le k} \left[ \left( \frac{n}{12} \right)^{\!\frac{1}{4}} \left(  \delta_{i_\ell,i_{\ell'}} \delta_{u_\ell,u_{\ell'}} - K^{\SSP}(i_\ell,u_\ell;i_{\ell'},u_{\ell'}) \right) \right] = \det_{1\le\ell,\ell'\le k} \left[ \mathcal{K}(\tau_\ell, \alpha_\ell; \tau_{\ell'},  \alpha_{\ell'}) \right].
    \end{equation}
\end{thm}

Since the limit \eqref{eq:limit kernel Pearcey} appears to be new, we report its proof, although the proof just amounts to an application of the saddle point method as in the previous subsections.

\begin{proof}
    Under the scaling \eqref{eq:scaling Pearcey}, we write the kernel \eqref{eq:kernel berele}, after shifting contours and evaluating the $z$-pole at $z=w$, as
    \begin{multline}\label{eq:kernel pre limit Pearcey}
         K^{\SSP}(i(\tau),u(\alpha);\,i(\tau'),u(\alpha')) = \ \delta_{\tau,\tau'} \,\delta_{\alpha,\alpha'}+ \mathbf{1}_{\tau < \tau'}\,\frac{1}{2\pi \mathrm{i}} \oint_{\Gamma_{1-\delta}} \frac{\dd w}{w} \,\frac{G_{i(\tau'),u(\alpha')}(w)}{G_{i(\tau),u(\alpha)}(w)}
         \\
          + \frac{1}{(2\pi \mathrm{i})^2} \oint_{\Gamma_{\frac{1}{2}+\delta}} \frac{\dd z}{ z} \oint_{\Gamma_{1-\delta}} \dd w  \,\frac{1-w^2}{(1-wz)(z-w)} \frac{G_{i(\tau'),u(\alpha')}(z)}{G_{i(\tau),u(\alpha)}(w)}.
    \end{multline}
    Computations from \Cref{rem:singular points at z=-1} show that the function $h_{x,y}$ from \eqref{eq:h} has a critical point of fourth order at $x=9/8,y=-1,z=-1$:
    \begin{equation*}
	   h_{\frac{9}{8},-1}'(-1) =h_{\frac{9}{8},-1}''(-1)=h_{\frac{9}{8},-1}'''(-1)=0, \qquad h_{\frac{9}{8},-1}''''(-1) = \frac{1}{12}.
    \end{equation*}
    Considering the scaling \eqref{eq:scaling Pearcey} and
\begin{equation*}
z=-1 + \left( \frac{n}{12} \right)^{\!-\frac{1}{4}} \xi, \qquad w=-1 + \left( \frac{n}{12} \right)^{\!-\frac{1}{4}} \eta,
\end{equation*}
we have
\begin{equation*}
            G_{i(\tau'), u(\alpha')}(z)
             =\exp\left\{n h_{\frac{1}{n}i(\tau'),\frac{1}{n}u(\alpha')} (z) \right\} 
             \approx c_n(\tau',\alpha') \exp \left\{ \frac{\xi^4}{24} + \frac{\tau'\xi^2}{2} + \alpha' \xi \right\},
\end{equation*}
where the term $c_n(\tau',\alpha')$ above does not depend on $\xi$. A similar estimate holds for $G_{i(\tau), u(\alpha)}(w)$, except that $\xi,\tau',\alpha'$ must be replaced by $\eta,\tau,\alpha$, respectively. It follows that
    \begin{align*}
            \frac{G_{i(\tau'), u(\alpha')}(w)}{G_{i(\tau), u(\alpha)}(w)}
            &\approx \frac{c_n(\tau',\alpha')}{c_n(\tau,\alpha)} \, e^{ \frac{(\tau'-\tau)\eta^2}{2} + (\alpha'-\alpha)\eta },\\
            \frac{G_{i(\tau'), u(\alpha')}(z)}{G_{i(\tau), u(\alpha)}(w)}
            &\approx \frac{c_n(\tau',\alpha')}{c_n(\tau,\alpha)} \,e^{ \frac{\xi^4 - \eta^4}{24} + \frac{\tau'\xi^2 -\tau\eta^2}{2} + \alpha'\xi - \alpha\eta }.
    \end{align*}
Moreover,
    \begin{equation*}
	\frac{\dd z}{z} \approx -\left( \frac{n}{12} \right)^{\!-\frac{1}{4}} \dd\xi,\qquad
 \dd w \approx \left( \frac{n}{12} \right)^{\!-\frac{1}{4}} \dd\eta,\qquad
    \frac{1-w^2}{(1-wz)(z-w)} \approx \left( \frac{n}{12} \right)^{\!\frac{1}{4}}  \frac{2\eta}{\xi^2-\eta^2}.
    \end{equation*}
    To obtain the final answer, first deform the $z$-contour to be of steepest descent and the $w$-contour to be of steepest ascent (as shown in \Cref{fig:contour plots}), then plug all estimates above into the formula~\eqref{eq:kernel pre limit Pearcey} for the kernel $K^{\SSP}$, next remove the gauge factor $c_n(\tau',\alpha')/c_n(\tau,\alpha)$, as it disappears when one considers determinants as in~\eqref{eq:limit kernel Pearcey}, then compute the Gaussian integral resulting from the first line in~\eqref{eq:kernel pre limit Pearcey}, and finally make a change of variables $\xi\mapsto-\xi$, $\eta\mapsto-\eta$ in the integral resulting from the second line in~\eqref{eq:kernel pre limit Pearcey}.
    \begin{figure}
        \centering
        \includegraphics[width=8cm]{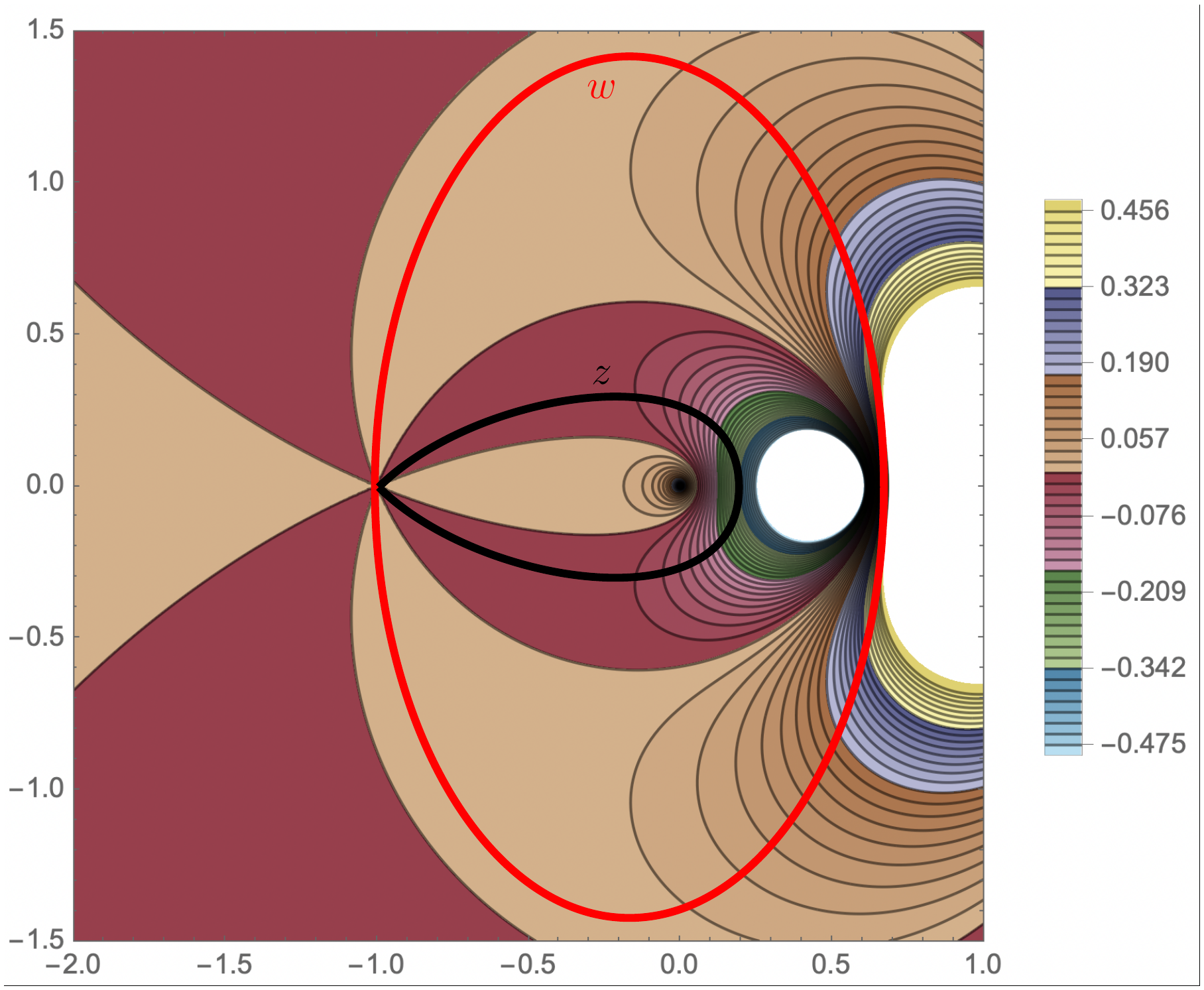}
        \caption{The contourplot of $\zeta\to\Re\left\{ h_{\frac{9}{8},-1}(\zeta) \right\} - h_{\frac{9}{8},-1}(-1)$. Overlayed black and red curves are $z$ and $w$ integration curves deformed to be respectively of steepest descent and steepest ascent for $\Re\left\{ h_{\frac{9}{8},-1}(\zeta) \right\}$. }
        \label{fig:contour plots}
    \end{figure}
\end{proof}

\begin{rem}
    The kernel $\mathcal{K}$ in \eqref{eq:symplectic Percey} resembles the Pearcey kernel from \cite[Eqn.~(72)]{OR-2007}, see also \cite{BH-1998a,BH-1998b,TW-2006}, which has the form
    \begin{equation*}
    \begin{split}
    \mathcal{K}(\tau,\alpha;\tau',\alpha') := \,
    &
    \mathbf{1}_{\tau<\tau'}\,
    \frac{e^{-\frac{(\alpha'-\alpha)^2}{2(\tau'-\tau)}}}{\sqrt{2\pi(\tau'-\tau)}}
    \\
    & 
    \qquad
    +
    \left(\int_{e^{\mathrm{i} \pi/4} \infty}^{e^{-\mathrm{i} \pi/4} \infty} + \int_{e^{-\mathrm{i} 3\pi/4} \infty}^{e^{\mathrm{i} 3\pi/4} \infty} \right) \frac{\dd\xi}{2\pi\mathrm{i}} \int_{-\mathrm{i} \infty}^{+\mathrm{i} \infty} \frac{\dd\eta}{2\pi\mathrm{i}} \,\frac{1}{\xi-\eta} \,e^{ \frac{\xi^4-\eta^4}{24} + \frac{\tau'\xi^2 - \tau\eta^2}{2} - \alpha'\xi + \alpha\eta}.
    \end{split}
\end{equation*}

    The differences are the integration contours for the $\xi$ variable and the cross terms, the one of our kernel being $\frac{2\eta}{\eta^2-\xi^2}$. One can draw a parallel with the case of the Airy kernel, where the presence of a similar cross term transforms the $\mathrm{Airy}_2$ kernel into the $\mathrm{Airy}_{2\to 1}$ kernel.
\end{rem}

\begin{rem}
    In \cite{BK-2010}, a variant of the Pearcey kernel, dubbed the \emph{symmetric Pearcey kernel}, was introduced to describe  point processes produced by non-intersecting random walks conditioned to not cross a wall \cite{KMFW-2011,Kua-2013}. Despite the similarity with the process performed by the last row of the partitions $\la^{(j)}$ from our Berele insertion process, the kernels from \cite{BK-2010} and \eqref{eq:symplectic Percey} appear to be different. We hope to offer clarifications concerning relations between these processes and correlation kernels in future works.
\end{rem}

An interesting observation is that the limit shape in the (not probabilistic) case of $x=9/8$, possesses a cubic singularity near the point $y=-1$ and this can be confirmed by using \eqref{eq:limiting density}. First, write
\begin{equation*}
    h_{\frac{9}{8},-1+\varepsilon}'(z) = \frac{(z+1)^3}{4 (z-2) (z-1) z (2 z-1)} - \frac{\varepsilon }{z}.
\end{equation*}
Then, for $\varepsilon>0$ small, we can express the 
root $z_+$ of $h'_{\frac{9}{8},-1+\varepsilon}(z)$ as
\begin{equation*}
    z_+ = -1+(3^{2/3}+ 3^{7/6} \mathrm{i}) \varepsilon^{1/3} +o(\varepsilon^{1/3}),
\end{equation*}
which implies that
\begin{equation*}
    \frac{\arg(z_+)}{\pi} = 1-\frac{3^{7/6}}{\pi} \varepsilon^{1/3} + o(\varepsilon^{1/3}).
\end{equation*}
Cubic cusps of equilibrium measures, or limit shapes, are known to be related to the Pearcey process \cite{BH-1998a,BH-1998b}.
This observation offers a further invitation to a more rigorous study of the process $\mathcal{L}(\lavec)$ around the macroscopic region $(x,y)=(9/8,-1)$, which we hope to cover in future works.

\bibliographystyle{plain}

\bigskip

Cesar Cuenca:

\smallskip

Department of Mathematics, The Ohio State University, Columbus, OH 43210, USA.

\smallskip

Email address: cesar.a.cuenk@gmail.com

\bigskip

Matteo Mucciconi:

\smallskip

Department of Mathematics, National University of Singapore,
 S17, 10 Lower Kent Ridge Road, 119076, Singapore

\smallskip

Email address: matteomucciconi@gmail.com

\end{document}